%% file: quantum-correlations-fixed.tex
\documentclass[10pt,runningheads]{article}
\input{format.tex}

\begin{document}

\title{\textbf{Quasi-polynomial time algorithms for free quantum games in bounded dimension}}

\author[1]{Hyejung H.~Jee}
\author[1,2]{Carlo Sparaciari}
\author[3]{Omar Fawzi}
\author[1]{Mario Berta}
\affil[1]{\small{Department of Computing, Imperial College London, London SW7 2AZ, UK}}
\affil[2]{\small{Department of Physics and Astronomy, University College London, London WC1E 6BT, UK}}
\affil[3]{\small{Laboratoire de l’Informatique du Parall\'elisme, ENS de Lyon, France}}

\date{}

\maketitle


\begin{abstract}
We give a converging semidefinite programming hierarchy of outer approximations for the set of quantum correlations of fixed dimension and derive analytical bounds on the convergence speed of the hierarchy. In particular, we give a semidefinite program of size $\exp(\mathcal{O}\big(T^{12}(\log^2(AT)+\log(Q)\log(AT))/\epsilon^2\big))$ to compute additive $\epsilon$-approximations on the values of two-player free games with $T\times T$-dimensional quantum assistance, where $A$ and $Q$ denote the numbers of answers and questions of the game, respectively. For fixed dimension $T$, this scales polynomially in $Q$ and quasi-polynomially in $A$, thereby improving on previously known approximation algorithms for which worst-case run-time guarantees are at best exponential in $Q$ and $A$. For the proof, we make a connection to the quantum separability problem and employ improved multipartite quantum de Finetti theorems with linear constraints.
We also derive an informationally complete measurement which minimises the loss in distinguishability relative to the quantum side information\,---\,which may be of independent interest.
\end{abstract}



\section{Introduction}

Thanks to the celebrated discovery of John Bell \cite{bell1964einstein}, it is well-known that quantum correlations can be used as a resource to overcome locality constraints. This represents one of the earliest examples of advantages provided by quantum correlations over classical ones and led to the development of a vast number of quantum information processing tasks that make use of quantum correlations. A prominent example is given by the field of device-independent quantum information processing, where the violation of certain locality constraints is used to certify the underlying state describing the system\,---\,even when we have no a priori knowledge about the devices used (see, e.g., \cite{acin2007device}). 

A common way of measuring the strength of correlations is using a two-player game $G$ as illustrated in Figure~\ref{fig:TwoPlayerGame}.
The classical value $\omega_{C}(G)$ of $G$ is the maximum winning probability that can be achieved using shared randomness between the players, and the quantum value $\omega_{Q}(G)$ is the maximum winning probability that can be achieved using arbitrary quantum states shared between the players. Given the description of a two-player game $G$, it is in general hard to compute $\omega_{C}(G)$ and $\omega_{Q}(G)$. 
In fact, the celebrated PCP theorem~\cite{arora1998proof,arora1998probabilistic} shows that approximating $\omega_{C}(G)$ within some constant multiplicative factor is NP-hard, even when the number of possible answers is restricted to a constant. Computing $\omega_{Q}(G)$ turns out to be much more difficult in general: culminating a long series of works starting from~\cite{cleve2004consequences,kempe2011entangled}, it was recently shown that approximating $\omega_{Q}(G)$ is not possible for an algorithm running in finite time~\cite{ji2020mip}. Despite this general undecidability result, there are some special classes of games for which approximations of $\omega_{Q}(G)$ can be computed in polynomial time (for example XOR games~\cite{cleve2004consequences} or unique games~\cite{kempe2010unique}), and the Navascu\'es-Pironio-Ac\'in (NPA) hierarchy \cite{navascues2008convergent,pironio2010convergent} provides general semidefinite programming (SDP) upper bounds on $\omega_{Q}(G)$, which are widely used in practice and give approximately tight bounds for many specific games of interest. For the classical value $\omega_{C}(G)$, if we restrict ourselves to free games, i.e., games where the questions of the players are chosen independently of each other, there is a quasi-polynomial time algorithm that can approximate $\omega_{C}(G)$ within any constant $\epsilon > 0$ additive error~\cite{aaronson2014multiple,brandao2017quantum}.

\begin{figure*}[t!]
\hspace{3.2cm}
\includegraphics[width=11cm]{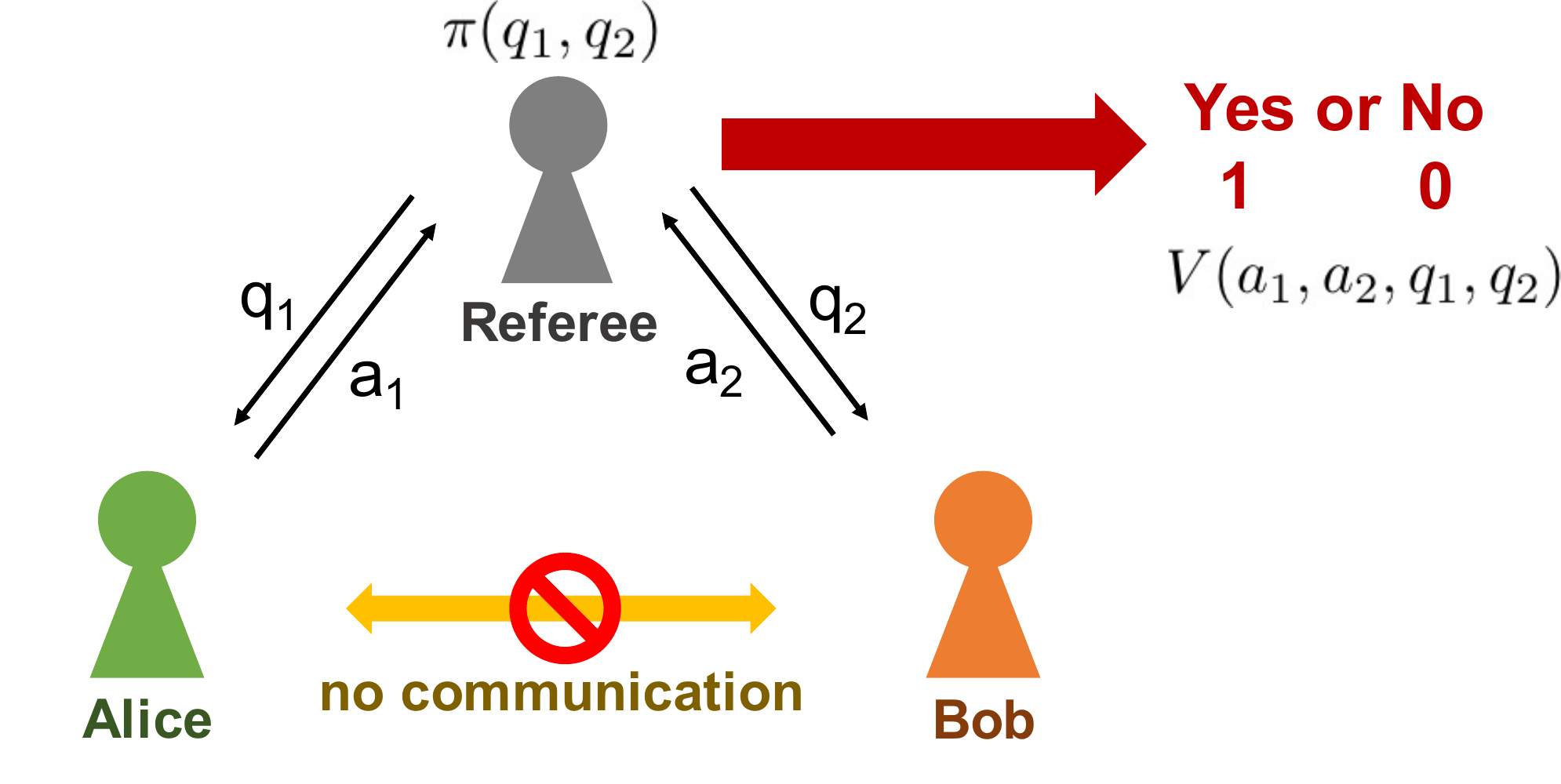}
\caption
{Two-player games. The referee gives Alice and Bob questions $q_1\in Q_1$ and $q_2\in Q_2$, and then Alice and Bob give answers $a_1\in A_1$ and $a_2\in A_2$ back to the referee depending on the questions they received. The referee decides whether Alice and Bob win or lose according to the rule function $V: A_1 \times A_2 \times Q_1 \times Q_2 \rightarrow \{0,1\}$, where $0$ denotes losing the game, and $1$ denotes winning the game. Alice and Bob cannot communicate with each other during the game, but they can agree on a strategy beforehand. We are interested in determining the value of the game, i.e., the maximum achievable winning probability, for different classes of strategies. Of particular interest is the classical value $\omega_{C}$ of the game, corresponding to the classical strategies where Alice and Bob only share randomness, and the quantum value $\omega_{Q}$ of the game, corresponding to general quantum strategies where Alice and Bob are allowed to share an arbitrary quantum state. The focus of this paper is the study of the dimension-bounded quantum value $\omega_{Q(T)}$, corresponding to quantum strategies where Alice and Bob are allowed to share a quantum state of restricted local dimension $T$. For the introduction, we assume that $|Q_1| = |Q_2| = Q$ and $|A_1| = |A_2| = A$. 
}
\label{fig:TwoPlayerGame}
\end{figure*}

In this paper, we study the dimension-bounded quantum value $\omega_{Q(T)}(G)$, which is the maximum winning probability that can be achieved using quantum states of dimension $T \times T$ shared between the players. It is simple to see that $\omega_{Q(1)}(G) = \omega_{C}(G)$ and $\omega_{Q}(G) = \sup_{T \geq 1} \omega_{Q(T)}(G)$. Determining the value of $\omega_{Q(T)}(G)$ for some fixed $T$ is particularly important for establishing \emph{dimension witnesses} of the underlying system (see, e.g., \cite{gallego2010device}).
Hierarchies of SDP upper bounds on $\omega_{Q(T)}(G)$ have been constructed in \cite{navascues2014characterization} where the authors exploit a connection to the {\it quantum separability problem}. Also, in \cite{navascues2015bounding,navascues2015characterizing}, the authors employ a moment matrix technique similar to the NPA hierarchy to derive SDP relaxations with improved practical performance compared to \cite{navascues2014characterization}. However, for these works the running time is either not analytically quantified or is at best exponential in the number of questions $Q$ and the number of answers $A$ of the game $G$.


In this work, we construct algorithms for computing $\omega_{Q(T)}(G)$ whose running time has a significantly improved dependence on $A$ and $Q$. More specifically, we construct a new hierarchy of SDP upper bounds on $\omega_{Q(T)}(G)$ and 
derive analytical bounds on the convergence speed. This gives bounds on the computational complexity of calculating $\omega_{Q(T)}(G)$ as a function of the size of the game $G$. For the case of two-player free games, we provide a semidefinite program of size
\begin{align}
\exp\left(\mathcal{O}\left(\frac{T^{12}}{\epsilon^2}\log(AT)(\log(Q)+\log(AT))\right)\right)
\end{align}
for computing additive $\epsilon$-approximations of $\omega_{Q(T)}(G)$
where $A$ and $Q$ denote the numbers of answers and questions of $G$, respectively (see Figure~\ref{fig:TwoPlayerGame}). Note that the dependence is quasi-polynomial in $A$ and polynomial in $Q$. As previously mentioned, the term \emph{free game} means that the choice of the questions for the two players is not correlated. If we set $T=1$, our result recovers the quasi-polynomial time approximation scheme for computing the classical value of two-player free games \cite{brandao2017quantum,aaronson2014multiple} (see Section~\ref{sec:discussion} for details). In the classical setting, there is also a matching hardness result assuming the Exponential Time Hypothesis~\cite{aaronson2014multiple}.
We also give a program for general games, which is still quasi-polynomial in $A$ but exponential in $Q$. Our approach for deriving the SDP relaxations is to make a connection to a variant of the quantum separability problem where the states involved are subject to some linear constraints. This variant of the quantum separability problem was first studied in~\cite{berta2018semidefinite}, and in our work we obtain improved multipartite quantum de Finetti theorems with linear constraints. This allows us to quantify the convergence speed of our hierarchy of SDP relaxations.


The remainder of this manuscript is structured as follows. In Section~\ref{sec:overview}, we give an overview of our results and techniques. In Section~\ref{sec:SDP}, we derive the constraints for characterising fixed-dimensional quantum correlations and state a new hierarchy of SDP relaxations for $\omega_{Q(T)}(G)$ of two-player free games.
In Section~\ref{sec:convergence}, we analytically prove the convergence rate of the derived hierarchy exploiting tripartite quantum de Finetti theorems with linear constraints. Then, in Section~\ref{sec:NPA}, we discuss how we can combine our derived constraints with the existing dimension agnostic NPA hierarchy. Lastly, in Section~\ref{sec:discussion} we conclude with several remarks on our results. Various technical arguments are deferred to appendices.


\section{Overview}\label{sec:overview}

\paragraph{Two-player games.} There exist various mathematical formulations of the correlations between distant parties, but in this paper we will use the non-local game scenario. Specifically, we will consider \emph{two-player games} where two distant parties are considered. In this formulation, we consider the correlation between two parties as a resource to win specific games. 

Let us assume that there are two spatially separated players, which we will conventionally call Alice and Bob, and a referee. When the game starts, the referee chooses questions $q_1$ and $q_2$ from the question set $Q_1 \times Q_2$ according to a given probability distribution $\pi(q_1, q_2)$ and sends them to Alice and Bob, respectively. Then, Alice and Bob must provide answers for their questions, $a_1$ and $a_2$, to the referee (see Figure~\ref{fig:TwoPlayerGame}). The correct answers are determined by a given rule function
\begin{align}
V: A_1 \times A_2 \times Q_1 \times Q_2 \rightarrow \{0,1\};
\end{align}
$0$ means the answers are wrong, and $1$ means the answers are correct. A specific two-player game $G$ can be represented by a question probability distribution $\pi(q_1,q_2)$ and a rule fuction $V(a_1,a_2,q_1,q_2)$, and hereafter we will denote a two-player game $G$ as $(V,\pi)$. Alice and Bob cannot communicate with each other during the game, but they can agree on a strategy beforehand, which can be described by a conditional probability distribution $p(a_1,a_2|q_1,q_2)$. Depending on the resources the players can access, this conditional probability describes different kinds of correlations. When no resources are shared, each player's answer only depends on the question they receive, so that the distribution takes the form
\begin{align}
p(a_1,a_2|q_1,q_2) = e(a_1|q_1)d(a_2|q_2),
\end{align}
where $e(a_1|q_1)$ and $d(a_2|q_2)$ are conditional probability distributions representing the strategy of Alice and Bob, respectively. When quantum resources are allowed, the distribution takes the more general form
\begin{align}
p(a_1,a_2|q_1,q_2) = \tr\left[ \rho_{T\hat{T}} \left( E_T(a_1|q_1)\otimes D_{\hat{T}}(a_2|q_2) \right) \right],
\end{align}
where $\rho_{T\hat{T}}$ is an entangled quantum state shared between Alice and Bob, and for fixed questions $q_1,q_2$ to Alice and Bob, $\{ E_T(a_1|q_1) \}_{a_1}$ and $\{ D_{\hat{T}}(a_2|q_2) \}_{a_2}$ are positive-operator valued measurements (POVMs) performed by Alice and Bob, respectively. Then, the difference in performance between classical and quantum correlations can be quantified in terms of a non-zero gap between the maximum winning probabilities achieved using the two different resources. For a given game $G=(V,\pi)$, the classical value, the classical maximum winning probability, is\footnote{We do not include shared randomness in \eqref{eq:CWinP} as it is not helpful in this scenario \cite{bell1964einstein}.}
\begin{align}\label{eq:CWinP}
\omega_C(V,\pi) := \max_{(e,d)} \sum_{q_1,q_2} \pi(q_1, q_2) \sum_{a_1, a_2} V(a_1,a_2,q_1,q_2) e(a_1|q_1) d(a_2|q_2),
\end{align}
where $\sum_{a_1}e(a_1|q_1)=1 \;\forall q_1$ and $\sum_{a_2}d(a_2|q_2)=1 \;\forall q_2$, while the quantum value is given by
\begin{align}\label{eq:QWinP}
\omega_Q(V,\pi) := \sup_{(E\otimes D,\rho)\;\text{on}\;\mathcal{H}_{T\hat{T}}} \sum_{q_1,q_2} \pi(q_1,q_2) \sum_{a_1,a_2} V(a_1,a_2,q_1,q_2) \tr\left[ \rho_{T\hat{T}} \left( E_T(a_1|q_1)\otimes D_{\hat{T}}(a_2|q_2) \right) \right],
\end{align}
where $\tr[\rho_{T\hat{T}}]=1$, $\sum_{a_1} E_T(a_1|q_1)=\id_T\;\forall q_1$, and $\sum_{a_2} D_{\hat{T}}(a_2|q_2)=\id_{\hat{T}}\;\forall q_2$. Now, the supremum is taken over all possible measurements and states as well as the Hilbert spaces $\mathcal{H}_{T\hat{T}}$. Then, denoting the winning probability for a given game $(V,\pi)$ with classical strategies $e,d$ as $p_{\text{win}}(V,\pi,e,d)$, {\it Bell inequalities} for the game can be described as $p_{\text{win}}(V,\pi,e,d)\leq\omega_C(V,\pi)$. Violations of such upper bounds with quantum resources are called violations of Bell inequalities.
For example, the CHSH inequality \cite{clauser1969proposed} can be derived from the CHSH game, which can be specified with a uniform distribution as $\pi(q_1,q_2)$ and $V_{\text{CHSH}}(a_1,a_2,q_1,q_2)$ as the rule function, where
\begin{align}
\text{$V_{\text{CHSH}}(a_1,a_2,q_1,q_2) = 1$ if $q_1\cdot q_2=a_1\oplus a_2$, and $V_{\text{CHSH}}(a_1,a_2,q_1,q_2) = 0$ otherwise.}
\end{align}

We can also define the maximum winning probability for a given two-player game $(V,\pi)$ with quantum assistance of fixed dimension $\mathbb{C}^T\otimes\mathbb{C}^T$, as
\begin{align}\label{eq:QWinP_T}
\omega_{Q(T)}(V,\pi) := \max_{(E\otimes D,\,\rho)\,\text{on}\;\mathbb{C}^T\otimes\mathbb{C}^T} \sum_{q_1,q_2} \pi(q_1,q_2) \sum_{a_1,a_2} V(a_1,a_2,q_1,q_2) \tr \left[ \rho_{T\hat{T}} \left( E_T(a_1|q_1)\otimes D_{\hat{T}}(a_2|q_2) \right) \right].
\end{align}
In this paper, we derive SDP upper bounds for calculating $\omega_{Q(T)}(V,\pi)$.

If not stated otherwise, we assume that the choice of questions for Alice and Bob are not correlated, i.e.~$\pi(q_1,q_2) = \pi_1(q_1)\pi_2(q_2)$, which corresponds to \emph{free games}.\footnote{We treat general games in Appendix \ref{app:GeneralGames}.} Henceforth, we denote $\mathcal{H}_A^{\otimes n}$ as $A_1^n=A^n$, and $\operatorname{dim}(\hil_A)$ as $|A|$. For simplicity, we also assume that the number of answers and questions for Alice and Bob are the same, denoted by $|A_1|=|A_2|=|A|$, and $|Q_1|=|Q_2|=|Q|$, as well as that the shared entanglement between Alice and Bob has fixed dimension $\dim\left(\mathcal{H}_{T\hat{T}}\right)=|T|^2$.


\paragraph{Connection with quantum separability.} Quantum separability problems are a specific form of optimisation problems, for which the optimisation is over the set of all separable quantum states (see, e.g., \cite{gharibian2008strong}). We find that the dimension-bounded quantm value $\omega_{Q(T)}$ of two-player free games~\eqref{eq:QWinP_T}, can be converted to an instance of the tripartite quantum separability problem with additional linear constraints, which is a variant of the quantum separability problem.

\begin{lemma}\label{lem:QWinP_TinSP}
For a two-player free game with $V(a_1,a_2,q_1,q_2)$, $\pi(q_1,q_2)=\pi_1(q_1)\pi_2(q_2)$, and $|T|^2$-dimensional assistance, we have
\begin{equation}
\begin{split}\label{eq:QWinP_TSepProb}
& \quad \; \omega_{Q(T)}(V, \pi) = |T|^2\cdot\max_{(E,D,\rho)} \; \tr \left[ \left( V_{A_1A_2Q_1Q_2}\otimes \Phi_{T\hat{T}|S\hat{S}}\right)\left( E_{A_1Q_1T}\otimes D_{A_2Q_2\hat{T}}\otimes\rho_{S\hat{S}} \right) \right]\\
\text{s.t.} \quad &\rho_{S\hat{S}} \geq 0\, , \quad \quad \Tr [\rho_{S\hat{S}}] = 1 \\
& E_{A_1Q_1T} = \sum_{a_1, q_1} \pi_1(q_1) \ket{a_1q_1}\!\!\bra{a_1q_1}_{A_1Q_1} \otimes \frac{E_T(a_1|q_1)}{|T|} \geq 0\, , \quad E_{Q_1T} = \sum_{q_1} \pi_1(q_1)\ket{q_1}\!\!\bra{q_1}_{Q_1} \otimes \frac{\mathds{1}_T}{|T|}\\
& D_{A_2Q_2\hat{T}} = \sum_{a_2,q_2} \pi_2(q_2) \ket{a_2q_2}\!\!\bra{a_2q_2}_{A_2Q_2} \otimes \frac{D_{\hat{T}}(a_2|q_2)}{|T|} \geq 0\, , \quad D_{Q_2\hat{T}} = \sum_{q_2} \pi_2(q_2) \ket{q_2}\!\!\bra{q_2}_{Q_2} \otimes \frac{\mathds{1}_{\hat{T}}}{|T|},
\end{split}
\end{equation}
where $\Phi_{T\hat{T}|S\hat{S}} = \ket{\Phi}\bra{\Phi}_{T\hat{T}|S\hat{S}}\,$, with the (non-normalised) maximally-entangled state $\ket{\Phi}_{T\hat{T}|S\hat{S}}= \sum_i \, \ket{i}_{T\hat{T}} \ket{i}_{S\hat{S}}$, and
$V_{A_1A_2Q_1Q_2}$ is a diagonal matrix whose entries are given by the rule function $V(a_1,a_2,q_1,q_2)$.
\end{lemma}

The proof can be found in Section~\ref{subsec:SeparabilityProblem}. Now $\omega_{Q(T)}(V,\pi)$ is expressed as an optimisation problem over product states with some additional linear constraints on the optimisation variable. Since product states are extreme points in the set of separable states, we can equivalently think of the above as a maximisation over mixtures of product states which satisfy the stated linear constraints.


\paragraph{Hierarchy of semidefinite programming relaxations.} 
Quantum separability conditions have well-known relaxations called the Doherty-Parrilo-Spedalieri (DPS) hierarchy \cite{doherty2004complete}.
In the DPS hierarchy, the notion of extendible states is employed; A state $\rho_{AB}$ is called an $n$-extendible state if there exists an extension $\rho_{AB^n}$ such that $\tr_{B^{n-1}}\left[ \rho_{AB^n} \right] = \rho_{AB}$, and it is invariant under any permutation of the $B$-registers with respect to $A$. This class of states has two main advantages: Firstly, for any $n$, the set of $n$-extendible states can be efficiently identified by a semidefinite program, which makes them a simple set for optimisation. Secondly, they are good approximations of separable states as a state is $n$-extendible for every $n\geq2$ if and only if it is separable \cite{fannes1988symmetric,raggio1989quantum}. This is related to a concept known as the monogamy of entanglement \cite{terhal2004entanglement}. As the set of $(n+1)$-extendible states is a subset of the set of $n$-extendible states, the set of $n$-extendible states becomes a better and better approximation for the set of separable states as $n$ increases. Since (\ref{eq:QWinP_TSepProb}) includes tripartite separability conditions, we need to consider the multipartite generalisation of the bipartite extendible states; $(n_1, n_2)$-extendible states $\rho_{ABC}$ which have an extension $\rho_{AB^{n_1}C^{n_2}}$. We refer to Section~\ref{subsubsec:ExtStates} for a more detailed explanation. As in the bipartite case, the set of $(n_1, n_2)$-extendible states converges to the separable set when $n_1\rightarrow\infty$ and $n_2\rightarrow\infty$ \cite{doherty2005detecting}, and thus they are an approximation of tripartite separable states.

In the standard quantum separability problem, by simply replacing the optimisation variable with an $(n,n)$-extendible state with respect to the appropriate partitions, we can derive SDP relaxations for the problem. However, due to the additional linear constraints in \eqref{eq:QWinP_TSepProb}, in addition to the $n$-extendibility conditions enforced by the DPS hierarchy, we need to add the linear constraints to the semidefinite program to guarantee that they are satisfied. The resulting hierarchy of SDP relaxations is given by
\begin{align}\label{eq:SDP}
&\qquad\sdp_n(V,\pi,T) := |T|^2 \max_{\rho} \tr\left[ \left( V_{A_1A_2Q_1Q_2}\otimes \Phi_{T\hat{T}|S\hat{S}} \right) \rho_{(A_1Q_1T)(A_2Q_2\hat{T})(S\hat{S})} \right]\\
\text{s.t.} \,\,\, & \rho_{(A_1Q_1T)(A_2Q_2\hat{T})^{n}(S\hat{S})^{n}} \geq 0\, , \quad \tr \left[ \rho_{(A_1Q_1T)(A_2Q_2\hat{T})^{n}(S\hat{S})^{n}} \right] = 1\\
& \rho_{(A_1Q_1T)(A_2Q_2\hat{T})^n(S\hat{S})^n} \text{ perm.~inv.~on } (A_2Q_2\hat{T})^n \text{ wrt } (A_1Q_1T)(S\hat{S})^n \label{eq:SDPExtConstraints0}\\
& \rho_{(A_1Q_1T)(A_2Q_2\hat{T})^{n}(S\hat{S})^n} \text{ perm.~inv.~on } (S\hat{S})^n \text{ wrt } (A_1Q_1T)(A_2Q_2\hat{T})^n \label{eq:SDPExtConstraints} \\
&\tr_{A_1}[\rho_{(A_1Q_1T)(A_2Q_2\hat{T})^{n}(S\hat{S})^{n}}] = \left( \sum_{q_1}\pi_1(q_1)\ket{q_1}\!\!\bra{q_1}_{Q_1}\otimes\frac{\mathds{1}_T}{|T|} \right) \otimes \tr_{A_1Q_1T}\left[\rho_{(A_1Q_1T)(A_2Q_2\hat{T})^{n}(S\hat{S})^{n}}\right]\label{eq:SDPadditionalLC1} \\
&\tr_{A_2}[\rho_{(A_1Q_1T)(A_2Q_2\hat{T})^{n}(S\hat{S})^{n}}] = \left( \sum_{q_2}\pi_2(q_2)\ket{q_2}\!\!\bra{q_2}_{Q_2}\otimes\frac{\mathds{1}_{\hat{T}}}{|T|} \right) \otimes \tr_{A_2Q_2\hat{T}}\left[\rho_{(A_1Q_1T)(A_2Q_2\hat{T})^n(S\hat{S})^{n}}\right]\label{eq:SDPadditionalLC2}\\
&\rho^{T_{A_1Q_1T}}_{(A_1Q_1T)(A_2Q_2\hat{T})^{n}(S\hat{S})^{n}} \geq 0\, ,\quad \rho^{T_{(A_2Q_2\hat{T})^{n}}}_{(A_1Q_1T)(A_2Q_2\hat{T})^{n}(S\hat{S})^{n}} \geq 0\, ,\quad \rho^{T_{(S\hat{S})^{n}}}_{(A_1Q_1T)(A_2Q_2\hat{T})^{n}(S\hat{S})^{n}} \geq 0\label{eq:SDPConstraints_PPT},\;\ldots\;,
\end{align}
where the last line \eqref{eq:SDPConstraints_PPT} contains all positive partial transpose (PPT) conditions with respect to all the cuts
\begin{align}
A_1Q_1T:A^1_2Q^1_2 \hat{T}^1: \cdots : A^n_2Q^n_2 \hat{T}^n :S^1\hat{S}^{1} : \cdots : S^n \hat{S}^n.
\end{align}
We emphasise that apart from the extendibility conditions \eqref{eq:SDPExtConstraints0} -- \eqref{eq:SDPExtConstraints}, compared to the original DPS hierarchy our SDP relaxations have the additional constraints \eqref{eq:SDPadditionalLC1} -- \eqref{eq:SDPadditionalLC2}. These come from the corresponding additional constraints in the separability problem \eqref{eq:QWinP_TSepProb}. As the set of $n$-extendible states always include the set of separable states, these semidefinite programs give a series of upper bounds for $\omega_{Q(T)}(V,\pi)$. Moreover, we discuss in Section~\ref{sec:NPA} how our relaxations can be combined with the NPA hierarchy to form the SDP relaxations 
\begin{align}\label{eq:SDP_NPA}
\text{$\sdp_n^{\npa}(V,\pi,T)$: $\sdp_n(V,\pi,T)$ with the additional constraints $\Gamma_n\big(\rho_{(A_1Q_1T)(A_2Q_2\hat{T})^{n}(S\hat{S})^{n}}\big)\geq0$,}
\end{align}
where $\Gamma_n(\rho)$ denotes the $n$-th level NPA matrix defined in Section~\ref{subsec:NPA+SDP}.


\paragraph{Multipartite quantum de Finetti theorem with linear constraints.} Quantum de Finetti theorems provide a quantitative bound on how close $n$-extendible states are to the set of separable states in trace distance as a function of $n$ and the dimension. Thus, we can exploit this to upper bound the gap between our SDP relaxations and the actual solution $\omega_{Q(T)}$. However, the two additional linear constraints \eqref{eq:SDPadditionalLC1} -- \eqref{eq:SDPadditionalLC2} make this approach more difficult. Without these linear constraints, the feasible states in our SDP relaxations would exactly be separable states, and hence we could apply the standard quantum de Finetti theorem to prove the convergence of the relaxations. However, in our case, the feasible set is no longer the set of separable states but is the convex hull of product states satisfying the additional linear constraints. As a consequence, we need adapted quantum de Finetti theorems.

In the main body of the paper, we give improved multipartite quantum de Finetti theorems with linear constraints. It is an extension of the bipartite quantum de Finetti theorem with linear constraints from \cite{berta2018semidefinite}, where the authors follow the information-theoretic proof technique introduced in \cite{brandao2016product,brandao2017quantum}. We would like to emphasise that the existence of the additional linear constraints and the use of this adapted quantum de Finetti theorems are crucial to obtain the improved complexity bounds on computation of $\omega_{Q(T)}(V,\pi)$. Here, we state a tripartite version of the theorems, which suits our setting.

\begin{theorem}\label{thm:MulQdeFTwithLC}
Let $\rho_{AB^{n_1}C^{n_2}}$ be a quantum state which is invariant under permutations on $B^{n_1}$ with respect to $AC^{n_2}$ and on $C^{n_2}$ with respect to $AB^{n_1}$, satisfying for linear maps $\mathcal{E}_{A\rightarrow\tilde{A}}$, $\Lambda_{B\rightarrow \tilde{B}}$, and $\Gamma_{C\rightarrow \tilde{C}}$ and operators $\mathcal{X}_{\tilde{A}}$, $\mathcal{Y}_{\tilde{B}}$, and $\mathcal{Z}_{\tilde{C}}$ that
\begin{align}
\left(\mathcal{E}_{A\rightarrow\tilde{A}}\otimes\mathcal{I}_{B^{n_1}C^{n_2}}\right) \left( \rho_{AB^{n_1}C^{n_2}} \right) = \mathcal{X}_{\tilde{A}} \otimes\rho_{B^{n_1}C^{n_2}} \quad & \quad \text{linear constraint on A} \\
\left(\Lambda_{B\rightarrow\tilde{B}}\otimes\mathcal{I}_{B^{n_1-1}C^{n_2}}\right) \left( \rho_{B^{n_1}C^{n_2}} \right) = \mathcal{Y}_{\tilde{B}} \otimes \rho_{B^{n_1-1}C^{n_2}} \quad & \quad \text{linear constraint on $B$}\\
\left(\mathcal{I}_{B^{n_1}C^{n_2-1}}\otimes\Gamma_{C\rightarrow\tilde{C}}\right) \left( \rho_{B^{n_1}C^{n_2}} \right) = \mathcal{Z}_{\tilde{C}} \otimes \rho_{B^{n_1}C^{n_2-1}} \quad & \quad \text{linear constraint on $C$}.
\end{align}
Then, there exist a probability distribution $\{p_i\}_{i\in I}$ and sets of quantum states $\{\sigma_A^i\}_{i\in I}$, $\{\omega_{B}^i\}_{i\in I}$ and $\{\tau_{C}^i\}_{i\in I}$ such that we have that
\begin{equation}\label{eq:MulQdeFTwithLCbound}
\begin{split}
&\Big\|\rho_{ABC} - \sum_{i\in I} p_i \sigma_A^i\otimes\omega_{B}^i\otimes\tau_{C}^i\Big\|_1
\leq O\Bigg( |BC|\Bigg( \sqrt{\frac{\log |A| + \log |B|}{n_2} + \frac{\log |A|}{n_1}} \Bigg)\Bigg)
\end{split}
\end{equation}
\begin{align}
\mathcal{E}_{A\rightarrow\tilde{A}} \left( \sigma_A^i \right) = \mathcal{X}_{\tilde{A}}, \quad 
\Lambda_{B\rightarrow\tilde{B}} \left(\omega_{B}^i\right) = \mathcal{Y}_{\tilde{B}}, \quad \Gamma_{C\rightarrow\tilde{C}}\left(\tau_{C}^i\right) = \mathcal{Z}_{\tilde{C}} \qquad \forall i\in I.
\end{align}
\end{theorem}

The proof is given in Appendix~\ref{app:MulQdeFTwithLC} and contains the following result that might be of independent interest.


\paragraph{Distortion relative to quantum side information.}

As a fundamental step in the proof of Theorem \ref{thm:MulQdeFTwithLC}, we derive an informationally complete measurement $\mathcal{M}_B$ that has optimal loss of distinguishability relative to quantum side information. That is, we have for all quantum states $\rho_{AB},\sigma_{AB}$ that
\begin{align}
\norm{\left(\mathcal{I}_A\otimes\mathcal{M}_B\right) \left(\rho_{AB}-\sigma_{AB}\right)}_1\leq\norm{\rho_{AB}-\sigma_{AB}}_1\leq2|B|\cdot\norm{\left(\mathcal{I}_A\otimes\mathcal{M}_B\right) \left(\rho_{AB}-\sigma_{AB}\right)}_1.
\end{align}
The proof is given in Appendix \ref{app:distortion_sideI} and improves on factor $\sqrt{18}|B|^{3/2}$ given in \cite[Equation (68)]{brandao2017quantum}. Together with the example states given in \cite{matthews2009distinguishability}, this establishes the dimension dependence for minimal distortion relative to quantum side information to be $\Omega(|B|)$.


\paragraph{Convergence rate of the hierarchy.} In our SDP relaxations $\sdp_n(V,\pi,T)$, the linear constraints are partial trace constraints on subsystems. A version of Theorem~\ref{thm:MulQdeFTwithLC} with partial trace constraints\,---\,that we state and prove in Section~\ref{subsec:TriQdeF}\,---\,then immediately gives a bound on the convergence speed of $\sdp_{n}(V,\pi,T)$.

\begin{theorem}\label{thm:Convergence}
Let $\sdp_{n}(V,\pi,T)$ be the $n$-th level SDP relaxation for the two-player free game with rule matrix $V$, probability distribution $\pi(q_1,q_2)=\pi_1(q_1)\pi_2(q_2)$, and quantum assistance of dimension $|T|^2$. Then, we have
\begin{align}
0 \leq \sdp_{n}(V,\pi,T) - \omega_{Q(T)}(V, \pi) \leq O\left( |T|^6\sqrt{\frac{\log|T||A|}{n}} \right).
\end{align}
Hence, we have $\omega_{Q(T)} (V,\pi) = \lim_{n\rightarrow\infty} \sdp_n(V,\pi,T)$.
\end{theorem}

Theorem~\ref{thm:Convergence} implies that we can estimate $\omega_{Q(T)}(V,\pi)$ within a constant additive error $\epsilon>0$ with a semidefinite program of size
\begin{align}\label{eq:SDPsize_freegames}
\exp\Big( \mathcal{O} \Big(\frac{|T|^{12}}{\epsilon^2}\log|AT| \left(\log|Q|+\log|AT| \right) \Big) \Big),
\end{align}
which is quasi-polynomial in the number of answers $|A|$ and polynomial in the number of questions $|Q|$. When $|T|=1$, the SDP relaxations become linear programs, and this result recovers the quasi-polynomial time approximation scheme for computing the classical value $\omega_C(V,\pi)$ for two-player free games --- which has a matcing hardness result \cite{aaronson2014multiple,brandao2017quantum}. In Appendix~\ref{app:GeneralGames}, we analyse the computational complexity for general games, for which we find SDP relaxations exponential in $|Q|$.


\paragraph{Examples and comparison with previous work.} The combined SDP relaxations $\sdp^{\npa}_n(V,\pi,T)$ in \eqref{eq:SDP_NPA} promise to give tighter bounds for dimension constraint settings compared to the dimension agnostic NPA hierarchy. We show in Appendix~\ref{app:sdp11_non_signalling} that
\begin{align}
\sdp_1(V,\pi,T)=\omega_{\ns}(V,\pi),
\end{align}
where the latter denotes the optimal success probability with the assistance of general non-signalling correlations\,---\,the so-called non-signalling value. As the NPA hierarchy already has non-signalling constraints built-in, we need to go to higher level relaxations of our hierarchy to see any improvements over the plain NPA hierarchy. For a second level relaxation $\sdp_{2,1}^{\npa_1}(V,\pi,T)$, which is $\sdp_{2,1}(V,\pi,T)$ with the first-level NPA constraint added, we give in Table~\ref{tab:SpecialRM} a rule matrix $W$ which improves for two-dimensional assistance on the plain first NPA level. That is, we find that our additional constraints for two-dimensional assistance give
\begin{align}
\sdp_{2,1}^{\npa_1}\left(W,\pi,2\right)=0.79888<0.80157=\npa_1(W,\pi),
\end{align}
where the latter is the value given by the first level of the NPA hierarchy. The computations were carried out in Python using the SDP solver MOSEK \cite{mosek17}\footnote{The code is available on GitHub:  \url{https://github.com/carlosparaciari/non_local_games}.}, with at least 5 digits accuracy as given from the value obtained from the primal versus the corresponding dual program.

\begin{table}[ht]
\centering
\begin{tabular}{ccccc}
\hline \hline
& $\mathbf{q_1=0, q_2=0}$ & $\mathbf{q_1=0, q_2=1}$ & $\mathbf{q_1=1, q_2=0}$ & $\mathbf{q_1=1, q_2=1}$\\
\hline
\multirow{4}{4em}{\textbf{Winning answers}} & $a_1=0, a_2=1$ & $a_1=0, a_2=2$ & $a_1=0, a_2=0$ & $a_1=1, a_2=2$ \\
& $a_1=1, a_2=2$ & $a_1=1, a_2=0$ & $a_1=0, a_2=2$ & \\
& $a_1=2, a_2=1$ & & $a_1=1, a_2=2$ & \\
& $a_1=2, a_2=2$ & & $a_1=2, a_2=1$ &\\
\hline \hline
\end{tabular}
\caption{\label{tab:SpecialRM} The rule matrix $W$ with $\sdp_{2,1}^{\npa_1}(W,\pi,2)<\npa_1(W,\pi)$, where $\npa_1$ denotes the value given by the first level of the NPA hierarchy. Here, $\pi$ is uniform, $|A|=3$, and $|Q|=2$. The table only shows the winning answers for each question set, where all the other answers lose.}
\end{table}

We give in Appendix \ref{app:adapted} an adapted hierarchy $\overline{\sdp}_n^{\operatorname{proj}}(V,\pi,T)$ of SDPs valid under the assumption of projective rank-one measurements for the special case of $|A|=|T|=2$. These SDPs have smaller program sizes compared to the original SDPs and are thus advantageous for the implementation of higher levels. In particular, for $|A|=|T|=2$ the optimal measurements are necessarily projective rank-one and we find that for $|Q|$ questions on each side distributed uniformly according to $\pi_{\text{unif}}$,
\begin{align}
\overline{\sdp}_{|Q|}^{\operatorname{proj}}\left(V,\pi_{\text{unif}},2\right)\leq\frac{1}{|Q|^2}\cdot\sdp_{\text{PPT}}(V),
\end{align}
where the latter denotes the PPT type SDP relaxation previously given in \cite[Equation (6)]{navascues2014characterization}. Hence, for such settings, our work can be understood as adding more constraints to the relaxations given in \cite{navascues2014characterization}. In the same spirit, we discuss in Appendix \ref{app:asymptotic} another adapted SDP hierarchy $\sdp_n^{\operatorname{proj}}(V,\pi,2)$ which improves the asymptotic scaling of the hierarchy in \cite{navascues2014characterization} from exponential to polynomial in $|Q|$ for additive $\epsilon$-approximations.

As a numerical example, for the $I_{3322}$ Bell inequality with the (non-binary) rule matrix $V_{3322}$ in the form of \cite{Froissart81,Collins04},\footnote{For general Bell inequalities the rule matrix $V(a_1,q_1,a_2,q_2)$ is not necessarily binary but can take general real values. In particular, this is the case for the $I_{3322}$ inequality.} we can reproduce the dimension witness
\begin{align}
\omega_{Q(2)}\left(V_{3322},\pi_{\text{unif}}\right)\leq\overline{\sdp}_{|Q|}^{\operatorname{proj}}\left(V_{3322},\pi_{\text{unif}},2\right)\leq0.25\quad\Big[=\omega_{Q(2)}\left(V_{3322},\pi_{\text{unif}}\right)<\omega_Q\left(V_{3322},\pi_{\text{unif}}\right)\Big],
\end{align}
where the last two steps are due to \cite{Pal10}. Overall, we note that the added constraints from the projective rank-one measurement assumption are typically very useful for practical performance. Going beyond that, numerical tests for low-level relaxations reveal that it seems challenging to compete with the further methods from \cite{navascues2015bounding,navascues2015characterizing} that were designed with practical purposes in mind but lack analytical bounds on the convergence speed.


\section{Hierarchy of semidefinite programming relaxations}\label{sec:SDP}

\subsection{Connection to quantum separability}\label{subsec:SeparabilityProblem}

In this section, we show how to connect the original formulation for the dimension-bounded quantum value $\omega_{Q(T)}$ of two-player free games, \eqref{eq:QWinP_T}, to the quantum separability problem (Lemma~\ref{lem:QWinP_TinSP}). We need the following slightly modified version of the swap trick.

\begin{lemma}\label{lem:modSwapTrick}
Let $M_{AB}$ be a linear operator on $\mathcal{H}_A \otimes \hil_B$, and $N_A$ be a linear operator on $\mathcal{H}_A$.
Then, it holds that
\begin{align}
\Trt{(N_A \otimes \id_B) M_{AB}} = \Trt{\left(F_{\hat{A}|A} \otimes \id_B\right) \left(N_{\hat{A}} \otimes M_{AB}\right)},
\end{align}
where $F_{\hat{A}|A}$ denotes the swap operator between $\hat{A}$ and $A$.
\end{lemma}

\begin{proof}
By inspection, we have that
\begin{align}
\Trt{\left(F_{\hat{A}|A} \otimes \id_B\right) \left(N_{\hat{A}} \otimes M_{AB}\right)} &= \Trt{\left(F_{\hat{A}|A} \otimes \id_B\right) \left( \sum_{i,j} n_{ij} \ket{i}\bra{j}_{\hat{A}} \otimes \sum_{k,\ell,s,t} m_{(k\ell)(st)} \ket{k}\bra{\ell}_A \otimes \ket{s}\bra{t}_B \right)} \\
&= \Trt{\sum_{i,j,k,\ell,s,t} n_{ij} \, m_{(k\ell)(st)} \ket{k}\bra{j}_{\hat{A}} \otimes \ket{i}\bra{\ell}_A \otimes \ket{s}\bra{t}_B} \\
&= \sum_{i,j,s,t} n_{ij} \, m_{(ji)(st)} = \Trt{(N_A \otimes \id_B) M_{AB}}.
\end{align}
\end{proof}

\begin{proof}[Proof of Lemma~\ref{lem:QWinP_TinSP}]
The starting point is the expression \eqref{eq:QWinP_T}, and for free games, i.e. $\pi(q_1,q_2) = \pi_1(q_1)\pi_2(q_2)$, we can write
\begin{align}
&\qquad \qquad \omega_{Q(T)}(V, \pi) = |T|^2 \max_{E, D, \rho} \; \tr \left[ \left( V_{A_1A_2Q_1Q_2}\otimes \rho_{T\hat{T}}\right)\left( E_{A_1Q_1T}\otimes D_{A_2Q_2\hat{T}} \right) \right] \label{eq:QWinP_TfulQ} \\
\text{s.t.} \quad &\rho_{T\hat{T}} \geq 0\, , \quad \quad \tr [\rho_{T\hat{T}}] = 1 \\
& E_{A_1Q_1T} = \sum_{a_1, q_1} \pi_1(q_1)\ket{a_1q_1}\!\!\bra{a_1q_1}_{A_1Q_1} \otimes \frac{E_T(a_1|q_1)}{|T|} \geq 0\, , \quad E_{Q_1T} = \sum_{q_1} \pi_1(q_1)\ket{q_1}\!\!\bra{q_1}_{Q_1} \otimes \frac{\mathds{1}_T}{|T|}\\
& D_{A_2Q_2\hat{T}} = \sum_{a_2,q_2} \pi_2(q_2)\ket{a_2q_2}\!\!\bra{a_2q_2}_{A_2Q_2} \otimes \frac{D_{\hat{T}}(a_2|q_2)}{|T|} \geq 0\, , \quad D_{Q_2\hat{T}} = \sum_{q_2} \pi_2(q_2) \ket{q_2}\!\!\bra{q_2}_{Q_2} \otimes \frac{\mathds{1}_{\hat{T}}}{|T|},
\end{align}
where we define $V_{A_1A_2Q_1Q_2}:=\sum_{a_1,a_2,q_1,q_2} V(a_1,a_2,q_1,q_2) \ket{a_1,a_2,q_1,q_2}\!\!\bra{a_1,a_2,q_1,q_2}$. Then, using the swap trick (Lemma~\ref{lem:modSwapTrick}) we can rewrite the objective function in \eqref{eq:QWinP_TfulQ} as
\begin{equation}
\begin{split}
\tr&\left[ \left( V_{A_1A_2Q_1Q_2}\otimes \rho_{T\hat{T}} \right)\left( E_{A_1Q_1T}\otimes D_{A_2Q_2\hat{T}} \right) \right]\\
=& \tr\left[ \left( \mathds{1}_{A_1A_2Q_1Q_2}\otimes\rho_{T\hat{T}} \right) \left(\left( V_{A_1A_2Q_1Q_2}\otimes \mathds{1}_{T\hat{T}} \right)\left( E_{A_1Q_1T}\otimes D_{A_2Q_2\hat{T}} \right) \right)\right]\\
=& \tr\left[ \left( \mathds{1}_{A_1A_2Q_1Q_2} \otimes F_{T\hat{T}|\mathbb{T}\hat{\mathbb{T}}} \right) \left( \left( \left( V_{A_1A_2Q_1Q_2}\otimes \mathds{1}_{T\hat{T}} \right) \left( E_{A_1Q_1T}\otimes D_{A_2Q_2\hat{T}} \right)\right) \otimes \rho_{\mathbb{T}\hat{\mathbb{T}}} \right) \right]  \quad\quad (\text{by Lemma~\ref{lem:modSwapTrick}})\\
=& \tr\left[ \left(\left( V_{A_1A_2Q_1Q_2}\otimes F_{T\hat{T}|\mathbb{T}\hat{\mathbb{T}}} \right)\left( E_{A_1Q_1T}\otimes D_{A_2Q_2\hat{T}} \otimes \rho_{\mathbb{T}\hat{\mathbb{T}}} \right) \right)\right],
\end{split}
\end{equation}
which has a similar form to the objective function in \eqref{eq:QWinP_TSepProb}, with the exception that $F_{T\hat{T}|\mathbb{T}\hat{\mathbb{T}}}$ replaces $\Phi_{T\hat{T}|\mathbb{T}\hat{\mathbb{T}}}$. To complete the proof, we write the swap operator $F_{A|\hat{A}}$ in terms of the (non-normalised) maximally-entangled state $\Phi_{A|\hat{A}} =\ket{\Phi}\bra{\Phi}_{A|\hat{A}}$, where $\ket{\Phi}_{A|\hat{A}} = \sum_{i=1}^{d_A} \ket{i}_A\ket{i}_{\hat{A}}$. Namely, we have $F_{A|\hat{A}} = \Phi_{A|\hat{A}}^{T_A}$, where $T_A$ denotes the transposition over the $A$ subsystem. Redefining the variable $\rho$ as $\rho^T$, we then immediately obtain \eqref{eq:QWinP_TSepProb} as this last step leaves the constraints invariant.
\end{proof}

Since product states are extreme points in the set of separable states, we can think of the maximisation in \eqref{eq:QWinP_TSepProb} as the maximisation over separable states. This makes the above problem an instance of the quantum separability problem but with additional linear constraints, in which the optimisation is over the convex hull of product states satisfying the additional linear constraints.


\subsection{Extendible states}\label{subsubsec:ExtStates}

In the previous section, we showed that $\omega_{Q(T)}(V,\pi)$ can be expressed as the quantum separability problem with additional linear constraints. However, optimising over the full set of separable states is known to be NP-hard \cite{gurvits2003classical,gharibian2008strong}, and thus the expression \eqref{eq:QWinP_TSepProb} does not immediately make the problem approachable. Fortunately, there are known relaxations of the separability condition such as the DPS hierarchy of semidefinite programs \cite{doherty2004complete} based on \emph{extendibility}.

\begin{definition}[Extendibility]\label{def:n_extendibleStates}
A bipartite quantum state $\rho_{AB}$ is $n$-extendible if there exists a multipartite quantum state $\rho_{AB^n}$ such that
\begin{align}\label{eq:ExtCondition}
\tr_{B^{n-1}}\left[ \rho_{AB^n} \right] = \rho_{AB},\quad\left(\mathcal{I}_A\otimes\uni^{\pi}_{B^n}\right)\left( \rho_{AB^n} \right) = \rho_{AB^n} \quad \forall\pi\in\mathcal{S}(B^n),
\end{align}
where $\mathcal{S}(B^n)$ is the symmetric group over $B^n$ with $(|B|^n)!$ many elements, $\uni^{\pi}_{B^n} ( \cdot ) = U^{\pi}_{B^n}(\cdot)(U^{\pi}_{B^n})^{\dagger}$ is the adjoint representation of the group, and $U^{\pi}_{B^n}$ is a unitary permutation operator acting on $B^n$.
\end{definition}

The second condition in the above definition means that the state $\rho_{AB^n}$ is invariant under any permutation on $B^n$ with respect to $A$. A quantum state $\rho_{AB}$ is $n$-extendible for all $n\geq2$ if and only if $\rho_{AB}$ is separable \cite{fannes1988symmetric,raggio1989quantum}. Moreover, deciding if a state has a suitable extension satisfying \eqref{eq:ExtCondition} can be done efficiently via SDPs \cite{doherty2002distinguishing,doherty2004complete}; for fixed $n$, the computation resources scale polynomially in the system dimension. The set of $n$-extendible states is thus an outer approximation for the separable set, which in fact converges to the set of separable states when $n\to\infty$. The DPS hierarchy exactly provides efficient SDP relaxations to characterise separable states using the extendibility constraints \eqref{eq:ExtCondition}. The same idea can be extended to the multipartite case as well.

\begin{definition}[Multipartite extendibility]
A $k$-partite state $\rho_{A_1A_2\cdots A_k}\in\hil_{A_1A_2\cdots A_k}$ is $(n_1,\cdots ,n_{k-1})$-extendible if there exists an extension $\rho_{A_1A_2^{n_1}\cdots A_k^{n_{k-1}}}\in\hil_{A_1A_2^{n_1}\cdots A_k^{n_{k-1}}}$ such that
\begin{align}\label{MulExtCondition}
\tr_{A_2^{n_1-1}A_3^{n_2-1}\cdots A_k^{n_{k-1}-1}}\Big[\rho_{A_1A_2^{n_1}\cdots A_k^{n_{k-1}}}\Big] &= \rho_{A_1A_2\cdots A_k}\\
\Big(\mathcal{I}_{A_1}\otimes\uni^{\pi_1}_{A_2^{n_1}}\otimes\cdots \otimes\uni^{\pi_{k-1}}_{A_k^{n_{k-1}}}\Big)\Big( \rho_{A_1A_2^{n_1}\cdots A_k^{n_{k-1}}} \Big) &= \rho_{A_1A_2^{n_1}\cdots A_k^{n_{k-1}}}
\end{align}
for any permutations $\pi_1\in\SH{A_2^{n_1}}, \cdots , \pi_{k-1}\in\SH{A_k^{n_{k-1}}}$. 
\end{definition}

Due to \cite{doherty2005detecting} it is sufficient to consider the cases $n\equiv n_1=\cdots =n_k$ for the convergence of multipartite extendible states, and as in the bipartite case, the set of $(n,\cdots ,n)$-extendible states converges to the set of $k$-partite separable states for $n\to\infty$.


\subsection{Semidefinite programming relaxations}\label{sec:sdps}

To derive our SDP relaxations, we simply replace the separable variables in \eqref{eq:QWinP_TSepProb}, $E_{A_1Q_1T}\otimes D_{A_2Q_2\hat{T}}\otimes\rho_{S\hat{S}}$, with a generic state $\rho_{A_1Q_1TA_2Q_2\hat{T}S\hat{S}}$ and add the extendibility constraints to the optimisation problem
\begin{equation}\label{eq:SDP+Conditions}
\begin{split}
& \qquad \qquad \sdp_{n_1,n_2}(V,\pi,T) := |T|^2 \max_{\rho} \tr\left[ \left( V_{A_1A_2Q_1Q_2}\otimes \Phi_{T\hat{T}|S\hat{S}} \right) \rho_{A_1Q_1TA_2Q_2\hat{T}S\hat{S}} \right]\\
\text{s.t.} \,\,\, & \rho_{(A_1Q_1T)(A_2Q_2\hat{T})^{n_1}(S\hat{S})^{n_2}} \geq 0\, , \quad \tr \left[ \rho_{(A_1Q_1T)(A_2Q_2\hat{T})^{n_1}(S\hat{S})^{n_2}} \right] = 1\\
& \rho_{(A_1Q_1T)(A_2Q_2\hat{T})^{n_1}(S\hat{S})^{n_2}} \text{ perm.~inv.~on } (A_2Q_2\hat{T})^{n_1} \text{ wrt } (A_1Q_1T)(S\hat{S})^{n_2}\\
& \rho_{(A_1Q_1T)(A_2Q_2\hat{T})^{n_1}(S\hat{S})^{n_2}} \text{ perm.~inv.~on } (S\hat{S})^{n_2} \text{ wrt } (A_1Q_1T)(A_2Q_2\hat{T})^{n_1}\\
&\tr_{A_1}[\rho_{(A_1Q_1T)(A_2Q_2\hat{T})^{n_1}(S\hat{S})^{n_2}}] = \left( \sum_{q_1}\pi_1(q_1)\ket{q_1}\!\!\bra{q_1}_{Q_1}\otimes\frac{\mathds{1}_T}{|T|} \right) \otimes \tr_{A_1Q_1T}[\rho_{(A_1Q_1T)(A_2Q_2\hat{T})^{n_1}(S\hat{S})^{n_2}}] \\
&\tr_{A_2}[\rho_{(A_1Q_1T)(A_2Q_2\hat{T})^{n_1}(S\hat{S})^{n_2}}] = \left( \sum_{q_2}\pi_2(q_2)\ket{q_2}\!\!\bra{q_2}_{Q_2}\otimes\frac{\mathds{1}_{\hat{T}}}{|T|} \right) \otimes \tr_{A_2Q_2\hat{T}}[\rho_{(A_1Q_1T)(A_2Q_2\hat{T})^{n_1}(S\hat{S})^{n_2}}]\\
&\rho^{T_{A_1Q_1T}}_{(A_1Q_1T)(A_2Q_2\hat{T})^{n_1}(S\hat{S})^{n_2}} \geq 0, \quad \rho^{T_{(A_2Q_2\hat{T})^{n_1}}}_{(A_1Q_1T)(A_2Q_2\hat{T})^{n_1}(S\hat{S})^{n_2}} \geq 0, \quad \rho^{T_{(S\hat{S})^{n_2}}}_{(A_1Q_1T)(A_2Q_2\hat{T})^{n_1}(S\hat{S})^{n_2}} \geq 0,\;\ldots\;,
\end{split}
\end{equation}
where in the last line $T_{X}$ denotes the partial transpose performed over subsystems $X$ giving the PPT conditions, leading to all possible PPT conditions with respect to valid partitions. The only difference from $\sdp_n(V,\pi,T)$ in the overview section is that we now allow general $n_1,n_2$. In Section~\ref{sec:NPA} we discuss how to combine $\sdp_{n_1,n_2}(V,\pi,T)$ with the NPA hierarchy to achieve a tight bound.

The relaxations \eqref{eq:SDP+Conditions} are used to derive the convergence speed of the hierarchy in the next section. However, with regards to numerics and following Lemma \ref{lem:QWinP_TinSP}, it is advantageous to explicitly take into account that $A_1, A_2, Q_1$, and $Q_2$ start out as classical systems. Hence, we can write the variable $\rho_{(A_1Q_1T)(A_2Q_2\hat{T})(S\hat{S})}$ in the objective function in the block-diagonal form
\begin{align}\label{eq:CQvariable}
\rho_{(A_1Q_1T)(A_2Q_2\hat{T})(S\hat{S})} = \sum_{a_1,a_2,q_1,q_2} \ket{a_1,a_2,q_1,q_2}\!\!\bra{a_1,a_2,q_1,q_2} \otimes \rho_{T\hat{T}S\hat{S}}(a_1,a_2,q_1,q_2).
\end{align}
Then, we can simplify the $A_1, A_2, Q_1, Q_2$ part of the trace in the object function as
\begin{align}\label{eq:SDP+Conditions_Classical}
\sdp_{n_1,n_2}(V,\pi,T) = |T|^2 \max_{\rho} \sum_{a_1,a_2,q_1,q_2} V(a_1,a_2,q_1,q_2) \tr\left[ \Phi_{T\hat{T}|S\hat{S}} \;\rho_{T\hat{T}S\hat{S}}(a_1,a_2,q_1,q_2) \right].
\end{align}
Similarly, it is straightforward to rewrite the linear constraints in terms of the variables in \eqref{eq:CQvariable}.

Next, we show that the value of $\sdp_{n_1, n_2}(V,\pi,T)$ given in \eqref{eq:SDP+Conditions} is naturally upper bounded by $1$.

\begin{proposition}\label{thm:sdp_upper_bound}
Let $\sdp_{n_1, n_2}(V,\pi, T)$ be the $(n_1, n_2)$-th level SDP relaxation for the $|T|$-dimensional two-player free game with rule matrix $V$ and a probability distribution $\pi(q_1,q_2)=\pi_1(q_1)\pi_2(q_2)$. Then, we have that
\begin{align}
0\leq\sdp_{n_1, n_2}(V,\pi,T) \leq 1.
\end{align}
\end{proposition}

\begin{proof}
Let $\rho_{(A_1Q_1T)(A_2Q_2\hat{T})(S\hat{S})}$ be the state optimising the $(n_1, n_2)$-th level SDP relaxation. From the form \eqref{eq:SDP+Conditions_Classical} together with $V(a_1,a_2,q_1,q_2)\leq1$ for all entries, we find
\begin{align}
\sdp_{n_1,n_2}(V,\pi,T) &\leq |T|^2 \max_{\rho} \sum_{a_1,a_2,q_1,q_2}\tr\left[ \Phi_{T\hat{T}|S\hat{S}} \;\rho_{T\hat{T}S\hat{S}}(a_1,a_2,q_1,q_2) \right]\\
&=|T|^2 \tr_{T\hat{T}S\hat{S}} \left[\Phi_{T\hat{T}|S\hat{S}} \tr_{A_1Q_1A_2Q_2} \left[ \rho_{(A_1Q_1T)(A_2Q_2\hat{T})(S\hat{S})}\right]\right]. \label{eq:reshaping_sdp_n1n2_2}
\end{align}
We can now use the linear constraints to write $\rho_{(T\hat{T})(S\hat{S})}$ explicitly. Consider the first linear
constraint \eqref{eq:SDPadditionalLC1}, and taking the trace on the subsystem $A_2$ gives us
\begin{align} \label{eq:mod_first_lin_constaint}
\tr_{A_1A_2}[\rho_{(A_1Q_1T)(A_2Q_2\hat{T})^{n_1}(S\hat{S})^{n_2}}]=\left( \sum_{q_1}\pi_1(q_1)\ket{q_1}\!\!\bra{q_1}_{Q_1}\otimes\frac{\mathds{1}_T}{|T|} \right)\otimes\tr_{A_2} \left[ \rho_{(A_2Q_2\hat{T})^{n_1}(S\hat{S})^{n_2}} \right].
\end{align}
If we now consider the second linear constraint \eqref{eq:SDPadditionalLC2}, and we take the trace
on the subsystems $A_1Q_1T$ and $A_2$, we get that
\begin{align} \label{eq:mod_second_lin_constaint}
\tr_{A_2} \left[ \rho_{(A_2Q_2\hat{T})^{n_1}(S\hat{S})^{n_2}} \right]=\left( \sum_{q_2} \pi_2(q_2) \ket{q_2}\!\!\bra{q_2}_{Q_2} \otimes \frac{\mathds{1}_{\hat{T}}}{|T|} \right)\otimes\rho_{(A_2Q_2\hat{T})^{n_1-1}(S\hat{S})^{n2}}.
\end{align}
By combining \eqref{eq:mod_first_lin_constaint} and \eqref{eq:mod_second_lin_constaint}, and
taking the trace on the remaining $n_1 - 1$ subsystems $A_2Q_2\hat{T}$ and $n_2 - 1$ subsystems $S\hat{S}$,
we obtain
\begin{align}
\tr_{A_1A_2} \left[ \rho_{(A_1Q_1T)(A_2Q_2\hat{T})(S\hat{S})} \right]=\left( \sum_{q_1} \pi_1(q_1) \ket{q_1}\!\!\bra{q_1}_{Q_1} \otimes \frac{\mathds{1}_T}{|T|} \right)\otimes\left( \sum_{q_2} \pi_2(q_2) \ket{q_2}\!\!\bra{q_2}_{Q_2} \otimes \frac{\mathds{1}_{\hat{T}}}{|T|} \right)\otimes\rho_{S\hat{S}}.
\end{align}
If we now take the trace over $Q_1Q_2$, we obtain the following explicit form for $\rho_{(T\hat{T})(S\hat{S})}$,
\begin{align}
\tr_{A_1Q_1A_2Q_2} \left[ \rho_{(A_1Q_1T)(A_2Q_2\hat{T})(S\hat{S})} \right]=\frac{\mathds{1}_{T\hat{T}}}{|T|^2} \otimes \rho_{S\hat{S}}.
\end{align}
By replacing the above state in \eqref{eq:reshaping_sdp_n1n2_2}, and by using the swap trick, we find that
\begin{align}
\operatorname{sdp}_{n_1,n_2}(V,\pi,T)\leq\tr_{T\hat{T}S\hat{S}} \left[ \Phi_{T\hat{T}|S\hat{S}} \left( \id_{T\hat{T}} \otimes \rho_{S\hat{S}} \right) \right]=\tr_{T\hat{T}S\hat{S}} \left[ F_{T\hat{T}|S\hat{S}} \left( \id_{T\hat{T}} \otimes \rho_{S\hat{S}} \right) \right]=\tr_{T\hat{T}} \left[ \rho_{T\hat{T}} \right] = 1,
\end{align}
where we have taken the partial transpose over $T\hat{T}$, and used $\Trt{M_{AB} N_{AB}^{T_A}} = \Trt{M_{AB}^{T_A} N_{AB}}$.
\end{proof}


\section{Convergence of the hierarchy}\label{sec:convergence}

\subsection{Setting}

Quantum de Finetti theorems give a quantitative upper bound on the trace distance between $(n_1,n_2)$-extendible states and the set of separable states as a function of $n_1$, $n_2$, and the dimension of the states. Thus, quantum de Finetti theorems can tell us how good the approximation $\sdp_{n_1,n_2}(V,\pi,T)$ is for each level, leading to a quantitative convergence speed of the relaxations. As pointed out in Section~\ref{sec:overview}, $\sdp_{n_1,n_2}(V,\pi,T)$ in \eqref{eq:SDP+Conditions} has additional linear constraints and hence we need adapted quantum de Finetti theorems that describe the distance between extendible states and the set of separable states consisting of product states which satisfy the linear constraints.

In the following section, we focus on the tripartite quantum de Finetti theorem with partial trace constraints as in \eqref{eq:SDPadditionalLC1} and \eqref{eq:SDPadditionalLC2}. Nevertheless, the proof technique is general and is extended in Appendix \ref{app:MulQdeFTwithLC} to multipartite quantum de Finetti theorems with arbitrary linear constraints as stated in \eqref{thm:MulQdeFTwithLC}. Employing the derived quantum de Finetti theorems, we prove the convergence rate of our SDP hierarchy in Section~\ref{subsec:ConvergenceRate}.


\subsection{Tripartite quantum de Finetti theorem with partial trace constraints}\label{subsec:TriQdeF}

We follow an information-theoretic proof technique, which is based on quantum entropy inequalities \cite{brandao2017quantum,brandao2016product}. Our results can be seen as a generalisation of \cite[Theorem 3.4]{berta2018semidefinite} to the multipartite case. The technical novelty of this generalisation is a more sophisticated use of the conditional multipartite quantum mutual information, as well as identifying the optimal loss in distinguishability relative to quantum side information (for the latter, see the second part of Lemma~\ref{lem:DataHiding} and its proof in Appendix~\ref{app:distortion_sideI}).

We make use of the properties of the conditional quantum mutual information. In the multipartite case, there exist several definitions of quantum mutual information~\cite{yang2009squashed}. Here, we use the following definition; given $k$ quantum systems $A_1, \ldots, A_k$ described by the state $\rho_{A_1 A_2 \ldots A_k}$, the multipartite quantum mutual information is defined as
\begin{align}
I(A_1:A_2: \ldots :A_k)_{\rho} := \sum_{i=1}^k S(A_i) - S(A_1 A_2 \ldots A_k),
\end{align}
where $S(A_i)=-\tr\left[\rho_{A_i}\log\rho_{A_i}\right]$ is the von Neumann entropy~\cite{lindblad1973entropy} of the reduced state $\rho_{A_i}$. It is worth noting that the above definition is equivalent to the one in terms of the relative entropy distance between the state and the tensor product of its marginals,
\begin{align}\label{eq:MutInfRelEnt}
I(A_1:A_2: \ldots :A_k)_{\rho} = D(\rho_{A_1\cdots A_k}||\rho_{A_1}\otimes \ldots \otimes\rho_{A_k}),
\end{align}
where $D(\rho||\sigma) = \Trt{\rho \left( \log \rho - \log \sigma \right)}$ is the relative entropy between two states $\rho$ and $\sigma$, s.t.~$\supp(\rho) \subset \supp(\sigma)$. Given $k$ quantum systems $A_1, \ldots, A_k$ and a classical system $B$, described by the global state $\rho_{A_1A_2\cdots A_kB}$, the conditional multipartite quantum mutual information is defined as 
\begin{align}
I(A_1:A_2: \ldots :A_k|B)_{\rho}:= \sum_{i=1}^k S(A_iB) - S(A_1A_2 \ldots A_kB)-S(B).
\end{align}
The conditional multipartite quantum mutual information with classical $B$-system can be expressed in terms of the bipartite one~\cite[Lemma 13]{brandao2017quantum},
\begin{align}\label{eq:MulttoBip}
I(A_1:\ldots:A_k|B)_{\rho} = I(A_1:A_2|B) + I(A_1A_2:A_3|B) + \ldots + I(A_1 \ldots A_{k-1}:A_k|B).
\end{align}
An additional property of the mutual information is the \emph{chain rule},
\begin{align}\label{eq:ChainRule}
I(AB:C) = I(B:C)+I(A:C|B),
\end{align}
which holds for the conditional mutual information as well,
\begin{align}\label{eq:ChainRuleCond}
I(AB:C|D) = I(B:C|D)+I(A:C|BD).
\end{align}
Finally, Pinsker's inequality states that
\begin{align}\label{eq:PinskerInq}
\norm{\rho - \sigma}_1^2 \leq 2 \ln 2\cdot D(\rho||\sigma),
\end{align}
where $\norm{A}_1:= \Trt{\sqrt{A^{\dagger}A}}$ is the trace norm. We define the conditional state $\rho_{A|z}$ of a state $\rho_{AZ}$ with a classical $Z$ system as
\begin{align}
\rho_{A|z} := \frac{\Tr_Z\left[\rho_{AZ}(\mathds{1}_A\otimes\ket{z}\!\!\bra{z}_Z)\right]}{\Tr\left[\rho_{AZ}(\mathds{1}_A\otimes\ket{z}\!\!\bra{z}_Z)\right]}.
\end{align}
This describes the state after measurement on system $Z$, when the measurement outcome is $z$.

The following lemma bounds the conditional quantum mutual information when the state satisfies a specific partial trace constraint. Note that specifying the linear constraints gives a stronger bound than the general case in terms of the dimension (Lemma~\ref{lem:CMIbound}).

\begin{lemma}\label{lem:BoundCondMIwithLC}
Let $\rho_{ABCD}$ be a quantum state such that $\Trp{A}{\rho_{ABCD}} = \rho_B \otimes \rho_{CD}$. Then, we have that
\begin{align}
I(AB:C|D)_{\rho} \leq 2 \log |A|.
\end{align}
\end{lemma}

\begin{proof}
Using the chain rule for the conditional quantum mutual information \eqref{eq:ChainRuleCond}, we have
\begin{align}
I(AB:C|D)_{\rho} &= I(B:C|D)_{\rho} + I(A:C|BD)_{\rho} = I(A:C|BD)_{\rho},
\end{align}
where we used the fact that the systems $BC$ are uncorrelated after tracing out the system $A$. Using the chain rule of the mutual information \eqref{eq:ChainRule} we get that
\begin{align}
I(A:C|BD) = I(A:BCD) - I(A:BD) \leq I(A:BCD) \leq 2 \log |A|.
\end{align}
\end{proof}

Exploiting Lemma~\ref{lem:BoundCondMIwithLC}, the next lemma provides a bound for the tripartite conditional quantum mutual information of a quantum-classical state.

\begin{lemma}\label{lem:BoundTripCondMIwithLC}
Consider a quantum-classical state $\rho_{(AB)Z_1^{n_1}W_1^{n_2}}$ with $Z$- and $W$-systems classical, which satisfies the linear constraint
\begin{align}
\Trp{A}{\rho_{(AB)Z_1^{n_1}W_1^{n_2}}} = \rho_B \otimes \rho_{Z_1^{n_1}W_1^{n_2}}.
\end{align}
Then, there exist integers $\bar{m} \in \left[0,n_1\right)$ and $\bar{\ell} \in \left[0,n_2\right)$ such that
\begin{align}
I(AB:Z_{\bar{m}+1}:W_{\bar{\ell}+1}|Z_1^{\bar{m}}W_1^{\bar{\ell}}) \leq 2\left(\frac{\log |A|}{n_1} + \frac{(\log |A| + \log |Z|)}{n_2}\right),
\end{align}
and
\begin{align}\label{eq:TripCondMIwithPT_E}
\E_{z_1^{\bar{m}}w_1^{\bar{\ell}}} \left[ \norm{ \rho_{(AB)Z_{\bar{m}+1}W_{\bar{\ell}+1}|z_1^{\bar{m}}w_1^{\bar{\ell}}} - \rho_{AB|z_1^{\bar{m}}w_1^{\bar{\ell}}}\otimes\rho_{Z_{\bar{m}+1}|z_1^{\bar{m}}w_1^{\bar{\ell}}}\otimes\rho_{W_{\bar{\ell}+1}|z_1^{\bar{m}}w_1^{\bar{\ell}}} }_1^2 \right] \leq 4 \ln 2 \left( \frac{\log |A|}{n_1} + \frac{\log |A| + \log |Z|}{n_2} \right).
\end{align}
\end{lemma}

\begin{proof}
As the first step, notice that for all $m$ and $\ell$, using the relation between multipartite mutual information and bipartite ones in \eqref{eq:MulttoBip}, we can express the multipartite conditional mutual information as,
\begin{align}\label{eq:TriToBi}
I(AB:Z_{m+1}:W_{\ell+1}|Z_1^{m}W_1^{\ell}) = I(AB:Z_{m+1}|Z_1^{m}W_1^{\ell}) + I(ABZ_{m+1}:W_{\ell+1}|Z_1^{m}W_1^{\ell}).
\end{align}
We now derive a bound for the right-hand side (RHS) of the above equation, valid for a specific choice of $m$ and $\ell$. For any $\ell \in \left[ 0, n_2 \right)$, we have that 
\begin{align}
\sum_{m=0}^{n_1-1} I(AB:Z_{m+1}|Z_1^mW_1^{\ell}) = I(AB:Z_1^{n_1}|W_1^{\ell}) \leq 2\log |A|,
\end{align}
where we used the chain rule in \eqref{eq:ChainRuleCond} for the first equality and Lemma~\ref{lem:BoundCondMIwithLC} for the second inequality. By summing over $\ell$, we obtain
\begin{align}\label{eq:FirBound}
\sum_{m=0}^{n_1-1} \sum_{\ell=0}^{n_2-1} I(AB:Z_{m+1}|Z_1^mW_1^{\ell}) \leq 2 n_2 \log |A|.
\end{align}
Similarly, for any $m \in \left[ 0, n_1 \right)$, we find that
\begin{align}
\sum_{\ell=0}^{n_2-1} I(ABZ_{m+1}:W_{\ell+1}|Z_1^mW_1^{\ell}) = I(ABZ_{m+1}:W_1^{n_2}|Z_1^{m}) \leq 2\log |AZ_{m+1}|,
\end{align}
where the inequality follows from Lemma~\ref{lem:BoundCondMIwithLC} when the first partition of the system is $AZ_{m+1}$. Summing over $m$ gives us
\begin{align}\label{eq:SecBound}
\sum_{m=0}^{n_1-1} \sum_{\ell=0}^{n_2-1} I(ABZ_{m+1}:W_{\ell+1}|Z_1^mW_1^{\ell}) \leq 2 n_1 \left( \log |A| + \log |Z| \right).
\end{align}
Combining \eqref{eq:FirBound} and \eqref{eq:SecBound}, we obtain
\begin{align}
2 n_2 \log |A| + 2 n_1 \left(\log |A| + \log |Z| \right) &\geq \sum_{m=0}^{n_1-1} \sum_{\ell=0}^{n_2-1} \left[ I(AB:Z_{m+1}|Z_1^mW_1^{\ell}) + I(ABZ_{m+1}:W_{\ell+1}|Z_1^mW_1^{\ell}) \right] \nonumber \\
&\geq n_1 n_2 \left[ I(AB:Z_{\bar{m}+1}|Z_1^{\bar{m}}W_1^{\bar{\ell}}) + I(ABZ_{\bar{m}+1}:W_{\bar{\ell}+1}|Z_1^{\bar{m}}W_1^{\bar{\ell}}) \right] \nonumber \\
&= n_1 n_2 \, I(AB:Z_{\bar{m}+1}:W_{\bar{\ell}+1}|Z_1^{\bar{m}}W_1^{\bar{\ell}})
\end{align}
where $\bar{m}$ and $\bar{\ell}$ are the indices of the smallest element in the sum, and the last equality follows from \eqref{eq:TriToBi}. This proves the first part of the theorem. The second part is obtained as follows. First notice that, when the conditioning system is classical, we can write the conditional mutual information as
\begin{align}
I(AB:Z_{\bar{m}+1}:W_{\bar{\ell}+1}|Z_1^{\bar{m}}W_1^{\bar{\ell}}) = \E_{z_1^{\bar{m}}w_1^{\bar{\ell}}} \left[ I(AB:Z_{\bar{m}+1}:W_{\bar{\ell}+1})_{\rho_{ABZ_{\bar{m}+1}W_{\bar{\ell}+1}|z_1^{\bar{m}}w_1^{\bar{\ell}}}} \right].
\end{align}
Then, by using the fact that the mutual information can be expressed in terms of the relative entropy, \eqref{eq:MutInfRelEnt}, and by using Pinsker's inequality in \eqref{eq:PinskerInq}, we can derive the second part of the theorem.
\end{proof}
\noindent We remark that the resulting bound is not symmetric under exchanging the systems $Z$ and $W$ even though the conditions are symmetric. This is because we used the bipartite expression for the conditional multipartite quantum mutual information in \eqref{eq:TriToBi}. However, the same proof works with the roles of the systems $Z$ and $W$ exchanged and one can freely choose whichever bound is stronger.

The next ingredient is a bound on the loss in distinguishability between quantum states when they are processed with optimal measurements\,---\,called \emph{minimal distortion}.

\begin{lemma}\label{lem:DataHiding}
\begin{enumerate}
\item There exist fixed measurements $\mathcal{M}_A$, $\mathcal{M}_B$, and $\mathcal{M}_C$ with at most $|A|^8$, $|B|^8$, and $|C|^8$ outcomes, respectively, such that for every traceless Hermitian operator $\gamma_{ABC}$ on $\mathcal{H}_{ABC}$ 
\begin{align}
\norm{\gamma_{ABC}}_1\leq 18^{3/2}\sqrt{|ABC|}\cdot\norm{\left(\mathcal{M}_A\otimes \mathcal{M}_B\otimes\mathcal{M}_C\right)\left(\gamma_{ABC}\right)}_1.
\end{align}
\item There exists a fixed measurement $\mathcal{M}_B$ with at most $|B|^6$ outcomes such that for every traceless Hermitian operator $\gamma_{AB}$ on $\mathcal{H}_{AB}$ 
\begin{align}\label{eq:DataHiding_Side}
\norm{\gamma_{AB}}_1\leq2|B|\cdot\norm{\left(\mathcal{I}_A\otimes\mathcal{M}_B\right) \left(\gamma_{AB}\right)}_1.
\end{align}
\end{enumerate}
\end{lemma}

The first part is straightforward from \cite[Lemma 14]{brandao2016product}. We remark that when a traceless Hermitian operator already has a classical subsystem, i.e., $\gamma_{ABCZ}$ with classical $Z$-system, the dimension factor only includes the dimension of the quantum systems
\begin{align}\label{eq:dataHiding_with_classical}
    \norm{\gamma_{ABCZ}}_1 \leq 18^{3/2}\sqrt{|ABC|}\cdot\norm{\left(\mathcal{M}_A\otimes\mathcal{M}_B\otimes\mathcal{M}_C\otimes\mathcal{I}_Z\right)\left(\gamma_{ABCZ}\right)}_1.
\end{align}
This follows easily as $\norm{\sum_z\rho^z_A\otimes\ket{z}\!\!\bra{z}}_1 = \sum_z \norm{\rho^z_A}_1$ for classical-quantum states $\rho_{AZ}$.

The proof of the second part is given in Appendix~\ref{app:distortion_sideI}. We note that the latter improves on factor $\sqrt{18}|B|^{3/2}$ given in \cite[Equation (68)]{brandao2017quantum}. As there exist quantum states $\rho_{AB}$ and $\sigma_{AB}$ with \cite{matthews2009distinguishability}
\begin{align}
\left\|\rho_{AB}-\sigma_{AB}\right\|_1=2\quad\text{and}\quad\sup_{\mathcal{M}_B}\norm{\left(\mathcal{I}_A\otimes\mathcal{M}_B\right) \left(\gamma_{AB}\right)}_1=\frac{2}{|B|+1},
\end{align}
this result establishes the dimension dependence $\Omega(|B|)$ for \emph{minimal distortion relative to quantum side information}. This answers a question left open in \cite{berta2018semidefinite}.

Now, we derive the tripartite quantum de Finetti theorem with partial trace constraints.
\begin{theorem}\label{thm:TriQdeFTwithPTLC}
Let $\rho_{(A\tilde{A})B^{n_1}(C\tilde{C})^{n_2}}$ be a quantum state invariant under permutation on $B^{n_1}$ and $(C\tilde{C})^{n_2}$ with respect to the other systems, satisfying
\begin{align}
&\Trp{A}{\rho_{(A\tilde{A})B^{n_1}(C\tilde{C})^{n_2}}} = \mathcal{X}_{\tilde{A}} \otimes \rho_{B^{n_1}(C\tilde{C})^{n_2}} \quad &\text{linear constraint on $A\tilde{A}$} \label{eq:TriQdeFT_LConA}\\
&\Trp{C}{\rho_{(A\tilde{A})B^{n_1}(C\tilde{C})^{n_2}}} = \mathcal{Z}_{\tilde{C}} \otimes \rho_{(A\tilde{A})B^{n_1}(C\tilde{C})^{n_2-1}} \quad &\text{linear constraint on $C\tilde{C}$} \label{eq:TriQdeFT_LConC}
\end{align}
for some operators $\mathcal{X}_{\tilde{A}}$, and $\mathcal{Z}_{\tilde{C}}$. Then, there exist a probability distribution $\{p_i\}_{i\in I}$ and sets of quantum states $\{\sigma_{A\tilde{A}}^i\}_{i\in I}$, $\{\omega_{B}^i\}_{i\in I}$ and $\{\tau_{C\tilde{C}}^i\}_{i\in I}$ such that
\begin{align}\label{eq:TriQdFTbetterbound}
&\norm{ \rho_{(A\tilde{A})B(C\tilde{C})} - \sum_{i\in I} p_i \, \sigma_{A\tilde{A}}^i\otimes\omega_{B}^i\otimes\tau_{C\tilde{C}}^i}_1\notag\\
&\leq\;\min\left\{18^{3/2}\sqrt{|A\tilde{A}BC\tilde{C}|},4|BC\tilde{C}|\right\} \sqrt{4 \ln 2} \left(\sqrt{\frac{\log |A| + 8 \log |B|}{n_2} + \frac{\log |A|}{n_1}}\right)
\end{align}
with $\Trp{A}{\sigma_{A\tilde{A}}^i} = \mathcal{X}_{\tilde{A}}$ and $\tr_{C} \left[\tau_{C\tilde{C}}^i\right] = \mathcal{Z}_{\tilde{C}}$ for all $i\in I$.
\end{theorem}

\begin{proof}
Let $\mathcal{M}_{B\rightarrow Y}$ ($\mathcal{M}_{C\tilde{C}\rightarrow Z}$) be a quantum-to-classical channel, i.e.~a measurement, from the quantum system $B$ ($C\tilde{C}$) to the classical system $Y$ ($Z$). By measuring both $B$ and $C\tilde{C}$, we obtain the following quantum-classical state,
\begin{align}
\rho_{(A\tilde{A})Y^{n_1}Z^{n_2}} = (\mathcal{I}_{A\tilde{A}} \otimes \mathcal{M}^{\otimes n_1}_{B\rightarrow Y} \otimes \mathcal{M}^{\otimes n_2}_{(C\tilde{C})\rightarrow Z}) \left( \rho_{(A\tilde{A})B^{n_1}(C\tilde{C})^{n_2}} \right),
\end{align}
where $\mathcal{M}_{B\rightarrow Y}^{\otimes n_1}$ is composed of $n_1$ independent and identical  measurements $\mathcal{M}_{B\rightarrow Y}$, each one acting on a different $B$ system, and $\mathcal{M}_{(C\tilde{C})\rightarrow Z}^{\otimes n_2}$ is defined similarly as well. It is easy to see that the post-measurement state still satisfies the linear constraint \eqref{eq:TriQdeFT_LConA}
\begin{align}
\Trp{A}{\rho_{(A\tilde{A})Y^{n_1}Z^{n_2}}} = \mathcal{X}_{\tilde{A}} \otimes \rho_{Y^{n_1}Z^{n_2}},
\end{align}
and therefore is compatible with the hypothesis of Lemma~\ref{lem:BoundTripCondMIwithLC}. Then, we can find $m \in \left[0,n_1\right)$ and $\ell \in \left[0,n_2\right)$ such that 
\begin{align}\label{eq:TRIBetterBoundstep1_PT}
\E_{y^{m}z^{\ell}} \left[ \norm{ \rho_{(A\tilde{A})Y_{m+1}Z_{\ell+1}|y^{m}z^{\ell}} - \rho_{A\tilde{A}|y^{m}z^{\ell}} \otimes \rho_{Y_{m+1}|y^{m}z^{\ell}} \otimes \rho_{Z_{\ell+1}|y^{m}z^{\ell}}}_1^2 \right] \leq 4 \ln 2 \left( \frac{\log |A|}{n_1} + \frac{\log |A| + \log |Y|}{n_2} \right)
\end{align}
by Lemma~\ref{lem:BoundTripCondMIwithLC}. For given $y^m$ and $z^{\ell}$, let us define the traceless Hermitian operator
\begin{align}
\gamma_{(A\tilde{A})B(C\tilde{C})} = \rho_{(A\tilde{A})B_{m+1}(C\tilde{C})_{\ell+1}|y^mz^{\ell}} - \rho_{A\tilde{A}|y^mz^{\ell}} \otimes \rho_{B_{m+1}|y^mz^{\ell}} \otimes \rho_{(C\tilde{C})_{\ell+1}|y^mz^{\ell}}.
\end{align}
This operator is related to the one in the left-hand side (LHS) of \eqref{eq:TRIBetterBoundstep1_PT} by a measurement on the $B$ and $C\tilde{C}$ systems, 
\begin{align}
\left(\mathcal{I}_{A\tilde{A}} \otimes \mathcal{M}_{B\rightarrow Y} \otimes \mathcal{M}_{C\tilde{C}\rightarrow Z}\right)\left( \gamma_{(A\tilde{A})B(C\tilde{C})} \right) = \rho_{(A\tilde{A})Y_{m+1}Z_{\ell+1}|y^{m}z^{\ell}} - \rho_{A\tilde{A}|y^{m}z^{\ell}} \otimes \rho_{Y_{m+1}|y^{m}z^{\ell}}\otimes \rho_{Z_{\ell+1}|y^{m}z^{\ell}}.
\end{align}
As in \cite{berta2018semidefinite}, we consider two ways to relate the trace norm of $\gamma_{(A\tilde{A})B(C\tilde{C})}$ to $(\mathcal{I}_{A\tilde{A}} \otimes \mathcal{M}_{B\rightarrow Y}\otimes\mathcal{M}_{C\tilde{C}\rightarrow Z})\left( \gamma_{(A\tilde{A})B(C\tilde{C})} \right)$. Firstly, we will exploit the second part of Lemma~\ref{lem:DataHiding} iteratively. If we let the measurements $\mathcal{M}_{B\rightarrow Y}$ and $\mathcal{M}_{C\tilde{C}\rightarrow Z}$ be the measurement described in the second part of Lemma~\ref{lem:DataHiding}, we get
\begin{align}
\norm{\gamma_{(A\tilde{A})B(C\tilde{C})}}_1 &\leq 2|C\tilde{C}| \norm{\left(\mathcal{I}_{A\tilde{A}B}\otimes\mathcal{M}_{C\tilde{C}\to Z}\right) \left(\gamma_{(A\tilde{A})B(C\tilde{C})}\right)}_1 \\
&\leq 2|B|\times 2|C\tilde{C}| \norm{(\mathcal{I}_{A\tilde{A}Z}\otimes\mathcal{M}_{B\rightarrow Y})\left((\mathcal{I}_{A\tilde{A}B}\otimes\mathcal{M}_{C\tilde{C}\rightarrow Z})\left(\gamma_{(A\tilde{A})B(C\tilde{C})}\right)\right)}_1 \\
&= 4|BC\tilde{C}|\norm{(\mathcal{I}_{A\tilde{A}}\otimes\mathcal{M}_{B\rightarrow Y}\otimes\mathcal{M}_{C\tilde{C}\rightarrow Z})\left(\gamma_{(A\tilde{A})B(C\tilde{C})}\right)}_1,
\end{align}
and $|Y|\leq|B|^6$. Secondly, if we use the first part of Lemma~\ref{lem:DataHiding} instead, we obtain
\begin{align}
\norm{\gamma_{(A\tilde{A})B(C\tilde{C})}}_1 &\leq \sqrt{18^3|A\tilde{A}BC\tilde{C}|} \, \norm{(\mathcal{M}_{A\tilde{A}}\otimes\mathcal{M}_{B\rightarrow Y} \otimes \mathcal{M}_{C\tilde{C}\rightarrow Z}) \left( \gamma_{(A\tilde{A})B(C\tilde{C})} \right) }_1 \nonumber \\
&\leq \sqrt{18^3|A\tilde{A}BC\tilde{C}|} \, \norm{(\mathcal{I}_{A\tilde{A}}\otimes\mathcal{M}_{B\rightarrow Y} \otimes \mathcal{M}_{C\tilde{C}\rightarrow Z}) \left( \gamma_{(A\tilde{A})B(C\tilde{C})} \right) }_1
\end{align}
with $|Y|\leq|B|^8$, where the second inequality follows from the monotonicity of the trace norm under CPTP maps. Combining \eqref{eq:TRIBetterBoundstep1_PT} with the above two results gives
\begin{align}\label{eq:TriQdeFTwithPT_eq}
\E_{y^mz^{\ell}} &\left[ \norm{ \rho_{(A\tilde{A})B_{m+1}(C\tilde{C})_{\ell+1}|y^mz^{\ell}} - \rho_{A\tilde{A}|y^mz^{\ell}} \otimes \rho_{B_{m+1}|y^mz^{\ell}} \otimes \rho_{(C\tilde{C})_{\ell+1}|y^mz^{\ell}}}_1^2 \right] \nonumber \\
&\leq\;\min\left\{18^{3/2}\sqrt{|A\tilde{A}BC\tilde{C}|},4|BC\tilde{C}|\right\}^2 (4 \ln 2) \left(\frac{\log |A| + 8 \log |B|}{n_2} + \frac{\log |A|}{n_1}\right)
\end{align}
Depending on the dimensions, we can choose the tighter one between the two bounds. The following chain of inequalities combining with \eqref{eq:TriQdeFTwithPT_eq} concludes the proof of the first part of the theorem,
\begin{align}
&\norm{ \rho_{(A\tilde{A})B_{m+1}(C\tilde{C})_{\ell+1}} - \E_{y^{m}z^{\ell}} \left[ \rho_{A\tilde{A}|y^mz^{\ell}} \otimes \rho_{B_{m+1}|y^mz^{\ell}} \otimes \rho_{(C\tilde{C})_{\ell+1}|y^mz^{\ell}} \right] }_1 \nonumber \\
&\leq \E_{y^mz^{\ell}} \left[ \norm{ \rho_{(A\tilde{A})B_{m+1}(C\tilde{C})_{\ell+1}|y^mz^\ell} - \rho_{A\tilde{A}|y^{m}z^{\ell}} \otimes \rho_{B_{m+1}|y^{m}z^{\ell}} \otimes \rho_{(C\tilde{C})_{\ell+1}|y^{m}z^{\ell}}}_1 \right] \qquad (\because\text{ Triangular Inequality}) \nonumber \\
&\leq \sqrt{ \E_{y^mz^{\ell}} \left[ \norm{ \rho_{(A\tilde{A})B_{m+1}(C\tilde{C})_{\ell+1}|y^mz^{\ell}} - \rho_{A\tilde{A}|y^{m}z^{\ell}} \otimes \rho_{B_{m+1}|y^{m}z^{\ell}} \otimes \rho_{(C\tilde{C})_{\ell+1}|y^{m}z^{\ell}}}_1^2 \right]} \qquad (\because\text{ Concavity}), \nonumber
\end{align}
where the quantum state $\E_{y^mz^{\ell}} \left[ \rho_{A\tilde{A}|y^mz^{\ell}} \otimes \rho_{B_{m+1}|y^mz^{\ell}} \otimes \rho_{(C\tilde{C})_{\ell+1}|y^mz^\ell} \right]$ is separable over the tripartite cut $A\tilde{A}|B|C\tilde{C}$. It is worth noting that, since the state under consideration is permutation invariant over $B^{n_1}$ and $(C\tilde{C})^{n_2}$, the result we obtain is independent of the specific $m$ and $\ell$ considered. This closes the proof of \eqref{eq:TriQdFTbetterbound}.

To conclude the proof, we need to show that each state in the mixture still satisfies the corresponding linear constraint. For a state $\rho_{A\tilde{A}|y^mz^{\ell}}$ describing the $A\tilde{A}$ system, we have
\begin{align}
\Trp{A}{\rho_{A\tilde{A}|y^mz^{\ell}}} &= \frac{ \Trp{Y^mZ^{\ell}} {(\id_{\tilde{A}} \otimes M_{B^m\rightarrow Y^m}^{y^m} \otimes M_{(C\tilde{C})^{\ell}\rightarrow Z^{\ell}}^{z^{\ell}})\Trp{A}{\rho_{(A\tilde{A})B^m(C\tilde{C})^{\ell}}}}} {\Trt{(\id_{A\tilde{A}} \otimes M_{B^m\rightarrow Y^m}^{y^m} \otimes M_{(C\tilde{C})^{\ell}\rightarrow Z^{\ell}}^{w^{\ell}})\rho_{(A\tilde{A})B^m(C\tilde{C})^{\ell}}}} \nonumber \\
&= \frac{\Trp{Y^mZ^{\ell}}{(\id_{\tilde{A}} \otimes M_{B^m\rightarrow Y^m}^{y^m} \otimes M_{(C\tilde{C})^{\ell}\rightarrow Z^{\ell}}^{z^{\ell}})  \left( \mathcal{X}_{\tilde{A}} \otimes \rho_{B^m(C\tilde{C})^{\ell}}\right)}}{\Trt{(M_{B^m\rightarrow Y^m}^{y^m} \otimes M_{(C\tilde{C})^{\ell}\rightarrow Z^{\ell}}^{z^{\ell}})\rho_{B^m(C\tilde{C})^{\ell}}}}\\
&= \mathcal{X}_{\tilde{A}},
\end{align}
where $M_{B^m\rightarrow Y^m}^{y^m}$ is the measurement effect over $B^m$ corresponding to the outcome $y^m$, while $M_{(C\tilde{C})^{\ell}\rightarrow Z^{\ell}}^{z^{\ell}}$ is the effect over $(C\tilde{C})^{\ell}$ corresponding to the outcome $z^{\ell}$. For a state $\rho_{(C\tilde{C})_{\ell+1}|y^mz^{\ell}}$, we have
\begin{align}
\tr_{C} \left[\rho_{(C\tilde{C})_{\ell+1}|y^mz^{\ell}}\right] &= \frac{\Trp{Y^mZ^{\ell}}{(M_{B^m\rightarrow Y^m}^{y^m}\otimes M_{(C\tilde{C})^{\ell}\rightarrow Z^{\ell}}^{z^{\ell}}\otimes\id_{\tilde{C}_{\ell+1}})\tr_C\left[\rho_{B^{m}(C\tilde{C})^{\ell+1}}\right]}}{\Trt{(M_{B^m\rightarrow Y^m}^{y^m}\otimes M_{(C\tilde{C})^{\ell}\rightarrow Z^{\ell}}^{z^{\ell}}\otimes\id_{(C\tilde{C})_{\ell+1}})\rho_{B^{m}(C\tilde{C})^{\ell+1}}}} \nonumber \\
&= \frac{\Trp{Y^mz^{\ell}}{(M_{B^m\rightarrow Y^m}^{y^m}\otimes M_{(C\tilde{C})^{\ell}\rightarrow Z^{\ell}}^{z^{\ell}}\otimes\id_{\tilde{C}_{\ell+1}})\left(\mathcal{Z}_{\tilde{C}_{\ell+1}} \otimes \rho_{B^m(C\tilde{C})^{\ell}} \right)}}{\Trt{(M_{B^m\rightarrow Y^m}^{y^m} \otimes M_{(C\tilde{C})^{\ell}\rightarrow Z^{\ell}}^{z^{\ell}})\rho_{B^{m}(C\tilde{C})^{\ell}}}}\\
&= \mathcal{Z}_{\tilde{C}_{\ell+1}}.
\end{align}
\end{proof}

In the case of general linear constraints, we have a worse bound with $|A||\tilde{A}|$ instead of $|A|$ in \eqref{eq:TriQdFTbetterbound} as stated in Theorem~\ref{thm:MulQdeFTwithLC}. We refer to Appendix~\ref{app:MulQdeFTwithLC} for the proof.


\subsection{Convergence rate}\label{subsec:ConvergenceRate}

In the last section, we proved the tripartite quantum de Finetti theorem with partial trace constraints on the state. In this section, we use this result to prove Theorem~\ref{thm:Convergence}, providing a bound on the quantitative accuracy of $\sdp_{n_1,n_2}(V,\pi,T)$ for each level.

\begin{theorem}\label{thm:Convergencen1n2}
Let $\sdp_{n_1,n_2}(V,\pi,T)$ be the $(n_1, n_2)$-th level SDP relaxation for the $|T|$-dimensional two-player free game with rule matrix $V$ and probability distribution $\pi(q_1,q_2)=\pi_1(q_1)\pi_2(q_2)$. Then, we have
\begin{align}
0 \leq \sdp_{n_1, n_2}(V,\pi,T) - w^Q_T(V, \pi) \leq O\Bigg( |T|^6\sqrt{\frac{\log|TA|}{n_2} + \frac{\log|A|}{n_1}} \Bigg).
\end{align}
\end{theorem}

\begin{proof}
Let $\rho_{A_1Q_1TA_2Q_2\hat{T}S\hat{S}}$ be the optimal state of the $(n_1, n_2)$-th level relaxation $\sdp_{n_1,n_2}(V,\pi,T)$. The state should be $(n_1,n_2)$-extendible since all feasible states must be $(n_1,n_2)$-extendible states satisfying the linear constraints. Then, we have
\begin{equation}
\begin{split}
&\sdp_{n_1,n_2}(V,\pi,T) = |T|^2 \tr\left[ \left( V_{A_1A_2Q_1Q_2}\otimes \Phi_{T\hat{T}|S\hat{S}}  \right) \rho_{A_1Q_1TA_2Q_2\hat{T}S\hat{S}} \right] \\
&= |T|^2 \tr\left[ \left( V_{A_1A_2Q_1Q_2}\otimes \Phi_{T\hat{T}|S\hat{S}}  \right) \left( \sum_i p_i \,\sigma_{A_1Q_1T}^i \otimes \omega_{A_2Q_2\hat{T}}^i \otimes \tau_{S\hat{S}}^i \right) \right] \\
&\quad + |T|^2 \tr\left[ \left( V_{A_1A_2Q_1Q_2}\otimes \Phi_{T\hat{T}|S\hat{S}}  \right) \left( \rho_{A_1Q_1TA_2Q_2\hat{T}S\hat{S}} - \sum_i p_i \,\sigma_{A_1Q_1T}^i \otimes \omega_{A_2Q_2\hat{T}}^i \otimes \tau_{S\hat{S}}^i  \right) \right] \\
&\leq \omega_{Q(T)}(V,\pi) + |T|^2 \tr\left[ \left( V_{A_1A_2Q_1Q_2}\otimes \Phi_{T\hat{T}|S\hat{S}} \right) \left( \rho_{A_1Q_1TA_2Q_2\hat{T}S\hat{S}} - \sum_i p_i \,\sigma_{A_1Q_1T}^i \otimes \omega_{A_2Q_2\hat{T}}^i \otimes \tau_{S\hat{S}}^i  \right) \right],
\end{split}
\end{equation}
where $\sum_i p_i \,\sigma_{A_1Q_1T}^i \otimes \omega_{A_2Q_2\hat{T}}^i \otimes \tau_{S\hat{S}}^i$ is one of the close separable states to $\rho_{A_1Q_1TA_2Q_2\hat{T}S\hat{S}}$ specified by Theorem~\ref{thm:TriQdeFTwithPTLC}.
As $\operatorname{sdp}_{n_1,n_2}(V,\pi,T)$ is an upper bound for $\omega_{Q(T)}(V,\pi)$ we obtain
\begin{equation}
\begin{split}
\Big|&\sdp_{n_1,n_2}(V,\pi,T) - \omega_{Q(T)}(V,\pi)\Big|\\
& \leq |T|^2 \left\vert \tr\left[ \left( V_{A_1A_2Q_1Q_2}\otimes \Phi_{T\hat{T}|S\hat{S}}  \right) \left( \rho_{A_1Q_1TA_2Q_2\hat{T}S\hat{S}} - \sum_i p_i \,\sigma_{A_1Q_1T}^i \otimes \omega_{A_2Q_2\hat{T}}^i \otimes \tau_{S\hat{S}}^i  \right) \right] \right\vert \\
&\leq |T|^2 \norm{V_{A_1A_2Q_1Q_2}\otimes \Phi_{T\hat{T}|S\hat{S}}}_{\infty} \norm{\rho_{A_1Q_1TA_2Q_2\hat{T}S\hat{S}} - \sum_i p_i \,\sigma_{A_1Q_1T}^i \otimes \omega_{A_2Q_2\hat{T}}^i \otimes \tau_{S\hat{S}}^i }_1 \\
& \qquad \qquad \qquad \qquad \qquad \qquad \qquad \qquad \qquad \qquad \qquad \qquad \qquad \qquad \qquad  (\because\text{ H\"older's inequality}) \\
&= |T|^2 \norm{V_{A_1A_2Q_1Q_2}}_{\infty} \norm{\Phi_{T\hat{T}|S\hat{S}}}_{\infty} \norm{\rho_{A_1Q_1TA_2Q_2\hat{T}S\hat{S}} - \sum_i p_i \,\sigma_{A_1Q_1T}^i \otimes \omega_{A_2Q_2\hat{T}}^i \otimes \tau_{S\hat{S}}^i }_1 \\
&= |T|^4 \norm{\rho_{A_1Q_1TA_2Q_2\hat{T}S\hat{S}} - \sum_i p_i \,\sigma_{A_1Q_1T}^i \otimes \omega_{A_2Q_2\hat{T}}^i \otimes \tau_{S\hat{S}}^i }_1 \qquad \left( \because \norm{V_{A_1A_2Q_1Q_2}}_{\infty} = 1, \; \norm{\Phi_{T\hat{T}|S\hat{S}}}_{\infty} = |T|^2 \right) \\ 
&\leq |T|^4 \left[18^{3/2}|T|^2 \left(\sqrt{2\ln2}\right) \left(\sqrt{\frac{ \left( \log |A| + 8\log |S\hat{S}| \right)}{n_2} + \frac{ \log |A|}{n_1}}\right)\right] \qquad (\because\text{ Theorem~\ref{thm:TriQdeFTwithPTLC}})\\
&= 18^{3/2} |T|^6  \left(\sqrt{2\ln2}\right) \left(\sqrt{\frac{ \left( \log|A| + 16\log|T| \right)}{n_2} + \frac{ \log|A|}{n_1}}\right).
\end{split}
\end{equation}
Here, we set $A=A_1$, $\tilde{A}=Q_1T$ , $B=S\hat{S}$, $C=A_2$, and $\tilde{C}=Q_2\hat{T}$ when we applied Theorem~\ref{thm:TriQdeFTwithPTLC}, and there is no $|AQ|$ contribution in the dimension factor coming from the optimal measurements with minimal distortion\,---\,i.e.~the factor $\min\left\{18^{3/2}\sqrt{|A\tilde{A}BC\tilde{C}|},4|BC\tilde{C}|\right\}$\,---\,since $A_1Q_1$ and $A_2Q_2$ are already classical systems; please recall the remark in \eqref{eq:dataHiding_with_classical} after Lemma~\ref{lem:DataHiding} from which the dimension factor originated. In both cases of Lemma~\ref{lem:DataHiding}, the dimension factor only comes from the measurements on the quantum systems.
\end{proof}

\begin{corollary}
Let $\sdp_{n_1, n_2}(V,\pi,T)$ be the $(n_1, n_2)$-th level relaxation for the $|T|$-dimensional two-player free game with rule matrix $V$ and probability distribution $\pi(q_1,q_2)=\pi_1(q_1)\pi_2(q_2)$. Then, we have
\begin{align}
\omega_{Q(T)}(V,\pi)=\lim_{n_1, n_2 \rightarrow\infty} \sdp_{n_1, n_2}(V,\pi,T).
\end{align}
\end{corollary}

Theorem~\ref{thm:Convergencen1n2} provides us the convergence speed of our SDP relaxations. For simplicity, let us assume $n_1=n_2=n$. To achieve a constant error $\epsilon$, we need to go up to the following level of the hierarchy:
\begin{equation}
\begin{split}
O\left( |T|^6\sqrt{\frac{\log|TA|}{n}} \right) \leq \epsilon \qquad \iff  \qquad n \geq O\left( |T|^{12}\frac{\log|TA|}{\epsilon^2} \right).
\end{split}
\end{equation}
When $n_1=n_2=\mathcal{O}\left( |T|^{12}\frac{\log|TA|}{\epsilon^2} \right)$, the size of the variables in the SDP is
\begin{align}\label{eq:SizeOfSDPforEpsilon}
\text{$(|A||Q||T|)^{\mathcal{O}\left( |T|^{12}\frac{\log|TA|}{\epsilon^2} \right)}$, which is $\exp\Big( \mathcal{O}\Big( \frac{|T|^{12}}{\epsilon^2} \left( \log^2|AT| + \log|Q|\log|TA| \right)\Big) \Big)$.}
\end{align}
This implies that approximating the quantum optimal winning probability of two-player free games for fixed dimension can be solved within the additive error $\epsilon>0$ in \emph{quasi-polynomial time} in terms of the sizes of answers and questions of the game. Note that this convergence rate is derived only from the linear constraints in $\sdp_{n_1,n_2}(V,\pi,T)$, and we used neither the PPT nor NPA constraints.


\section{Combining with the NPA hierarchy}\label{sec:NPA}

\subsection{Setting}

In this section, we discuss how to combine our $\sdp_{n_1,n_2}(V,\pi,T)$ in \eqref{eq:SDP+Conditions} with the NPA hierarchy \cite{navascues2008convergent,pironio2010convergent}. Combining these two hierarchies is advantageous because the resulting SDP relaxation always returns a bound which is equal or tighter than the ones obtained with the individual relaxations. We review the NPA hierarchy in Section~\ref{subsec:NPAhierarchy} and show how to combine the two hierarchies in Section~\ref{subsec:NPA+SDP}. 


\subsection{NPA hierarchy}\label{subsec:NPAhierarchy}

The NPA hierarchy gives necessary conditions satisfied by quantum correlations \cite{navascues2008convergent,pironio2010convergent}. As mentioned, a correlation between two parties can be represented as a conditional joint probability $p(a_1,a_2|q_1,q_2)$, where $q_1$ and $q_2$ are the labels of the measurements performed by Alice and Bob respectively, while $a_1$ and $a_2$ are the outcomes of such measurements. For a given choice of $q_1$ and $q_2$, the correlation $p(a_1,a_2|q_1,q_2)$ is classified as quantum if there exist a quantum state $\rho_{AB}$ and local measurements $\left\{ E(a_1,q_1) = \tilde{E}(a_1,q_1)_A\otimes\id_B \right\}_{a_1}$ and $\left\{ E(a_2,q_2) = \id_A\otimes \tilde{E}(a_2,q_2)_B \right\}_{a_2}$ such that
\begin{align}
p(a_1,a_2|q_1,q_2) = \tr\left[ E(a_1,q_1)E(a_2,q_2)\rho_{AB} \right].
\end{align}
In the NPA hierarchy, we can always assume without loss of generality that the measurements $\{E(a_i,q_i)\}$ are composed by orthogonal projectors as the dimension of the system is unbounded. However, this is no longer true when we combine the NPA hierarchy with our SDP relaxations, which have the dimension restriction. In the next section, we will clarify how we can take this into account when we write the NPA constraints for the combined hierarchy. For now, let us assume the measurements are composed of orthogonal projectors and satisfy the following properties:
\begin{flalign}
&\text{(i) hermiticity: } E(a_i,q_i)^{\dagger}=E(a_i,q_i)&\\
&\text{(ii) orthogonality: } E(a_i,q_i)E(\bar{a}_i,\bar{q}_i) = \delta_{a_i,\bar{a}_i}E(a_i,q_i) \;\text{ if }\; q_i=\bar{q}_i\\
&\text{(iii) completeness: } \sum_{a_i} E(a_i,q_i) = \id_{AB}\;\;\forall q_i\\
&\text{(iv) commutativity: } [E(a_1,q_1),E(a_2,q_2)]=0\;\;\forall a_1,a_2,q_1,q_2. 
\end{flalign}
For simplicity of notation, let us define the set $\boldsymbol{\alpha}:=\{(a_1,q_1)\}$ and $\boldsymbol{\beta}\equiv \{(a_2,q_2)\}$, and denote $P_{\alpha\beta}:=p(a_1,a_2|q_1,q_2)$, where $\alpha=(a_1,q_1)\in\boldsymbol{\alpha}$ and $\beta=(a_2,q_2)\in\boldsymbol{\beta}$. The set $\boldsymbol{\alpha}$ has $m_1 = |A_1||Q_1|$ elements, and the set $\boldsymbol{\beta}$ has $m_2=|A_2||Q_2|$ elements. 

Let us assume that a correlation $P_{\alpha\beta}$ is a quantum correlation which has a quantum state $\rho_{AB}$ and measurement operators $\{E(a_i,q_i)\}$. By taking products of $E(a_i,q_i)$ or linear combinations of such products, we can construct a set of $n$ operators, $\mathcal{S}=\{S_1,\cdots ,S_n\}$. For each set $\mathcal{S}$, we can construct an $n\times n$ matrix $\Gamma$ of the form
\begin{align}\label{eq:NPAGamma}
\Gamma_{ij} = \tr\left[ S_i^{\dagger}S_j \rho_{AB} \right].
\end{align}
From the construction, $\Gamma$ is a Hermitian matrix which satisfies that
\begin{align}\label{eq:NPAcondition1}
\text{if } \; \sum_{i,j}c_{ij}S_i^{\dagger}S_j = 0\quad\text{then}\quad\sum_{i,j} c_{ij}\Gamma_{ij}=0,
\end{align}
as well as that
\begin{align}\label{eq:NPAcondition2}
\text{if } \; \sum_{i,j} c_{ij}S_i^{\dagger}S_j = \sum_{\alpha,\beta}d_{\alpha\beta}E(\alpha)E(\beta)\quad\text{then}\quad\sum_{i,j} c_{ij}\Gamma_{ij} = \sum_{\alpha,\beta}d_{\alpha\beta}P_{\alpha\beta},
\end{align}
where $\alpha\in\boldsymbol{\alpha}$ and $\beta\in\boldsymbol{\beta}$. Here, the coefficients $c_{ij}$ and $d_{\alpha\beta}$ are determined by the set $\mathcal{S}$. Moreover, we have 
\begin{align}\label{eq:NPAconditionPositiveSD}
\Gamma \geq 0.
\end{align}
As a result, we get the following necessary conditions for quantum correlations.

\begin{lemma}\label{lem:NPA}
\cite{navascues2007bounding,navascues2008convergent} For a given set $\mathcal{S}=\{S_1,\cdots ,S_n\}$ constructed by taking products or linear combinations of $\{E(a_i,q_i)\}$ satisfying the conditions (i)-(iv), a necessary condition for a correlation $p(a_1,a_2|q_1,q_2)$ to be quantum is that there exists a Hermitian $n\times n$ matrix, $\Gamma$, which satisfies the conditions \eqref{eq:NPAcondition1}, \eqref{eq:NPAcondition2}, and \eqref{eq:NPAconditionPositiveSD}.
\end{lemma}

The simplest example is when we choose $\mathcal{S}$ as $\{E(a_1,q_1)\}\cup\{E(a_2,q_2)\}=\{E(a_i,q_i)\}$. For any given correlation $p(a_1,a_2|q_1,q_2)$, the corresponding $\tilde{\Gamma}$ for $\mathcal{S}=\{E(a_i,q_i)\}$ constructed from \eqref{eq:NPAGamma} takes the form
\begin{align}\label{eq:NPA1Gamma}
\tilde{\Gamma} = \begin{pmatrix} Q&P\\P^T&R \end{pmatrix},
\end{align}
where the $m_1\times m_2$ sub-matrix $P$ has the form
\begin{align}\label{eq:NPA1Gamma_P_conditions}
P_{\alpha\beta} = p(a_1,a_2|q_1,q_2) = \tr\left[E(a_1,q_1)E(a_2,q_2)\rho_{AB}\right], \qquad \alpha=(a_1,q_1) , \beta=(a_2,q_2),
\end{align}
and the $m_1\times m_1$ sub-matrix $Q$ and $m_2\times m_2$ sub-matrix $R$ satisfy
\begin{equation}\label{eq:NPA1Gamma_QR_conditions}
\begin{split}
&Q_{\alpha\alpha} = p(a_1|q_1) = \tr\left[E(a_1,q_1)\rho_{AB}\right] , \qquad \alpha = (a_1,q_1)\in\boldsymbol{\alpha} \\
&Q_{\alpha\alpha'} = 0 , \hspace*{47mm} \alpha = (a_1,q_1),\, \alpha' = (a'_1,q'_1) \in\boldsymbol{\alpha} \qquad \text{for} \ q_1 = q'_1 \\
&R_{\beta\beta} = p(a_2|q_2) = \tr\left[E(a_2,q_2)\rho_{AB}\right], \qquad \beta = (a_2,q_2)\in\boldsymbol{\beta} \\
&R_{\beta\beta'} = 0, \hspace*{47mm} \beta = (a_2,q_2),\, \beta' = (a'_2,q'_2) \in\boldsymbol{\beta} \qquad \text{for} \ q_2=q'_2,
\end{split}
\end{equation}
where $p(a_1|q_1) = \sum_{a_2}p(a_1,a_2|q_1,q_2)$ and $p(a_2|q_2) = \sum_{a_1}p(a_1,a_2|q_1,q_2)$ are the marginal probabilities.\footnote{It is worth noting that we have used the assumption that the measurements are composed by orthogonal projectors to set some of the entries of the $Q$ and $R$ matrices to zero. When this assumption cannot be made, for example when we combine this with our SDP relaxations, one should replace these entries with new variables of the problem.} Note that we are suppressing the conditioning arguments $q_2$ and $q_1$, since the distributions under consideration are non-signalling. With this construction, we can specify all entries of $\tilde{\Gamma}$ except the entries of $Q_{\alpha\alpha'}$ when $q_1 \neq q'_1$, and $R_{\beta\beta'}$ when $q_2 \neq q'_2$. The matrix $\tilde{\Gamma}$ automatically satisfies the conditions \eqref{eq:NPAcondition1} and \eqref{eq:NPAcondition2} due to \eqref{eq:NPA1Gamma_QR_conditions}. The only remaining check is whether we can find a matrix of the form $\tilde{\Gamma}$ in \eqref{eq:NPA1Gamma} which is positive semidefinite by giving appropriate values to the undecided entries of $\tilde{\Gamma}$. In other words, for a given correlation $p(a_1,a_2|q_1,q_2)$, we need to decide whether we can find $\tilde{\Gamma}$ with the form \eqref{eq:NPA1Gamma} such that
\begin{align}\label{eq:NPA1condition_PosSem}
\tilde{\Gamma} \geq 0.
\end{align}
As different sets $\mathcal{S}$ provide different necessary conditions for quantum correlations, we can build a hierarchy of necessary conditions by using a sequence of independent sets $\left\{\mathcal{S}_k\right\}_k$, where $\mathcal{S}_k \subset \mathcal{S}_{k+1}$ for all $k$. Conventionally, in the NPA hierarchy, the first level is chosen as $\mathcal{S}_1 = \mathcal{S}_0\cup\{E(a_i,q_i)\}$ where $\mathcal{S}_0 = \{\id\}$, the second level is $\mathcal{S}_2=\mathcal{S}_0\cup\mathcal{S}_1\cup\{E(a_i,q_i)E(a_j,q_j)\}$, the third level is $\mathcal{S}_3=\mathcal{S}_0\cup\mathcal{S}_1\cup\mathcal{S}_2\cup\{E(a_i,q_i)E(a_j,q_j)E(a_k,q_k)\}$, and so on. Therefore, the $k$th level NPA matrix $\Gamma_k$ will have the form 
\begin{align}\label{eq:NPAGamma_k}
\left( \Gamma_k\right)_{ij} = \tr\left[ S_i^{\dagger}S_j \rho_{AB} \right]
\end{align}
with the set $\mathcal{S}_k$. In fact, the completeness of the NPA hierarchy was proven in \cite{navascues2008convergent}. That is, a correlation which satisfies the $k$-th level NPA condition for all $k\geq1$ is a quantum correlation.


\subsection{NPA as constraints on $\sdp_{n_1,n_2}(V,\pi,T)$}\label{subsec:NPA+SDP}

In this section, we provide a procedure for expressing the entries of the NPA matrix $\Gamma_k$, given in \eqref{eq:NPAGamma_k}, in terms of the optimisation variable $\rho_{(A_1Q_1T)(A_2Q_2\hat{T})^{n_1}(S\hat{S})^{n_2}}$ in $\sdp_{n_1,n_2}$. As an example, we show how to write the first level of the NPA constraint and postpone the general case to Appendix \ref{app:npa_constraints}.

We begin by recalling that the optimal winning probability of quantum correlations can be written as,
\begin{align}\label{eq:QWinP_NPA}
\omega_{Q(T)}(V,\pi) = \max_{p\in\mathcal{Q}} \sum_{q_1,q_2} \pi(q_1,q_2) \sum_{a_1,a_2} V(a_1,a_2,q_1,q_2) p(a_1,a_2|q_1,q_2),
\end{align}
where $\mathcal{Q}$ is the set of all quantum correlations. If we compare this expression with the objective function of our $\sdp_{n_1,n_2}(V,\pi,T)$, under the assumption that $A_1A_2Q_1Q_2$ is a classical subsystem in \eqref{eq:SDP+Conditions_Classical}, we can derive the relation between $p(a_1,a_2|q_1,q_2)$ and the variable $\rho_{T\hat{T}^{n_1}(S\hat{S})^{n_2}}(a_1,a_2^{n_1},q_1,q_2^{n_1})$, which is
\begin{align}\label{eq:JointPasRho}
p(a_1,a_2|q_1,q_2) = \frac{|T|^2}{\pi(q_1,q_2)} \tr\left[ \Phi_{T\hat{T}|S\hat{S}} \;\rho_{T\hat{T}S\hat{S}}(a_1,a_2,q_1,q_2) \right],
\end{align}
where $\rho_{T\hat{T}S\hat{S}}(a_1,a_2,q_1,q_2)$ is the reduced state of $\rho_{T\hat{T}^{n_1}(S\hat{S})^{n_2}}(a_1,a_2^{n_1},q_1,q_2^{n_1})$:
\begin{align}
\rho_{T\hat{T}S\hat{S}}(a_1,a_2,q_1,q_2) = \sum_{a_2^{n_1-1}, q_2^{n_1-1}} \tr_{\hat{T}^{n_1-1}(S\hat{S})^{n_2-1}}\left[ \rho_{T\hat{T}^{n_1}(S\hat{S})^{n_2}}(a_1,a_2^{n_1},q_1,q_2^{n_1}) \right].
\end{align}
The marginal probabilities for Alice are given by
\begin{equation}\label{eq:MarPasRho_A}
\begin{split}
p(a_1|q_1) &= \sum_{a_2} p(a_1,a_2|q_1,q_2) = \frac{|T|^2}{\pi(q_1,q_2)} \tr\left[ \Phi_{T\hat{T}|S\hat{S}} \; \left( \sum_{a_2} \rho_{T\hat{T}S\hat{S}}(a_1,a_2,q_1,q_2) \right) \right] \\
&= \frac{|T|}{\pi_1(q_1)} \tr\left[ \Phi_{T\hat{T}|S\hat{S}} \; \left( \id_{\hat{T}}\otimes\rho_{TS\hat{S}}(a_1,q_1)\right) \right],
\end{split}
\end{equation}
where we used the linear constraints of the program and the fact that we consider free games. Similarly, the marginal probability distribution for Bob is given by,
\begin{align}\label{eq:MarPasRho_B}
p(a_2|q_2)=\frac{|T|}{\pi_2(q_2)} \tr\left[ \Phi_{T\hat{T}|S\hat{S}} \; \left( \id_{T}\otimes\rho_{\hat{T}S\hat{S}}(a_2,q_2)\right)\right].
\end{align}
Then, using the joint and marginal probabilities we can construct the first level NPA matrix $\Gamma_1$, which is $\tilde{\Gamma}$ in \eqref{eq:NPA1Gamma} combined with the row and column corresponding to the base set $\mathcal{S}_0 = \{\id\}$. The sole entries of $\Gamma_1$ which we cannot relate to the optimisation variable are those in the sub-matrices $Q$ and $R$ with different inputs $q$; these entries become new variables of the problem. To impose the NPA constraint on our $\sdp_{n_1,n_2}(V,\pi,T)$, we can now simply add the constraint $\Gamma_1(\rho_{(A_1Q_1T)(A_2Q_2\hat{T})^{n_1}(S\hat{S})^{n_2}}) \geq 0$ to the SDP.

For any $k \geq 2$, the $k$-th level of the NPA hierarchy involves a matrix $\Gamma_k$ where some of the entries are of the form $p(a_1^m,a_2^{\ell}|q_1^m,q_2^{\ell})$ where $m$ and $\ell$ are such that $m + \ell \leq 2k$. Among these entries, the ones we can express in terms of the optimization variable are those with $m=1$ (since we are not extending the sub-system $A_1Q_1T$) and $\ell \leq k$ (as long as $k \leq n_1$). The relation between these entries is
\begin{align}
p(a_1, a_2^{\ell}|q_1, q_2^{\ell})=\frac{|T|^{\ell+1}}{\pi_1(q_1) \, \prod_{j=1}^{\ell} \pi_2(q_2^{(j)})}
\Trt{\left( P_{\text{cyclic}}^{\ell+1} \right)^{T_{S\hat{S}}} \,\left( \bigotimes_{i=2}^{\ell} \id_{T_i} \right) \otimes\rho_{T\hat{T}^{\ell}S\hat{S}}(a_1, a_2^{\ell}, q_1^{m}, q_2^{\ell})} \quad \forall\ell \ : \ \ell \leq k,
\end{align}
where $P_{\text{cyclic}}^{n}$ is a unitary operator acting on $n$ copies of the subsystem $T\hat{T}$, implementing a cyclic permutation over these copies described by the following action over tuples, $\left( 1, 2 , 3 , \ldots, n \right) \rightarrow \left( 2, 3, \ldots, n, 1 \right)$. The above equation is explicitly derived in Appendix \ref{app:npa_constraints}. Then, to impose the $k$-th level NPA constraint, we can simply add the positive semi-definite condition $\Gamma_k \geq 0$ to $\sdp_{n_1,n_2}(V,\pi,T)$.


\section{Conclusions}\label{sec:discussion}

We explored the characterisation of quantum correlation with dimension constraints. More specifically, we gave a converging hierarchy of SDP relaxations for the set of quantum correlations with fixed dimension and provided analytical bounds on the convergence rate by means of multipartite quantum de Finetti theorems with linear constraints. We conclude with a few remarks about our results.

\begin{itemize}
\item \textbf{Quantum information applications:} Our approach can be applied to other problems of interest. The most obvious example would be the multipartite extension of two-player games.
In addition, as pointed out in \cite{berta2018semidefinite}, finding the maximum success probability for transmitting a message under a given noisy channel can be formalised in terms of a quantum separability problem as well. Thus, our techniques can be applied to the multipartite generalisation of this quantum error correction problem, and the particularly promising point is the possibility to add, even in this setting, NPA-type constraints.

\item \textbf{Dimension reduction using symmetry:} Compared to a given level of the NPA hierarchy, the corresponding level of $\sdp_{n_1,n_2}(V,\pi,T)$ has an optimization variable with a bigger size, and a higher number of constraints.
One way to improve this aspect is to use the symmetry embedded in the semidefinite program. Our SDP relaxations have a few symmetries: (i) $\Phi_{T\hat{T}S\hat{S}}$ in the objective function is invariant under any local unitary transformation, and (ii) depending on the game, the rule matrix $V(a_1,a_2,q_1,q_2)$ is also invariant under some actions. We refer to \cite{rosset2018symdpoly} for an example of a symmetry-finding program that can potentially reduce the program size in numerical implementations.

\item \textbf{Improving on $|T|$-dependence:} When we set $|T|=1$, our relaxations become exactly equivalent to the linear program relaxations given in \cite{brandao2017quantum}
\begin{align}
\operatorname{lp}_n(V,\pi,1) := \max_{p} \sum_{a,b,x,y} V(a,b,x,y) \,p(a,b,x,y)
\end{align}
\begin{equation}
\begin{split}
\text{s.t.} \;\; & p(a,b^{n},x,y^{n}) \geq 0 \quad \forall\, a,b^{n},x,y^{n} \, , \sum_{a,b^{n},x,y^{n}} p(a,b^{n},x,y^{n}) = 1,\, \text{perm. inv. on } b^{n}, y^{n} \text{ wrt } a^{n} \text{ and } x^{n}\\
&\sum_{a} p(a,b^{n},x,y^{n}) = \pi_1(x) p(b^{n},y^{n}),\,\sum_{b} p(a,b^{n},x,y^{n}) = \pi_2(y) p(a,b^{n-1},x,y^{n-1}) .
\end{split}
\end{equation}
As there are also matching strong hardness results \cite{aaronson2014multiple}, the complexity derived in \eqref{eq:SizeOfSDPforEpsilon} seems near optimum with respect to $|Q|$ and $|A|$. However, there could very well exist more efficient approximation algorithms in terms of the $|T|$-dependence. One might for example explore $\epsilon$-net based methods, as in \cite{shi2015epsilon,harrow15}.

\item \textbf{General games:} In the main text, we assume that the given probability distribution of the questions for Alice and Bob is not correlated, i.e.~$\pi(q_1,q_2)=\pi_1(q_1)\pi_2(q_2)$, which corresponds to free games. However, we can also derive lower bounds on the computational complexity for general games, when $\pi(q_1,q_2)\neq\pi_1(q_1)\pi_2(q_2)$. The key idea is to absorb $\pi(q_1,q_2)$ into the rule matrix $V(a_1,a_2,q_1,q_2)$ instead of $E_{A_1Q_1T}$ and $D_{A_2Q_2\hat{T}}$ when we convert the problem into the quantum separability problem. We refer to Appendix~\ref{app:GeneralGames} for the derivation, where we find that in contrast to free games, our bounds become exponential in terms of $|Q|$.
\end{itemize}




\begin{appendices}

\section{Distortion relative to quantum side information}\label{app:distortion_sideI}

Here, we prove the second part of Lemma~\ref{lem:DataHiding} which states that for a traceless Hermitian operator $\gamma_{AB}$ on $\hil_{AB}$, there exists a measurement $\mathcal{M}_B$ on $\mathcal{H}_B$ with at most $|B|^6$ outcomes such that  $\norm{\left(\mathcal{I}_A\otimes\mathcal{M}_B\right) \left(\gamma_{AB}\right)}_1 \geq \frac{1}{2|B|} \norm{\gamma_{AB}}_1.$ The proof is inspired from \cite[Theorem 16]{lami2018ultimate}.

\begin{proof}[Proof of Lemma~\ref{lem:DataHiding}]
Let us start with the maximally entangled state
\begin{align}
\text{$\Phi_{A'|B'} =\ket{\Phi}\!\!\bra{\Phi}_{A'|B'}$ where $\ket{\Phi}_{A'|B'} = \frac{1}{|A'||B'|} \sum_i \ket{i}_{A'}\ket{i}_{B'}$, and $|A'|=|B'|$.}
\end{align}
We can create a separable state $\omega_{A'B'}$ by mixing $\Phi_{A'|B'}$ with another separable state $\sigma_{A'B'} = \frac{\id_{A'B'}-\Phi_{A'|B'}}{|B'|^2-1}$ as
\begin{align}
\omega_{A'B'} = \frac{1}{|B'|}\Phi_{A'B'} + \frac{|B'|-1}{|B'|}\sigma_{A'B'} \; \in \;\operatorname{SEP(A':B')},
\end{align}
where $\operatorname{SEP(A':B')}$ denotes the set of separable states with respect to the bipartition $A'|B'$. Hence, we can write $\omega_{A'B'} = \sum_{i} p_i \omega_{A'}^i\otimes\omega_{B'}^i$ for some probability distribution $\{p_i\}_i$ and states $\{\omega_{A'}^i\}_i$ and $\{\omega_{B'}^i\}_i$ with at most $|A'B'|^2$ elements \cite{Horodecki97}. Next, we define a measurement $\mathcal{M}_B$ with operators $\{\tilde{M}_B(i,k)\}_{i,k}$, as well as a set of measurements $\{\mathcal{M}_A^{i,k}\}_{i,k}$ with operators $\{\tilde{M}_A^{i,k}(j)\}_j$ as
\begin{align}\label{eq:DataHiding_M_B}
\tilde{M}_B(i,k) &= \tr_{B'}\left[ p_i U_B^{\dagger}(k) \sqrt{\omega_{B'}^i}\Phi_{BB'}\sqrt{\omega_{B'}^i}U_B(k) \right]\quad\text{and}\\
\tilde{M}_A^{i,k}(j) &= \tr_{A'}\left[ \sqrt{\omega_{A'}^i}U_{A'}^{\dagger}(k)N_{AA'}(j) U_{A'}(k) \sqrt{\omega_{A'}^i} \right],\label{eq:DataHiding_M_A}
\end{align}
where $U(k)$ denote generalised Pauli operators, $\omega^i_{A'}$ and $\omega^i_{B'}$ are the elements of the decomposition of $\omega_{A'B'}$, and $\{N_{AA'}(j)\}_j$ are measurement operators defined later. We can check that both definitions indeed correspond to valid measurements:
\begin{align}
\text{$\sum_{i,k}\tilde{M}_B(i,k) = \id_B$, $\sum_j \tilde{M}_A^{i,k}(j) = \id_A$, and $\tilde{M}_B(i,k), \tilde{M}_A^{i,k}(j) \geq 0 \;\; \forall i,k,j$.}
\end{align}
The goal is to show that $\mathcal{M}_B$ defined in \eqref{eq:DataHiding_M_B} gives rise to \eqref{eq:DataHiding_Side}. Before showing that, however, it is helpful to understand where these measurements came from. They are related to the quantum teleportation protocol from \cite{bennett1993teleporting}. Without loss of generality, let us assume that $|A|\geq|B|=|A'|=|B'|$. Then, the quantum teleportation protocol from $B$ to $A$ is a quantum channel defined as \cite{bennett1993teleporting}
\begin{align}
\tau_{ABA'B'\rightarrow AA'}(\cdot) = \sum_{k=1}^{|B|^2} U_{A'}(k)\tr_{BB'}\left[ (\cdot) \left( \id_{AA'}\otimes U_B(k)\Phi_{BB'}U_B^{\dagger}(k)\right) \right] U_{A'}^{\dagger}(k).
\end{align}
For a traceless Hermitian operator $\gamma_{AB}$, we then consider
\begin{align}
\norm{\tau_{ABA'B'\rightarrow AA'}\left(\gamma_{AB}\otimes\omega_{A'B'}\right)}_1 = \sum_j \left\vert \tr\left[ N_{AA'}(j) \left(\tau_{ABA'B'\rightarrow AA'}\left(\gamma_{AB}\otimes\omega_{A'B'}\right)\right) \right] \right\vert,
\end{align}
where we used the expression $\norm{X_A}_1 = \max_{\{M_A(i)\}_i} \sum_i \left\vert \tr\left[M_A(i)X_A\right]\right\vert$ for the trace norm with corresponding arg max $\{N_{AA'}(j)\}_j$ to be used in \eqref{eq:DataHiding_M_A}. We have
\begin{align}
&\norm{\tau_{ABA'B'\rightarrow AA'}\left(\gamma_{AB}\otimes\omega_{A'B'}\right)}_1 \\
&= \sum_j \left\vert \sum_k \tr\left[ N_{AA'}(j) \left( U_{A'}(k)\tr_{BB'}\left[ \left(\gamma_{AB}\otimes\omega_{A'B'}\right) \left( \id_{AA'}\otimes U_B(k)\Phi_{BB'}U_B^{\dagger}(k)\right) \right] U_{A'}^{\dagger}(k) \right) \right]  \right\vert\\
&= \sum_j \left\vert \sum_k \tr\left[ \left( U_{A'}^{\dagger}(k)N_{AA'}(j)U_{A'}(k) \otimes \id_{BB'}  \right) \left( \left(\gamma_{AB}\otimes\omega_{A'B'}\right) \left( \id_{AA'}\otimes U_B(k)\Phi_{BB'}U_B^{\dagger}(k)\right) \right) \right]  \right\vert \\
&= \sum_j \left\vert \sum_k \tr\left[ \left( U_{A'}^{\dagger}(k)N_{AA'}(j)U_{A'}(k) \otimes U^{\dagger}_B(k)\Phi_{BB'}U_B(k)  \right)  \left(\gamma_{AB}\otimes\left(\sum_{i} p_i \omega_{A'}^i\otimes\omega_{B'}^i\right)\right)  \right]  \right\vert \\
&= \sum_j \left\vert \sum_{i,k} \tr\left[ \left(\left( \sqrt{\omega^i_{A'}}U_{A'}^{\dagger}(k)N_{AA'}(j)U_{A'}(k)\sqrt{\omega^i_{A'}}\right) \otimes \left( p_iU^{\dagger}_B(k)\sqrt{\omega^i_{B'}}\Phi_{BB'}\sqrt{\omega^i_{B'}}U_B(k)  \right)\right)  \left(\gamma_{AB}\otimes \id_{A'B'}\right)  \right]  \right\vert \\
&= \sum_j \left\vert \sum_{i,k} \tr\left[ \gamma_{AB} \left(\tilde{M}_A^{i,k}(j) \otimes \tilde{M}_B(i,k) \right) \right] \right\vert. \label{eq:Data_Hiding_eq1}
\end{align}
The measurement $\mathcal{M}_B$ defined in \eqref{eq:DataHiding_M_B} now gives rise to
\begin{align}
&\norm{\left(\mathcal{I}_A\otimes \mathcal{M}_B\right) \left(\gamma_{AB}\right)}_1\\
&= \;\sum_{i,k} \norm{\tr_B\left[ \left(\id_A\otimes\tilde{M}_B(i,k)\right) \gamma_{AB}\right]}_1\\
&= \;\sum_{i,k} \max_{\{M_A^{i,k}(j)\}_j} \sum_j \left\vert \tr\left[ \left( M_A^{i,k}(j) \otimes \tilde{M}_B(i,k)\right)\gamma_{AB} \right] \right\vert\\
&\geq \;\sum_{i,k} \sum_j \left\vert \tr\left[ \left( \tilde{M}_A^{i,k}(j) \otimes \tilde{M}_B(i,k)\right)\gamma_{AB} \right] \right\vert\\
&\geq \;\sum_j \left\vert \sum_{i,k} \tr\left[ \left( \tilde{M}_A^{i,k}(j) \otimes \tilde{M}_B(i,k)\right)\gamma_{AB} \right] \right\vert \label{eq:DataHiding_eq2}\\
&= \norm{\tau_{ABA'B'\rightarrow AA'}\left(\gamma_{AB}\otimes\omega_{A'B'}\right)}_1 \hspace*{20mm} (\text{by }\; \eqref{eq:Data_Hiding_eq1}) \\
&= \norm{\tau_{ABA'B'\rightarrow AA'}\left(\gamma_{AB}\otimes\left( \frac{1}{|B|}\Phi_{A'B'} + \frac{|B|-1}{|B|}\sigma_{A'B'} \right)\right)}_1 \\
&= \norm{\frac{1}{|B|}\tau_{ABA'B'\rightarrow AA'}\left(\gamma_{AB}\otimes \Phi_{A'B'} \right) + \frac{|B|-1}{|B|}\tau_{ABA'B'\rightarrow AA'}\left(\gamma_{AB}\otimes\sigma_{A'B'} \right)}_1 \\
&\geq \frac{1}{|B|} \norm{\tau_{ABA'B'\rightarrow AA'}\left(\gamma_{AB}\otimes \Phi_{A'B'} \right)}_1 - \norm{\tau_{ABA'B'\rightarrow AA'}\left(\gamma_{AB}\otimes\left(\frac{|B|-1}{|B|}\sigma_{A'B'}\right) \right)}_1, \label{eq:DataHiding_eq3}
\end{align}
where in the third line we substituted the measurement operators $\{\tilde{M}_A^{i,k}(j)\}_j$ instead of the maximisation, and in the last line we used the reverse triangular inequality.
Note that the first term in the last line is equivalent to $\norm{\gamma_{AB}}_1$ since $\Phi_{A'B'}$ is the maximally entangled state. Let us investigate the second term more closely. We have the chain of elementary implications
\begin{align}
&\frac{|B|-1}{|B|}\sigma_{A'B'} \leq \frac{|B|-1}{|B|}\sigma_{A'B'} + \frac{1}{|B|}\Phi_{A'B'} = \omega_{A'B'} \\
&\Rightarrow \; \gamma_{AB}\otimes\frac{|B|-1}{|B|}\sigma_{A'B'} \leq \gamma_{AB}\otimes\omega_{A'B'} \\
&\Rightarrow \;\norm{\tau_{ABA'B'\rightarrow AA'}\left(\gamma_{AB}\otimes\left(\frac{|B|-1}{|B|}\sigma_{A'B'}\right) \right)}_1 \leq \norm{\tau_{ABA'B'\rightarrow AA'}\left(\gamma_{AB}\otimes\omega_{A'B'} \right)}_1 \\
&\Rightarrow \;\norm{\tau_{ABA'B'\rightarrow AA'}\left(\gamma_{AB}\otimes\left(\frac{|B|-1}{|B|}\sigma_{A'B'}\right) \right)}_1 \leq \sum_j \left\vert \sum_{i,k} \tr\left[ \gamma_{AB} \left(\tilde{M}_A^{i,k}(j) \otimes \tilde{M}_B(i,k) \right) \right] \right\vert \qquad (\text{by }\; \eqref{eq:Data_Hiding_eq1})\\
&\Rightarrow \;\norm{\tau_{ABA'B'\rightarrow AA'}\left(\gamma_{AB}\otimes\left(\frac{|B|-1}{|B|}\sigma_{A'B'}\right) \right)}_1 \leq \norm{\left(\mathcal{I}_A\otimes \mathcal{M}_B\right) \left(\gamma_{AB}\right)}_1 \qquad (\text{by }\; \eqref{eq:DataHiding_eq2})
\end{align}
and substituting this into \eqref{eq:DataHiding_eq3} yields the claim
\begin{align}
\norm{\left(\mathcal{I}_A\otimes \mathcal{M}_B\right) \left(\gamma_{AB}\right)}_1 \geq \frac{1}{|B|}\norm{\gamma_{AB}}_1 - \norm{\left(\mathcal{I}_A\otimes \mathcal{M}_B\right) \left(\gamma_{AB}\right)}_1.
\end{align}
It remains to quantify the number of measurement outcomes of $\mathcal{M}_B$ with operators $\{\tilde{M}_B(i,k)\}_{i,k}$ defined in \eqref{eq:DataHiding_M_A}. The index $i$ came from the number of elements in the separable state $\omega_{A'B'}$, which is at most $|A'B'|^2 = |B|^4$, and the index $k$ came from the number of generalised Pauli operators, which is $|B|^2$. Therefore, the number of outcomes is at most $|B|^6$.
\end{proof}


\section{Tri- and multipartite quantum de Finetti theorems with linear constraints}\label{app:MulQdeFTwithLC}

Here, we prove multipartite quantum de Finetti theorems with general linear constraints. The proof is similar as for the partial trace constraints in Section~\ref{subsec:TriQdeF}. Let us recall some useful relations of the multipartite conditional quantum mutual information:
\begin{itemize}
\item Definition
\begin{align}
I(A_1:A_2:\cdots :A_k|B)_{\rho}:= \sum_{i=1}^k S(A_iB) - S(A_1A_2\cdots A_kB)-S(B)
\end{align}
\item Reduction multipartite to bipartite
\begin{align}\label{eq:MulttoBip2}
I(A_1:\cdots :A_k|R)_{\rho} = I(A_1:A_2|R) + I(A_1A_2:A_3|R) + \cdots  + I(A_1\cdots A_{k-1}:A_k|R)
\end{align}
\item Chain rule
\begin{align}\label{eq:ChainRule2}
I(A:BX|Z) = I(A:X|Z) + I(A:B|XZ).
\end{align}
\end{itemize}

Let us first find general bounds on the conditional quantum mutual information and multipartite conditional quantum mutual information, which are the general versions of Lemma~\ref{lem:BoundCondMIwithLC} and Lemma~\ref{lem:BoundTripCondMIwithLC}.

\begin{lemma}\label{lem:CMIbound}
For a quantum state $\rho_{AZX}$ classical on the $Z$- and $X$-system, it holds that
\begin{align}
I(A:Z|X) \leq \log|A|.
\end{align}
\end{lemma}
\begin{proof}
For $X$ classical, the conditional quantum mutual information can be written as $I(A:Z|X) = \sum_x p_x I(A:Z)_{\rho_{AZ|x}}$. We know that $I(A:Z)\leq\log|A|$ when $Z$ is classical, and hence
\begin{align}
I(A:Z|X) = \sum_x p_x I(A:Z)_{\rho_{AZ|x}}\leq\log|A|.
\end{align}
\end{proof}

\begin{lemma}\label{lem:mulCMIbound}
Consider a quantum state $\rho_{AZ_1^{n_1}W_1^{n_2}}$ classical on the $Z$- and $W$-systems. Then, there exist $0\leq\bar{m}<n_1$ and $0\leq\bar{l}<n_2$ such that
\begin{align}
I(A:Z_{\bar{m}+1}:W_{\bar{l}+1}|Z_1^{\bar{m}}W_1^{\bar{l}}) \leq \frac{\log |A|}{n_1} + \frac{\log |A| + \log |Z|}{n_2}.
\end{align}
Moreover, by Pinsker's inequality \eqref{eq:PinskerInq} this implies that
\begin{align}
\mathbb{E}_{z_1^{\bar{m}}w_1^{\bar{l}}}\left\{ \norm{ \rho_{AZ_{\bar{m}+1}W_{\bar{l}+1}|z_1^{\bar{m}}w_1^{\bar{l}}} - \rho_{A|z_1^{\bar{m}}w_1^{\bar{l}}}\otimes\rho_{Z_{\bar{m}+1}|z_1^{\bar{m}}w_1^{\bar{l}}}\otimes\rho_{W_{\bar{l}+1}|z_1^{\bar{m}}w_1^{\bar{l}}} }_1^2 \right\}\notag
\leq 2\ln2 \left(\frac{\log |A|}{n_1} + \frac{\log |A| + \log |Z|}{n_2}\right).
\end{align}
\end{lemma}

\begin{proof}
The multipartite quantum mutual information $I(A:Z_{\bar{m}+1}:W_{\bar{l}+1}|Z_1^{\bar{m}}W_1^{\bar{l}})$ can be expressed in terms of bipartite ones using \eqref{eq:MulttoBip2}:
\begin{align}
I(A:Z_{\bar{m}+1}:W_{\bar{l}+1}|Z_1^{\bar{m}}W_1^{\bar{l}}) = I(A:Z_{\bar{m}+1}|Z_1^{\bar{m}}W_1^{\bar{l}}) + I(AZ_{\bar{m}+1}:W_{\bar{l}+1}|Z_1^{\bar{m}}W_1^{\bar{l}}).
\end{align}
The two terms in RHS are the bipartite mutual information between quantum and classical systems, and hence we can find an upper bound for each term using the chain rule in \eqref{eq:ChainRule2} and Lemma~\ref{lem:CMIbound}.

\textit{First term:} For any $l$, it holds that
\begin{align}
I(A:Z_1^{n_1}|W_1^{l}) = \sum_{m=0}^{n_1-1} I(A:Z_{m+1}|Z_1^mW_1^l) \leq \log |A|,
\end{align}
where the first equality is the chain rule in \eqref{eq:ChainRule2} and the second inequality is by applying Lemma~\ref{lem:CMIbound} to $I(A:Z_1^{n_1}|W_1^{l})$. Then, summing over all $l$ gives us
\begin{align}\label{eq:firstterm}
\sum_{m=0}^{n_1-1} \sum_{l=0}^{n_2-1} I(A:Z_{m+1}|Z_1^mW_1^l) \leq n_2\log |A|.
\end{align}

\textit{Second term:} Using the same argument, for any $m$, it holds that
\begin{align}
I(AZ_{m+1}:W_1^{n_2}|Z_1^{m}) = \sum_{l=0}^{n_2-1} I(AZ_{m+1}:W_{l+1}|Z_1^mW_1^l) \leq \log |AZ_{m+1}|,
\end{align}
and summing over $m$ gives us
\begin{align}\label{eq:secondterm}
\sum_{m=0}^{n_1-1} \sum_{l=0}^{n_2-1} I(AZ_{m+1}:W_{l+1}|Z_1^mW_1^l) \leq n_1\left(\log|A| + \log|Z_1| \right).
\end{align}
Combining \eqref{eq:firstterm} and \eqref{eq:secondterm} gives
\begin{equation}
\begin{split}
n_2\log|A| + n_1\left(\log|A| + \log|Z_1| \right) &\geq \sum_{m=0}^{n_1-1} \sum_{l=0}^{n_2-1} \left[ I(A:Z_{m+1}|Z_1^mW_1^l) + I(AZ_{m+1}:W_{l+1}|Z_1^mW_1^l) \right]\\
&\geq  n_1n_2\left[ I(A:Z_{\bar{m}+1}|Z_1^{\bar{m}}W_1^{\bar{l}}) + I(AZ_{\bar{m}+1}:W_{\bar{l}+1}|Z_1^{\bar{m}}W_1^{\bar{l}}) \right],
\end{split}
\end{equation}
where $\bar{m}$ and $\bar{l}$ are the indices of the smallest element in the sum. Dividing both sides by $n_1n_2$ gives us the desired relation
\begin{equation}
\begin{split}
I(A:Z_{\bar{m}+1}:W_{\bar{l}+1}|Z_1^{\bar{m}}W_1^{\bar{l}}) &= I(A:Z_{\bar{m}+1}|Z_1^{\bar{m}}W_1^{\bar{l}}) + I(AZ_{\bar{m}+1}:W_{\bar{l}+1}|Z_1^{\bar{m}}W_1^{\bar{l}}) \\
&\leq \frac{\log|A|}{n_1} + \frac{\log|A| + \log|Z_1|}{n_2}.
\end{split}
\end{equation}
\end{proof}

It is straightforward to extend Lemma~\ref{lem:mulCMIbound} to the general k-partite case. Now, let us prove the multipartite quantum de Finetti theorem with general linear constraints. We just present the tripartite version in this appendix, but the same proof works for any multipartite case if we have the multipartite generalisation of Lemma~\ref{lem:mulCMIbound}.

\begin{theorem}
Let $\rho_{AB^{n_1}C^{n_2}}$ be a quantum state on $\mathcal{H}_A\otimes\mathcal{H}_{B^{n_1}}\otimes\mathcal{H}_{C^{n_2}}$ which is invariant under permutations on $\mathcal{H}_{B^{n_1}}$ and $\mathcal{H}_{C^{n_2}}$ such that
\begin{equation}
\begin{split}
(\mathcal{E}_{A\rightarrow\tilde{A}}\otimes\mathcal{I}_{B^{n_1}C^{n_2}}) \left( \rho_{AB^{n_1}C^{n_2}} \right) = \mathcal{X}_{\tilde{A}} \otimes\rho_{B^{n_1}C^{n_2}} \quad & \quad \text{linear constraint on A} \\
(\Lambda_{B\rightarrow\tilde{B}}\otimes\mathcal{I}_{C^{n_2}}) \left( \rho_{B^{n_1}C^{n_2}} \right) = \mathcal{Y}_{\tilde{B}} \otimes \rho_{B^{n_1-1}C^{n_2}} \quad & \quad \text{linear constraint on B}\\
(\mathcal{I}_{B^{n_1}}\otimes\Gamma_{C\rightarrow\tilde{C}}) \left( \rho_{B^{n_1}C^{n_2}} \right) = \mathcal{Z}_{\tilde{C}} \otimes \rho_{B^{n_1}C^{n_2-1}} \quad & \quad \text{linear constraint on C}.
\end{split}
\end{equation}
for some linear maps $\mathcal{E}_{A\rightarrow\tilde{A}}$, $\Lambda_{B\rightarrow \tilde{B}}$ and $\Gamma_{C\rightarrow \tilde{C}}$ and operators $\mathcal{X}_{\tilde{A}}$, $\mathcal{Y}_{\tilde{B}}$ and $\mathcal{Z}_{\tilde{C}}$. Then, there exist a probability distribution $\{p_i\}_{i\in I}$ and sets of quantum states $\{\sigma_A^i\}_{i\in I}$, $\{\omega_B^i\}_{i\in I}$ and $\{\tau_C^i\}_{i\in I}$ such that
\begin{align}\label{eq:TriQdFTbound}
\norm{\rho_{ABC} - \sum_{i\in I} p_i \sigma_A^i\otimes\omega_B^i\otimes\tau_C^i}_1 \leq &\; \min\left\{18^{3/2}\sqrt{|ABC|},4|BC|\right\}\sqrt{2\ln2}\left(\sqrt{\frac{\log|A|}{n_1} + \frac{\log|A| + 8\log|B|}{n_2}}\right)
\end{align}
with $\mathcal{E}_{A\rightarrow\tilde{A}} \left( \sigma_A^i \right) = \mathcal{X}_{\tilde{A}}$, $\Lambda_{B\rightarrow\tilde{B}} \left(\omega_B^i\right) = \mathcal{Y}_{\tilde{B}}$, and $\Gamma_{C\rightarrow\tilde{C}}\left(\tau_C^i\right) = \mathcal{Z}_{\tilde{C}}$ for all $i\in I$.
\end{theorem}

\begin{proof}
Let $\mathcal{M}_{B\rightarrow Y}$ be a quantum-to-classical measurement from $B$ to the classical system $Y$, and $\mathcal{M}_{C\rightarrow Z}$ be a quantum-to-classical measurement from $C$ to the classical system $Z$. We apply these measurements to the quantum state $\rho_{AB^{n_1}C^{n_2}}$ and will denote the outcome quantum-classical state as $\rho_{AY^{n_1}Z^{n_2}}$. Then, according to Lemma~\ref{lem:mulCMIbound}, we can find $m\in\{0,\cdots ,n_1-1\}$ and $\ell\in\{0,\cdots ,n_2-1\}$ such that 
\begin{equation}\label{eq:TRIBetterBoundstep1}
\begin{split}
\mathbb{E}_{y^{m}z^{\ell}}&\left\{ \norm{ \rho_{AY_{m+1}Z_{\ell+1}|y^{m}z^{\ell}} - \rho_{A|y^{m}z^{\ell}}\otimes\rho_{Y_{m+1}|y^{m}z^{\ell}}\otimes\rho_{Z_{\ell+1}|y^{m}z^{\ell}} }_1^2 \right\}
\leq \frac{(2\ln 2)\log|A|}{n_1} + \frac{(2\ln2)\left(\log|A| + \log|Y| \right)}{n_2}.
\end{split}
\end{equation}
Let us define $\gamma_{ABC} \equiv \rho_{AB_{m+1}C_{\ell+1}|y^mz^{\ell}} - \rho_{A|y^mz^{\ell}}\otimes\rho_{B_{m+1}|y^mz^{\ell}}\otimes\rho_{C_{\ell+1}|y^mz^{\ell}}$. Note that 
\begin{equation}
\begin{split}
\mathcal{I}_A\otimes\mathcal{M}_{B\rightarrow Y}\otimes\mathcal{M}_{C\rightarrow Z} \left(\gamma_{ABC}\right) =\rho_{AY_{m+1}Z_{\ell+1}|y^{m}z^{\ell}} - \rho_{A|y^{m}z^{\ell}}\otimes\rho_{Y_{m+1}|y^{m}z^{\ell}}\otimes\rho_{Z_{\ell+1}|y^{m}z^{\ell}}.
\end{split}
\end{equation}
As shown in the proof of Theorem~\ref{thm:TriQdeFTwithPTLC}, there are two ways to relate $\norm{\gamma_{ABC}}_1$ to $\norm{ \mathcal{I}_A\otimes\mathcal{M}_{B\rightarrow Y}\otimes\mathcal{M}_{C\rightarrow Z} \left(\gamma_{ABC}\right)}_1$. Using the second part of Lemma~\ref{lem:DataHiding} iteratively, we can obtain
\begin{align}
\norm{\gamma_{ABC}}_1 &\leq 2|C| \norm{ \mathcal{I}_{AB}\otimes\mathcal{M}_{C\rightarrow Z} \left(\gamma_{ABC}\right)}_1\\
&\leq 4|B||C| \norm{ \mathcal{I}_{AC}\otimes\mathcal{M}_{B\rightarrow Y} \left(\mathcal{I}_{AB}\otimes\mathcal{M}_{C\rightarrow Z} \left(\gamma_{ABC}\right)\right)}_1 \\
&= 4|BC|\norm{ \mathcal{I}_{A}\otimes\mathcal{M}_{B\rightarrow Y} \otimes\mathcal{M}_{C\rightarrow Z} \left(\gamma_{ABC}\right)}_1,
\end{align}
and $|Y|\leq|B|^6$.
Using the first part of Lemma~\ref{lem:DataHiding}, we obtain
\begin{equation}
\begin{split}
\norm{\gamma_{ABC}}_1^2 &\leq \left(\sqrt{18^3|ABC|}\right)^2 \norm{ \mathcal{M}_A\otimes\mathcal{M}_{B\rightarrow Y}\otimes\mathcal{M}_{C\rightarrow Z} \left(\gamma_{ABC}\right)}_1^2\\
&\leq \left(\sqrt{18^3|ABC|}\right)^2 \norm{ \mathcal{I}_A\otimes\mathcal{M}_{B\rightarrow Y}\otimes\mathcal{M}_{C\rightarrow Z} \left(\gamma_{ABC}\right)}_1^2 \quad \quad (\because\text{ monotonicity})
\end{split}
\end{equation}
with $|Y|\leq|B|^8$. Combining \eqref{eq:TRIBetterBoundstep1} with the above two results we obtain
\begin{align}
\mathbb{E}_{y^{m}z^{\ell}} & \left\{ \norm{ \rho_{AB_{m+1}C_{\ell+1}|y^{m}z^{\ell}} - \rho_{A|y^{m}z^{\ell}}\otimes\rho_{B_{m+1}|y^{m}z^{\ell}}\otimes\rho_{C_{\ell+1}|y^{m}z^{\ell}} }_1^2 \right\} \\ 
&\leq \min\left\{18^3|ABC|,16|BC|^2\right\}2\ln2\left(\frac{\log|A|}{n_1} + \frac{\left(\log|A| + 8\log|B| \right)}{n_2} \right).
\end{align}
Then, we have
\begin{equation}
\begin{split}
&\norm{\rho_{AB_{m+1}C_{\ell+1}} - \mathbb{E}_{y^mz^{\ell}} \left\{ \rho_{A|y^mz^{\ell}}\otimes\rho_{B_{m+1}|y^mz^{\ell}}\otimes\rho_{C_{\ell+1}|y^mz^{\ell}} \right\}}_1\\
&\leq \; \mathbb{E}_{y^mz^{\ell}} \left\{ \norm{\rho_{AB_{m+1}C_{\ell+1}|y^mz^{\ell}} - \rho_{A|y^{m}z^{\ell}}\otimes\rho_{B_{m+1}|y^{m}z^{\ell}}\otimes\rho_{C_{\ell+1}|y^{m}z^{\ell}}}_1 \right\} \quad\quad (\because\text{ triangular inequality})\\
&\leq \; \sqrt{\mathbb{E}_{y^mz^{\ell}} \left\{ \norm{\rho_{AB_{m+1}C_{\ell+1}|y^mz^{\ell}} - \rho_{A|y^{m}z^{\ell}}\otimes\rho_{B_{m+1}|y^{m}z^{\ell}}\otimes\rho_{C_{\ell+1}|y^{m}z^{\ell}}}_1^2 \right\}} \quad\quad (\because\text{ concavity})\\
&\leq \; \min\left\{18^{3/2}\sqrt{|ABC|},4|BC|\right\}\sqrt{2\ln2}\left(\sqrt{\frac{\log|A|}{n_1} + \frac{\left(\log|A| + 8\log|B|\right)}{n_2}}\right).
\end{split}
\end{equation}
where we used the triangular inequality for the Schatten $p$-norm in the second line and the concavity of the square function in the third line. As $\mathbb{E}_{y^mz^{\ell}} \left\{ \rho_{A|y^mz^{\ell}}\otimes\rho_{B_{m+1}|y^mz^{\ell}}\otimes\rho_{C_{\ell+1}|y^mz^{\ell}} \right\}$ is a separable state with respect to the tripartition $A|B|C$, this proves the first half of the theorem.

The remaining part is to check whether $\rho_{A|y^mz^{\ell}}$, $\rho_{B_{m+1}|y^mz^{\ell}}$ and $\rho_{C_{\ell+1}|y^mz^{\ell}}$ satisfy the desired linear constraints. Let us denote $M^{y_i}_{B_i}$ and $M^{z_i}_{C_i}$ as the measurement operators of the measurements $\mathcal{M}_{B_i\rightarrow Y_i}$ and $\mathcal{M}_{C_i\rightarrow Z_i}$ corresponding to the measurement outcomes $y_i$ and $z_i$, respectively. Then, we can find the followings:
\begin{equation}
\begin{split}
\mathcal{E}_{A\rightarrow\tilde{A}} \left(\sigma_A^i\right) &= \mathcal{E}_{A\rightarrow\tilde{A}} \left(\rho_{A|y^mz^{\ell}}\right)\\
&= \frac{\Tr_{B^mC^{\ell}}\left[ (\mathds{1}_A\otimes M^{y_1}_{B_1}\otimes \cdots \otimes M^{y_m}_{B_m}\otimes M^{z_1}_{C_1}\otimes \cdots \otimes M^{z_{\ell}}_{C_{\ell}}) \mathcal{E}_{A\rightarrow\tilde{A}} \left(\rho_{AB^mC^{\ell}}\right) \right]}{\Tr\left[ (\mathds{1}_A\otimes M^{y_1}_{B_1}\otimes \cdots \otimes M^{y_m}_{B_m}\otimes M^{z_1}_{C_1}\otimes \cdots \otimes M^{z_{\ell}}_{C_{\ell}}) \rho_{AB^mC^{\ell}} \right]}\\
&= \frac{\Tr_{B^mC^{\ell}}\left[ (\mathds{1}_A\otimes M^{y_1}_{B_1}\otimes \cdots \otimes M^{y_m}_{B_m}\otimes M^{z_1}_{C_1}\otimes \cdots \otimes M^{z_{\ell}}_{C_{\ell}}) \left(\mathcal{X}_{\tilde{A}}\otimes \rho_{B^mC^{\ell}}\right)\right]}{\Tr\left[ (\mathds{1}_A\otimes M^{y_1}_{B_1}\otimes \cdots \otimes M^{y_m}_{B_m}\otimes M^{z_1}_{C_1}\otimes \cdots \otimes M^{z_{\ell}}_{C_{\ell}}) \rho_{AB^mC^{\ell}} \right]}\\
&= \mathcal{X}_{\tilde{A}}.
\end{split}
\end{equation}

\begin{equation}
\begin{split}
\Lambda_{B\rightarrow\tilde{B}} \left(\omega_B^i\right) &= \Lambda_{B\rightarrow\tilde{B}} \left(\rho_{B_{m+1}|y^mz^{\ell}}\right)\\
&= \frac{\Tr_{B^mC^{\ell}}\left[ (\mathds{1}_{\tilde{B}}\otimes M^{y_1}_{B_1}\otimes \cdots \otimes M^{y_m}_{B_m}\otimes M^{z_1}_{C_1}\otimes \cdots \otimes M^{z_{\ell}}_{C_{\ell}}) \Lambda_{B\rightarrow\tilde{B}} \left(\rho_{B^{m+1}C^{\ell}}\right) \right]}{\Tr\left[ (M^{y_1}_{B_1}\otimes \cdots \otimes M^{y_m}_{B_m}\otimes \mathds{1}_{B_{m+1}}\otimes M^{z_1}_{C_1}\otimes \cdots \otimes M^{z_{\ell}}_{C_{\ell}}) \rho_{B^{m+1}C^{\ell}} \right]}\\
&= \frac{\Tr_{B^mC^{\ell}}\left[ (\mathds{1}_{\tilde{B}}\otimes M^{y_1}_{B_1}\otimes \cdots \otimes M^{y_m}_{B_m}\otimes M^{z_1}_{C_1}\otimes \cdots \otimes M^{z_{\ell}}_{C_{\ell}}) \left(\mathcal{Y}_{\tilde{B}}\otimes\rho_{B^mC^{\ell}} \right)\right]}{\Tr\left[ (M^{y_1}_{B_1}\otimes \cdots \otimes M^{y_m}_{B_m}\otimes \mathds{1}_{B_{m+1}}\otimes M^{z_1}_{C_1}\otimes \cdots \otimes M^{z_{\ell}}_{C_{\ell}}) \rho_{B^{m+1}C^{\ell}} \right]}\\
&= \mathcal{Y}_{\tilde{B}}.
\end{split}
\end{equation}

\begin{equation}
\begin{split}
\Gamma_{C\rightarrow\tilde{C}} \left(\tau_C^i\right) &= \Gamma_{C\rightarrow\tilde{C}} \left(\rho_{C_{\ell+1}|y^mz^{\ell}}\right)\\
&= \frac{\Tr_{B^mC^{\ell}}\left[ (\mathds{1}_{\tilde{C}}\otimes M^{y_1}_{B_1}\otimes \cdots \otimes M^{y_m}_{B_m}\otimes M^{z_1}_{C_1}\otimes \cdots \otimes M^{z_{\ell}}_{C_{\ell}}) \Gamma_{C\rightarrow\tilde{C}} \left(\rho_{B^mC^{\ell+1}}\right) \right]}{\Tr\left[ (M^{y_1}_{B_1}\otimes \cdots \otimes M^{y_m}_{B_m}\otimes M^{z_1}_{C_1}\otimes \cdots \otimes M^{z_{\ell}}_{C_{\ell}}\otimes \mathds{1}_{C_{\ell+1}}) \left(\rho_{B^mC^{\ell+1}}\right) \right]}\\
&= \frac{\Tr_{B^mC^{\ell}}\left[ (\mathds{1}_{\tilde{C}}\otimes M^{y_1}_{B_1}\otimes \cdots \otimes M^{y_m}_{B_m}\otimes M^{z_1}_{C_1}\otimes \cdots \otimes M^{z_{\ell}}_{C_{\ell}}) \left(\mathcal{Z}_{\tilde{C}}\otimes\rho_{B^mC^{\ell}} \right)\right]}{\Tr\left[ (M^{y_1}_{B_1} \otimes \cdots \otimes M^{y_m}_{B_m} \otimes M^{z_1}_{C_1}\otimes \cdots \otimes M^{z_{\ell}}_{C_{\ell}} \otimes \mathds{1}_{C_{\ell+1}}) \left(\rho_{B^mC^{\ell+1}}\right) \right]}\\
&= \mathcal{Z}_{\tilde{C}}.
\end{split}
\end{equation}
\end{proof}


\section{General games}\label{app:GeneralGames}

Here, we consider general games when the questions for Alice and Bob are correlated, i.e. $\pi(q_1,q_2) \neq \pi_1(q_1)\pi_2(q_2)$. Recall that the value for a two-player game with quantum assistance of dimension $|T|$ can be written as
\begin{align}\tag{\ref{eq:QWinP_T}}
\omega_{Q(T)} = \max_{E,D,\rho} \sum_{q_1,q_2} \pi(q_1,q_2) \sum_{a_1,a_2} V(a_1,a_2,q_1,q_2) \tr \left[ \rho_{T\hat{T}} \left( E_T(a_1|q_1)\otimes D_{\hat{T}}(a_2|q_2) \right) \right].
\end{align}
Let us divide this with $|T|^2|Q_1||Q_2|$ and then multiply by it again
\begin{align}
\omega_{Q(T)} = |Q_1||Q_2||T|^2 \max_{E,D,\rho} \sum_{q_1,q_2} \pi(q_1,q_2) \sum_{a_1,a_2} V(a_1,a_2,q_1,q_2) \tr \left[ \rho_{T\hat{T}}\left( \frac{E_T(a_1|q_1)}{|T||Q_1|} \otimes \frac{D_{\hat{T}}(a_2|q_2)}{|T||Q_2|} \right) \right].
\end{align}
Writing classical systems as quantum systems, we obtain
\begin{equation}\label{eq:GGfullQ}
\begin{split}
\omega_{Q(T)}(V, \pi) = \; |Q_1|&|Q_2||T|^2 \max_{E, D, \rho} \; \tr \left[ \left( V^{\pi}_{A_1A_2Q_1Q_2} \otimes \rho_{T\hat{T}} \right)\left( E_{A_1Q_1T}\otimes D_{A_2Q_2\hat{T}} \right) \right]\\
\text{s.t.} \quad  &\rho_{T\hat{T}} \geq 0\, , \quad \quad \Tr [\rho_{T\hat{T}}] = 1 \\
& E_{A_1Q_1T} = \sum_{a_1, q_1} \ket{a_1q_1}\!\!\bra{a_1q_1}_{A_1Q_1} \otimes \frac{E_T(a_1|q_1)}{|T||Q_1|} \geq 0\\
& E_{Q_1T} = \sum_{q_1} \ket{q_1}\!\!\bra{q_1}_{Q_1} \otimes \frac{\mathds{1}_T}{|T||Q_1|} = \frac{\mathds{1}_{Q_1T}}{|T||Q_1|}\\
& D_{A_2Q_2\hat{T}} = \sum_{a_2q_2} \ket{a_2q_2}\!\!\bra{a_2q_2}_{A_2Q_2} \otimes \frac{D_{\hat{T}}(a_2|q_2)}{|T||Q_2|} \geq 0\\
& D_{Q_2\hat{T}} = \sum_{q_2} \ket{q_2}\!\!\bra{q_2}_{Q_2} \otimes \frac{\mathds{1}_{\hat{T}}}{|T||Q_2|} = \frac{\mathds{1}_{Q_2\hat{T}}}{|T||Q_2|},
\end{split}
\end{equation}
where $V^{\pi}_{A_1A_2Q_1Q_2} = \sum_{a_1,q_1,a_2,q_2} \pi(q_1, q_2) V(a_1,a_2,q_1,q_2) \ket{a_1,a_2,q_1,q_2}\!\!\bra{a_1,a_2,q_1,q_2}$. If we use the same modified swap trick from Lemma~\ref{lem:modSwapTrick}, we can rewrite this as an instance of the quantum separability problem
\begin{equation}\label{eq:GGSepP}
\begin{split}
w^Q_T(V, \pi) = |Q_1|&|Q_2||T|^2 \max_{E, D, \rho} \; \tr \left[ \left( V^{\pi}_{A_1A_2Q_1Q_2}\otimes \Phi_{T\hat{T}|S\hat{S}}\right)\left( E_{A_1Q_1T}\otimes D_{A_2Q_2\hat{T}}\otimes\rho_{S\hat{S}} \right) \right],
\end{split}
\end{equation}
with the same constraints in \eqref{eq:GGfullQ}. Then, using extendible states, we can write a hierarchy of SDP relaxations for general games
\begin{align}\label{eq:GGsdpRelax}
\operatorname{sdp}^G_{n_1,n_2} (V,\pi,T) :=\;|Q_1||Q_2||T|^2 \max_{\rho} \tr\left[ \left( V^{\pi}_{A_1A_2Q_1Q_2}\otimes \Phi_{T\hat{T}|S\hat{S}} \right) \rho_{(A_1Q_1T)(A_2Q_2\hat{T})(S\hat{S})} \right]
\end{align}
\begin{subequations}
\label{eq:sdpRelax_constraints}
\begin{align}
\text{s.t.} \qquad
& \rho_{(A_1Q_1T)^{n_1}(A_2Q_2\hat{T})^{n_2}(S\hat{S})} \geq 0, \quad \tr\left[ \rho_{(A_1Q_1T)^{n_1}(A_2Q_2\hat{T})^{n_2}(S\hat{S})} \right] = 1\label{cond:rho_state} \\
&\mathcal{U}^{\pi}_{(A_1Q_1T)^{n_1}} \left( \rho_{(A_1Q_1T)^{n_1}(A_2Q_2\hat{T})^{n_2}(S\hat{S})} \right) = \rho_{(A_1Q_1T)^{n_1}(A_2Q_2\hat{T})^{n_2}(S\hat{S})} \quad \quad \forall \pi\in\mathcal{S}((A_1Q_1T)^{n_1})\\
&\mathcal{U}^{\pi}_{(A_2Q_2\hat{T})^{n_2}} \left( \rho_{(A_1Q_1T)^{n_1}(A_2Q_2\hat{T})^{n_2}(S\hat{S})} \right) = \rho_{(A_1Q_1T)^{n_1}(A_2Q_2\hat{T})^{n_2}(S\hat{S})} \quad \quad \forall \pi\in\mathcal{S}((A_2Q_2\hat{T})^{n_2})\\
&\tr_{A_1}[\rho_{(A_1Q_1T)^{n_1}(A_2Q_2\hat{T})^{n_2}(S\hat{S})}] = \left( \frac{\mathds{1}_{Q_1T}}{|T||Q_1|} \right) \otimes \rho_{(A_1Q_1T)^{n_1-1}(A_2Q_2\hat{T})^{n_2}(S\hat{S})}\label{cond:rho_lc1}\\
&\tr_{A_2}[\rho_{(A_1Q_1T)^{n_1}(A_2Q_2\hat{T})^{n_2}(S\hat{S})}] = \left( \frac{\mathds{1}_{Q_2\hat{T}}}{|T||Q_2|} \right) \otimes \rho_{(A_1Q_1T)^{n_1}(A_2Q_2\hat{T})^{n_2-1}(S\hat{S})}\label{cond:rho_lc2}\\
&\rho^{T_{A_1Q_1T}}_{(A_1Q_1T)(A_2Q_2\hat{T})^{n}(S\hat{S})^{n}} \geq 0, \quad \rho^{T_{(A_2Q_2\hat{T})^{n}}}_{(A_1Q_1T)(A_2Q_2\hat{T})^{n}(S\hat{S})^{n}} \geq 0, \quad \rho^{T_{(S\hat{S})^{n}}}_{(A_1Q_1T)(A_2Q_2\hat{T})^{n}(S\hat{S})^{n}} \geq 0
\end{align}    
\end{subequations}
It is again easy to bound the computational complexity of these SDP relaxations using the multipartite quantum de Finetti theorem derived in Section~\ref{subsec:TriQdeF}. We find that
\begin{equation}
\begin{split}
&|\operatorname{sdp}^G_{n_1, n_2}(V,\pi,T) - \omega_{Q(T)}(V,\pi)|
\\
&\leq |Q_1||Q_2||T|^2 \left\vert \tr \left[ \left( V^{\pi}_{A_1A_2Q_1Q_2}\otimes\Phi_{T\hat{T}|S\hat{S}}\right) \left( \tilde{\rho}_{A_1Q_1TA_2Q_2\hat{T}S\hat{S}} - \sum_i p_i \sigma^i_{A_1Q_1T}\otimes\omega^i_{A_2Q_2\hat{T}}\otimes\tau^i_{S\hat{S}} \right) \right]\right\vert \\
&\leq |Q_1||Q_2||T|^2 \left( 18^{3/2} |T|^4 \left(\sqrt{4\ln2}\right) \left(\sqrt{\frac{\left( \log|A_1| + 16\log|T| \right)}{n_2} + \frac{\log|A_1|}{n_1}}\right)\right),
\end{split}
\end{equation}
where we assume that $\tilde{\rho}_{A_1Q_1TA_2Q_2\hat{T}S\hat{S}}$ is the optimal solution of $\sdp^G_{n_1,n_2}(V,\pi,T)$, and $\sum_i p_i \sigma^i_{A_1Q_1T}\otimes\omega^i_{A_2Q_2\hat{T}}\otimes\tau^i_{S\hat{S}}$ is one of the close separable states to $\tilde{\rho}_{A_1Q_1TA_2Q_2\hat{T}S\hat{S}}$ specified by Theorem~\ref{thm:TriQdeFTwithPTLC}. The intermediate step is the same as in the proof of Theorem~\ref{thm:Convergencen1n2}. The additional $|Q_1||Q_2|$ factor leads to a worse convergence rate; setting $n:=n_1=n_2$, the size of the SDP becomes
\begin{align}
\exp \Bigg( \mathcal{O}\Bigg( \frac{|T|^{12}|Q|^4\left(\log^2|A||T|+\log|A||T|\log|Q|\right)}{\epsilon^2} \Bigg) \Bigg).
\end{align}
In contrast to free games, \eqref{eq:SDPsize_freegames}, this is exponential in terms of $|Q|$.


\section{General NPA constraints on optimisation variable}\label{app:npa_constraints}

In the main text, we explained how to add the NPA constraints to our SDP relaxations, and we noticed that some of the entries of the NPA matrix $\Gamma_k$ can be expressed in terms of linear combinations of the optimisation variable. Here, we explicitly show how these linear combinations are derived, for an SDP relaxation obtained by extending both the subsystems $A_1Q_1T$ and $A_2Q_2\hat{T}$. This relaxation is defined as
\begin{align} \label{eq:obj_funct_original}
\overline{\sdp}_{n_1,n_2}(V,\pi,T):= |T|^2 \max_{\rho} \sum_{a_1,a_2,q_1,q_2} V(a_1,a_2,q_1,q_2) \Trt{\Phi_{T\hat{T}|S\hat{S}} \, \rho_{T\hat{T}S\hat{S}}(a_1,a_2,q_1,q_2)}
\end{align}
\begin{align}
\text{s.t.} \qquad
&\rho_{T^{n_1}\hat{T}^{n_2}S\hat{S}}(a_1^{n_1}, a_2^{n_2}, q_1^{n_1}, q_2^{n_2}) \geq 0, \quad \sum_{a_1^{n_1}, a_2^{n_2}, q_1^{n_1}, q_2^{n_2}} \Trt{\rho_{T^{n_1}\hat{T}^{n_2}S\hat{S}}(a_1^{n_1}, a_2^{n_2}, q_1^{n_1}, q_2^{n_2})} = 1\\
&\rho_{T^{n_1}\hat{T}^{n_2}S\hat{S}}(a_1^{n_1}, a_2^{n_2}, q_1^{n_1}, q_2^{n_2}) \ \text{perm.~inv.~on} \ (A_1Q_1T)^{n_1} \ \text{and} \ (A_2Q_2\hat{T})^{n_2}\ \text{wrt to other systems}\\
&\sum_{a_1} \, \rho_{T^{n_1}\hat{T}^{n_2}S\hat{S}}(a_1^{n_1}, a_2^{n_2}, q_1^{n_1}, q_2^{n_2}) = \pi_1(q_1) \, \frac{\mathds{1}_T}{|T|} \otimes \rho_{T^{n_1-1}\hat{T}^{n_2}S\hat{S}}(a_1^{n_1-1}, a_2^{n_2}, q_1^{n_1-1}, q_2^{n_2})\\
&\sum_{a_2} \, \rho_{T^{n_1}\hat{T}^{n_2}S\hat{S}}(a_1^{n_1}, a_2^{n_2}, q_1^{n_1}, q_2^{n_2}) = \pi_2(q_2) \, \frac{\mathds{1}_{\hat{T}}}{|T|} \otimes \rho_{T^{n_1}\hat{T}^{n_2-1}S\hat{S}}(a_1^{n_1}, a_2^{n_2-1}, q_1^{n_1}, q_2^{n_2-1}).
\end{align}
It is worth noting that in the above relaxation we have extended the subsystem $A_1Q_1T$ rather than $S\hat{S}$ as in $\sdp_{n_1,n_2}(V,\pi,T)$. This relaxation has slightly worse analytical convergence bounds than the one described in the main text, since we can no more exploit the partial trace condition if we extend the subsystem $A_1Q_1T$. Nonetheless, we consider it here for the following reasons:

\begin{itemize}
\item From the point of view of numerically implementing the SDP, the relaxation $\overline{\sdp}_{n_1,n_2}(V,\pi,T)$ requires a smaller number of variables than $\sdp_{n_1,n_2}(V,\pi,T)$, particularly when we add the NPA constraints. This is because, as we show in the remaining of the appendix, we can rewrite some of the entries of $\Gamma_k$ as linear combinations of the optimisation variable $\rho_{T^{n_1}\hat{T}^{n_2}S\hat{S}}(a_1^{n_1}, a_2^{n_2}, q_1^{n_1}, q_2^{n_2})$. For the relaxation $\sdp_{n_1,n_2}(V,\pi,T)$, only a subset of these entries can be expressed as a linear combination of the corresponding variable, and the remaining entries need to be accounted as new variables.
\item The relation between the entries of $\Gamma_k$ and the optimisation variable of the relaxation $\overline{\sdp}_{n_1,n_2}(V,\pi,T)$ is more general than the one for the relaxation $\sdp_{n_1,n_2}(V,\pi,T)$ used in the main text. For this reason, in this appendix, we explicitly derive the former relation and briefly comment on how to re-purpose it for the SDP hierarchy used in the main text.
\item The relaxations $\overline{\sdp}_{n_1,n_2}(V,\pi,T)$ scale better than $\sdp_{n_1,n_2}(V,\pi,T)$ in numerical implementations. Extending the subsystem $A_1Q_1T$ only increases the total dimension of the quantum system in the variable by $|T|$, but extending the subsystem $S\hat{S}$ increases it by $|T|^2$. Thus, even though $\sdp_{n_1,n_2}(V,\pi,T)$ has a better theoretical convergence speed,  it is practically more advantageous to use $\overline{\sdp}_{n_1,n_2}(V,\pi,T)$.
\end{itemize}

Our goal is to provide a set of instructions that allows us to express the elements of the NPA matrix $\Gamma_k$ as functions of the variable in the new relaxations $\overline{\sdp}_{n_1,n_2}(V,\pi,T)$, where $k \leq \min \left\{ n_1, n_2 \right\}$. Let us first notice that our relaxation aims to approximate the optimal quantum winning probability of a game $(V,\pi)$, which we rewrite here for convenience
\begin{align}
\omega_{Q(T)}(V,\pi) = \max_{E,D,\,\rho\in\mathbb{C}^{T^2}} \sum_{q_1,q_2} \pi(q_1,q_2) \sum_{a_1,a_2} V(a_1,a_2,q_1,q_2) \tr \left[ \rho_{T\hat{T}} \left( E_T(a_1|q_1)\otimes D_{\hat{T}}(a_2|q_2) \right) \right].
\end{align}
By direct comparison with the objective function in \eqref{eq:obj_funct_original}, we see that, for our relaxation to obtain the optimal value $\omega_{Q(T)}(V,\pi)$, the optimisation variable reduced to the system $T\hat{T}S\hat{S}$ needs to be as follow,
\begin{align} \label{eq:opt_red_variable}
\rho_{T\hat{T}S\hat{S}}(a_1,q_1,a_2,q_2)=\frac{\pi(q_1,q_2)}{|T|^2} \, E_T(a_1|q_1)\otimes D_{\hat{T}}(a_2|q_2) \otimes \rho_{S\hat{S}}^T\qquad\forall \, a_1,q_1,a_2,q_2,
\end{align}
where $\left\{ E_T(a_1|q_1) \right\}_{a_1,q_1}$ and $\left\{ D_{\hat{T}}(a_2|q_2) \right\}_{a_2,q_2}$ are the optimal measurements, and $\rho_{S\hat{S}}$ is the optimal state of the assisting system. In order to derive the full optimisation variable $\rho_{T^{n_1}\hat{T}^{n_2}S\hat{S}}(a_1^{n_1}, a_2^{n_2}, q_1^{n_1}, q_2^{n_2})$, we need to first make use of the assumption that $(V,\pi)$ is a free game. Thus, for a fixed value of $a_1$, $q_1$, $a_2$ and $q_2$ we can rewrite \eqref{eq:opt_red_variable} as
\begin{align}
\rho_{T\hat{T}S\hat{S}}(a_1,q_1,a_2,q_2)=E_T(a_1,q_1) \otimes D_{\hat{T}}(a_2,q_2) \otimes \rho_{S\hat{S}}^T,
\end{align}
where $E_T(a_1,q_1) = \frac{\pi_1(q_1)}{|T|} \, E_T(a_1|q_1)$ and
$D_{\hat{T}}(a_2,q_2) = \frac{\pi_2(q_2)}{|T|} \, D_{\hat{T}}(a_2|q_2)$.
\par
We can extend this state by taking $n_1$ i.i.d copies of the system in $T$, and $n_2$ i.i.d copies of the system in $\hat{T}$, obtaining the following assignment for the objective variable which is optimal and satisfies all the constraints in $\overline{\sdp}_{n_1,n_2}$,
\begin{align} \label{eq:opt_var_SDP_NPA}
&\rho_{T^{n_1}\hat{T}^{n_2}S\hat{S}}(a_1^{n_1}, a_2^{n_2}, q_1^{n_1}, q_2^{n_2})\notag\\
&=E_{T_1}(a_1^{(1)},q_1^{(1)})\otimes \ldots \otimes E_{T_{n_1}}(a_1^{(n_1)},q_1^{(n_1)})\otimes D_{\hat{T}_1}(a_2^{(1)},q_2^{(1)})\otimes \ldots \otimes D_{\hat{T}_{n_2}}(a_2^{(n_2)},q_2^{(n_2)})\otimes\rho_{S\hat{S}}^T.
\end{align}
We can now make use of the explicit form of the optimal variable to derive the NPA constraints for our SDP relaxations. The highest level of the NPA hierarchy we can fully implement is given by the minimum between $n_1$ and $n_2$, and in the following we assume without loss of generality that $n_1 \geq n_2$. For a given $k \leq n_2$, the element of the NPA matrix $\Gamma_k$ which we can express as a function of the optimisation variable are of the form,
\begin{align} \label{eq:NPA_element_mat}
p(a_1^m, a_2^{\ell}|q_1^m, q_2^{\ell})=\Trt{\left(E_T(a_1^{(1)}|q_1^{(1)}) \ldots E_T(a_1^{(m)}|q_1^{(m)})\otimes D_{\hat{T}}(a_2^{(1)}|q_2^{(1)}) \ldots D_{\hat{T}}(a_2^{(\ell)}|q_2^{(\ell)})\right)\rho_{T\hat{T}}},
\end{align}
where $m, \ell \leq k$. The other elements of the matrix are either zeros, if we make the additional assumption that $\left\{ E_T(a_1|q_1) \right\}_{a_1}$ and $\left\{ D_{\hat{T}}(a_2|q_2) \right\}_{a_2}$ are projective measurements (PVMs) for each $q_1$ and $q_2$ respectively, or they need to be considered as new variables of the problem.

To rewrite the elements of the NPA matrix in \eqref{eq:NPA_element_mat} in terms of $\rho_{T^{n_1}\hat{T}^{n_2}S\hat{S}}(a_1^{n_1}, a_2^{n_2}, q_1^{n_1}, q_2^{n_2})$ we first need to introduce the following lemma which generalises Lemma~\ref{lem:modSwapTrick} in the main text.

\begin{lemma} \label{lem:general_swap_trick}
Consider a set of operators $\left\{ M_i \right\}_{i=1}^n$, each of them acting over the Hilbert space $\hil$.
Then, it holds that
\begin{align} \label{eq:general_swap_trick}
\Trt{P_{\text{cyclic}}^{n} \, (M_1 \otimes M_2 \otimes \ldots \otimes M_n)} = \Trt{M_1 M_2 \ldots M_n},
\end{align}
where $P_{\text{cyclic}}^{n} \in \mathcal{B}(\hil)$ is the unitary operator associated with the cyclic permutation $\pi$, which acts over the $n$ tuple as
$\pi \, (1,2,3,\ldots,n) = (2,3,\ldots,n,1)$.
\end{lemma}

\begin{proof}
The lemma is proven by explicitly computing the RHS and LHS of \eqref{eq:general_swap_trick}. For the RHS we have,
\begin{align}
\Trt{M_1 M_2 \ldots M_n}&=\Trt{\sum_{i_1,j_1,i_2,j_2,\ldots,i_n,j_n}m^{(1)}_{i_1,j_1} m^{(2)}_{i_2,j_2} \ldots m^{(n)}_{i_n,j_n}\ket{i_1}\bra{j_1} \ket{i_2}\bra{j_2} \ldots \ket{i_n}\bra{j_n}} \\
&=\sum_{i_1,j_1,i_2,j_2,\ldots,i_n,j_n}m^{(1)}_{i_1,j_1} m^{(2)}_{i_2,j_2} \ldots m^{(n)}_{i_n,j_n} \,\delta_{j_1,i_2} \delta_{j_2,i_3} \ldots \delta_{j_n,i_1} =\sum_{i_1,i_2,\ldots,i_n}m^{(1)}_{i_1,i_2} m^{(2)}_{i_2,i_3} \ldots m^{(n)}_{i_n,i_1}.
\end{align}
To compute the LHS, we first need the explicit form of the operator $P_{\text{cyclic}}^{n}$,
\begin{align}
P_{\text{cyclic}}^{n}=\sum_{k_1,k_2,k_3,k_4,\ldots,k_n}\ket{k_2}\bra{k_1}_1 \otimes \ket{k_3}\bra{k_2}_2 \otimes \ket{k_4}\bra{k_3}_3 \otimes \ldots \otimes \ket{k_1}\bra{k_n}_n,
\end{align}
which takes a vector on the subsystem $k$ and maps it to the subsystem $k-1$, with the exception of $k=1$ which is mapped into the $n$-th subsystem.
The LHS of \eqref{eq:general_swap_trick} is then
\begin{align}
&\Trt{P_{\text{cyclic}}^{n} \, (M_1 \otimes M_2 \otimes \ldots \otimes M_n)}\\
&=\sum_{i_1,j_1,i_2,j_2,\ldots,i_n,j_n} \,
m^{(1)}_{i_1,j_1} m^{(2)}_{i_2,j_2} \ldots m^{(n)}_{i_n,j_n} \, 
\Trt{P_{\text{cyclic}}^{n} \, 
\ket{i_1}\bra{j_1}_1 \otimes \ket{i_2}\bra{j_2}_2
\otimes \ldots \otimes
\ket{i_n}\bra{j_n}_n} \\
&=\sum_{i_1,j_1,i_2,j_2,\ldots,i_n,j_n} \,
m^{(1)}_{i_1,j_1} m^{(2)}_{i_2,j_2} \ldots m^{(n)}_{i_n,j_n} \,
\Trt{\ket{i_2}\bra{j_1}_1 \otimes \ket{i_3}\bra{j_2}_2
\otimes \ldots \otimes
\ket{i_1}\bra{j_n}_n} \\
&=\sum_{i_1,j_1,i_2,j_2,\ldots,i_n,j_n} \,
m^{(1)}_{i_1,j_1} m^{(2)}_{i_2,j_2} \ldots m^{(n)}_{i_n,j_n} \,
\delta_{i_2,j_1} \delta_{i_3,j_2} \ldots \delta_{i_1,j_n}
=\sum_{i_1,i_2,\ldots,i_n} \,
m^{(1)}_{i_1,i_2} m^{(2)}_{i_2,i_3} \ldots m^{(n)}_{i_n,i_1},
\end{align}
which concludes the proof.
\end{proof}

We can now derive the map between the elements of $\Gamma_k$ and the optimisation variable in our SDP relaxations. 
Without loss of generality, let us assume that $m \geq \ell$ in \eqref{eq:NPA_element_mat}; by reordering the operators we get
\begin{align}
p(a_1^m, a_2^{\ell}|q_1^m, q_2^{\ell})
&=
\Trt{
\left(
\prod_{i=1}^{\ell}
E_T(a_1^{(i)}|q_1^{(i)}) \otimes D_{\hat{T}}(a_2^{(i)}|q_2^{(i)})
\prod_{j=\ell+1}^{m}
E_T(a_1^{(j)}|q_1^{(j)})
\otimes
\id_{\hat{T}}
\right)
\rho_{T\hat{T}}
} \\
&=
\Trt{
P_{\text{cyclic}}^{m+1} \,
\left(
\bigotimes_{i=1}^{\ell}
E_{T_i}(a_1^{(i)}|q_1^{(i)})
\otimes
D_{\hat{T}_i}(a_2^{(i)}|q_2^{(i)})
\right)
\otimes
\left(
\bigotimes_{j=\ell+1}^{m}
E_{T_{j}}(a_1^{(j)}|q_1^{(j)})
\otimes
\id_{\hat{T}_{j}}
\right)
\otimes
\rho_{S\hat{S}}
},
\end{align}
where the second equality follows from Lemma~\ref{lem:general_swap_trick} with $\hil_{T\hat{T}}$. The operator in the above equation is close to the optimisation variable $\rho_{T^{n_1}\hat{T}^{n_2}S\hat{S}}(a_1^{n_1}, a_2^{n_2}, q_1^{n_1}, q_2^{n_2})$, whose dimension can be reduced by summing over the classical variables while tracing out the quantum degrees of freedom
\begin{align}
\rho_{T^{m}\hat{T}^{m}S\hat{S}}(a_1^{m}, a_2^{\ell}, q_1^{m}, q_2^{\ell})
&=
\sum_{\substack{a^{(m+1)}_1,q^{(m+1)}_1,\ldots,a^{(n_1)}_1,q^{(n_1)}_1, \\
a^{(\ell+1)}_2,q^{(\ell+1)}_2,\ldots,a^{(n_2)}_2,q^{(n_2)}_2}}
\Trp{T_{m+1} \ldots T_{n_1} \hat{T}_{m+1} \ldots \hat{T}_{n_2} }{
\rho_{T^{n_1}\hat{T}^{n_2}S\hat{S}}(a_1^{n_1}, a_2^{n_2}, q_1^{n_1}, q_2^{n_2})
} \\
&=
\left(
\bigotimes_{i=1}^{\ell}
E_{T_i}(a_1^{(i)},q_1^{(i)})
\otimes
D_{\hat{T}_i}(a_2^{(i)},q_2^{(i)})
\right)
\otimes
\left(
\bigotimes_{j=\ell+1}^{m}
E_{T_{j}}(a_1^{(j)},q_1^{(j)})
\otimes
\frac{\id_{\hat{T}_j}}{|T|}
\right)
\otimes
\rho_{S\hat{S}}^T,
\end{align}
where we have used the fact that $\sum_{a_1,q_1} E_T(a_1,q_1) = \frac{\id_T}{|T|}$ and $\sum_{a_2,q_2} E_{\hat{T}}(a_2,q_2) = \frac{\id_{\hat{T}}}{|T|}$. By combining together the two equations above, we find that the elements of $\Gamma_k$ can be expressed in terms of the optimisation variable as
\begin{equation} \label{eq:NPA_entries_SDP_variable}
p(a_1^m, a_2^{\ell}|q_1^m, q_2^{\ell})
=
\begin{cases}
\frac{|T|^{2m}}{\prod_{i=1}^m \pi_1(q_1^{(i)}) \, \prod_{j=1}^{\ell} \pi_2(q_2^{(j)})}
\Trt{
\left( P_{\text{cyclic}}^{m+1} \right)^{T_{S\hat{S}}} \,
\rho_{T^{m}\hat{T}^{m}S\hat{S}}(a_1^{m}, a_2^{\ell}, q_1^{m}, q_2^{\ell})
}, & \forall \, \ell, m \ : \ \ell \leq m \leq k \\
\frac{|T|^{2\ell}}{\prod_{i=1}^m \pi_1(q_1^{(i)}) \, \prod_{j=1}^{\ell} \pi_2(q_2^{(j)})}
\Trt{
\left( P_{\text{cyclic}}^{\ell+1} \right)^{T_{S\hat{S}}} \,
\rho_{T^{\ell}\hat{T}^{\ell}S\hat{S}}(a_1^{m}, a_2^{\ell}, q_1^{m}, q_2^{\ell})
}, & \forall \, \ell, m \ : \ m \leq \ell \leq k,
\end{cases}  
\end{equation}
where we have included the case in which $\ell \geq m$, that can be derived analogously.
\par
When the NPA constraints are applied to the original $\sdp_{n_1,n_2}(V,\pi,T)$ in the main text, the optimization variable is given by $\rho_{T(\hat{T})^{n_1}(S\hat{S})^{n_2}}(a_1, a_2^{n_1}, q_1, q_2^{n_1})$. The relation between this variable and the entries of the NPA matrix $\Gamma_k$ can be obtain following the same steps presented in this appendix. The main difference is that the optimisation variable is defined over a single subsystem $T$. As a result, when $\ell \geq 2$ we need to pad the variable with the $\ell-1$ copies of the maximally-mixed state, so as to be able to apply the operator $P_{\text{cyclic}}^{\ell+1}$. The relation between entries of $\Gamma_k$ and the optimisation variable are thus given by,
\begin{align}
p(a_1, a_2^{\ell}|q_1, q_2^{\ell})=\frac{|T|^{\ell+1}}{\pi_1(q_1) \, \prod_{j=1}^{\ell} \pi_2(q_2^{(j)})}
\Trt{\left( P_{\text{cyclic}}^{\ell+1} \right)^{T_{S\hat{S}}} \,
\left( \bigotimes_{i=2}^{\ell} \id_{T_i} \right) \otimes\rho_{T\hat{T}^{\ell}S\hat{S}}(a_1, a_2^{\ell}, q_1^{m}, q_2^{\ell})} \qquad \forall \, \ell \ : \ \ell \leq k.
\end{align}


\section{Non-signalling value}\label{app:sdp11_non_signalling}

Apart from the classical and quantum values defined in \eqref{eq:CWinP} and \eqref{eq:QWinP}, respectively, we can define another quantity called the non-signalling value. This is the optimal success probability achieved by non-signalling correlations such as 
\begin{align}
\omega_{\ns}(V,\pi) := \max_{p \in\ns} \sum_{q_1,q_2} \pi(q_1,q_2) \sum_{a_1,a_2} V(a_1,a_2,q_1,q_2) \, p(a_1,a_2|q_1,q_2),
\end{align}
where $\ns$ denotes the set of all non-signalling correlations
\begin{align}
\text{$\sum_{a_1} p(a_1,a_2|q_1,q_2) = p(a_2|q_2) \; \forall a_2, q_2$ and $\sum_{a_2} p(a_1,a_2|q_1,q_2) = p(a_1|q_1) \; \forall a_1, q_1$.}
\end{align}
As any classical or quantum correlation satisfies the non-signalling condition, the non-signalling value gives an upper bound to the classical and quantum values
\begin{align}
\omega_C(V,\pi) \;\leq\; \omega_{Q(T)}(V,\pi) \;\leq\; \omega_Q(V,\pi) \;\leq\; \omega_{\ns}(V,\pi).
\end{align}
In this section, we show that $\sdp_{1,1}(V,\pi,T)$ is equal to $\omega_{\ns}(V,\pi)$ for any $|T|$.

\begin{lemma}
Let $\sdp_{1,1} (V,\pi,T)$ be the first level SDP relaxation for the game with $V$ and $\pi_1(q_1)\pi_2(q_2)$. Then, we have for all $|T|$ that
\begin{align}
\sdp_{1,1}(V,\pi,T) = \omega_{\ns}(V,\pi).
\end{align}
\end{lemma}

\begin{proof}
We first show that $\sdp_{1,1}(V,\pi,T) \leq \omega_{\ns}(V,\pi)$. From $\sdp_{1,1}(V,\pi,T)$ we have the linear constraint 
\begin{align}\label{eq:sdp11}
\sum_{a_2} \rho_{T\hat{T}S\hat{S}} (a_1,a_2,q_1,q_2) = \pi_2(q_2) \frac{\id_{\hat{T}}}{|T|} \otimes \rho_{TS\hat{S}}(a_1,q_1).
\end{align}
Then, using the expression for $p(a_1,a_2|q_1,q_2)$ in \eqref{eq:JointPasRho}, we can write the sum of $p(a_1,a_2|q_1,q_2)$ over $a_2$ as
\begin{equation}
\begin{split}
\sum_{a_2} p(a_1,a_2|q_1,q_2) &= \frac{|T|^2}{\pi_1(q_1)\pi_2(q_2)} \tr\left[\Phi_{T\hat{T}|S\hat{S}} \left( \sum_{a_2} \rho_{T\hat{T}S\hat{S}}(a_1,a_2,q_1,q_2)\right)\right]\\
&= \frac{|T|}{\pi_1(q_1)} \tr\left[\Phi_{T\hat{T}|S\hat{S}} \left( \id_{\hat{T}} \otimes \rho_{TS\hat{S}}(a_1,q_1)\right)\right].
\end{split}
\end{equation}
Using a more general formula for $p(a_1|q_1)$ without assuming the non-signalling, we obtain
\begin{equation}
\begin{split}
p(a_1|q_1) &= \sum_{a_2,q_2}p(q_2)p(a_1,a_2|q_1,q_2) = \frac{|T|^2}{\pi_1(q_1)} \tr\left[\Phi_{T\hat{T}|S\hat{S}} \left( \sum_{a_2,q_2} \rho_{T\hat{T}S\hat{S}}(a_1,a_2,q_1,q_2)\right)\right]\\
&= \frac{|T|^2}{\pi_1(q_1)} \tr\left[\Phi_{T\hat{T}|S\hat{S}} \left( \sum_{q_2} \pi_2(q_2) \frac{\id_{\hat{T}}}{|T|} \otimes \rho_{TS\hat{S}}(a_1,q_1)\right)\right] = \frac{|T|}{\pi_1(q_1)} \tr\left[\Phi_{T\hat{T}|S\hat{S}} \left( \id_{\hat{T}} \otimes \rho_{TS\hat{S}}(a_1,q_1)\right)\right],
\end{split}
\end{equation}
which is same as the expression for $\sum_{a_2} p(a_1,a_2|q_1,q_2)$. Similarly, we can also show the non-signalling condition for $a_1$ using the linear constraint on $A_1$ in $\sdp_{1,1}$. Thus, all the states satisfying the linear constraints of $\sdp_{1,1}(V,\pi,T)$ are non-signalling and hence form a smaller set for optimisation than the set of non-signalling correlations.

Next, we show that $\sdp_{1,1}(V,\pi,T) \geq \omega_{\ns}(V,\pi)$. For any non-signalling correlation $p(a_1,a_2|q_1,q_2)$, we can construct the Ansatz state
\begin{align}
\rho_{T\hat{T}S\hat{S}}(a_1,a_2,q_1,q_2) = \pi_1(q_1)\pi_2(q_2)p(a_1,a_2|q_1,q_2)\frac{\id_{T\hat{T}}}{|T|^2}\otimes\rho_{S\hat{S}},
\end{align}
where $\rho_{S\hat{S}}$ is an arbitrary state. We can easily check that this state is a valid feasible state of $\sdp_{1,1}(V,\pi,T)$, i.e. it satisfies all conditions in $\sdp_{1,1}(V,\pi,T)$. The objective function becomes
\begin{equation}
\begin{split}
&\sum_{a_1,a_2,q_1,q_2} V(a_1,a_2,q_1,q_2) \pi_1(q_1)\pi_2(q_2)p(a_1,a_2|q_1,q_2) \tr \left[\Phi_{T\hat{T}|S\hat{S}} \id_{T\hat{T}}\otimes\rho_{S\hat{S}}\right]\\
&=\sum_{a_1,a_2,q_1,q_2} V(a_1,a_2,q_1,q_2) \pi_1(q_1)\pi_2(q_2)p(a_1,a_2|q_1,q_2), 
\end{split}
\end{equation}
which is a success probability of the game with the strategy $p(a_1,a_2|q_1,q_2)$. This implies that the feasible states of $\sdp_{1,1}(V,\pi,T)$ can cover all non-signalling cases, and hence $\sdp_{1,1}(V,\pi,T) \geq \omega_{\ns}(V,\pi)$.
\end{proof}





\section{Rank-one projective measurements}\label{app:rank1ProjMeas}

\subsection{Adapted SDP hierarchies}\label{app:adapted}

For low dimensional examples and corresponding numerics it is important to take into account any additional structure available in order to achieve dimension reductions. In particular, for measurements that are projective rank-one, we can connect our results to the previous work \cite{navascues2014characterization} covering such cases. To exemplify this, we consider here the restricted case $A = \{0,1\}$ with $|T| = 2$.

Let us recall the expression of the optimal success probability $\omega_{Q(T)}(V,\pi)$ after the swap trick
\begin{align}\label{eq:pwinQT_classical}
\omega_{Q(2)}(V,\pi) = \max_{(E,D,\rho)} \sum_{a_1,a_2,q_1,q_2} V(a_1,a_2,q_1,q_2) \tr\left[\Phi_{T\hat{T}|S\hat{S}}\left(E_T(a_1,q_1)\otimes D_{\hat{T}}(a_2,q_2)\otimes\rho_{S\hat{S}}\right)\right].
\end{align}
Recalling that $E_T(a_1,q_1) = \pi_1(q_1) E_T(a_1|q_1)$ and similarly for $D_{\hat{T}}(a_2,q_2)$, as well as exploiting that
$E_T(1|q_1) = \id_T - E_T(0|q_1)$ and $D_{\hat{T}}(1|q_2) = \id_{\hat{T}} - D_{\hat{T}}(0|q_2)$, we have
\begin{align}
\omega_{Q(2)}(V,\pi) = \max_{(E,D,\rho)} \sum_{a_1,a_2,q_1,q_2} V(a_1,a_2,q_1,q_2) \tr\left[\Lambda_{T\hat{T}S\hat{S}}(a_1,a_2)\left(E_T(q_1)\otimes D_{\hat{T}}(q_2)\otimes\rho_{S\hat{S}}\right)\right],
\end{align}
where $E_T(q_1):=E_T(0,q_1)$ and $D_{\hat{T}}(q_2):=D_{\hat{T}}(0,q_2)$, and we defined the matrices
\begin{align}
\Lambda_{T\hat{T}S\hat{S}}(a_1,a_2):= \left(a_1\id_{TS}+(-1)^{a_1}\Phi_{T|S}\right)\otimes\left(a_2\id_{\hat{T}\hat{S}}+(-1)^{a_2}\Phi_{\hat{T}|\hat{S}}\right).
\end{align}
When converting the $Q_1, Q_2$ systems to diagonal quantum systems, we obtain
\begin{align}
\omega_{Q(2)}&(V,\pi) =
\max_{(E,D,\rho)} \sum_{a_1,a_2} \tr\left[\left(V_{Q_1Q_2}(a_1,a_2)\otimes\Lambda_{T\hat{T}S\hat{S}}(a_1,a_2)\right)\left(E_{Q_1T}\otimes D_{Q_2\hat{T}}\otimes\rho_{S\hat{S}}\right)\right] \\
\text{s.t.}\quad &\rho_{S\hat{S}} \geq 0\, , \quad \tr\left[\rho_{S\hat{S}}\right] = 1\\
& E_{Q_1T} = \sum_{q_1}\pi_1(q_1)\ket{q_1}\!\!\bra{q_1}_{Q_1}\otimes E_T(0|q_1)\, , \quad \tr_{T}\left[E_{Q_1T}\right] = \sum_{q_1}\pi_1(q_1)\ket{q_1}\!\!\bra{q_1}_{Q_1} \\
& D_{Q_2\hat{T}} = \sum_{q_2} \pi_2(q_2)\ket{q_2}\!\!\bra{q_2}_{Q_2}\otimes D_{\hat{T}}(0|q_2)\, , \quad \tr_{\hat{T}}\left[D_{Q_2\hat{T}}\right] = \sum_{q_2} \pi_2(q_2) \ket{q_2}\!\!\bra{q_2}_{Q_2},
\end{align}
where $V_{Q_1Q_2}(a_1,a_2):=\sum_{q_1,q_2}V(a_1,a_2,q_1,q_2)\ket{q_1,q_2}\!\!\bra{q_1,q_2}_{Q_1,Q_2}$. This then motivates the SDP hierarchy
\begin{align}
\overline{\sdp}_{n_1,n_2}^{\operatorname{proj}}&(V,\pi,2) := \max_{\rho} \sum_{a_1,a_2} \tr\left[\left(V_{Q_1Q_2}(a_1,a_2)\otimes\Lambda_{T\hat{T}S\hat{S}}(a_1,a_2)\right)\rho_{Q_1TQ_2\hat{T}S\hat{S}}\right] \\
\text{s.t.}\quad &\rho_{(Q_1T)^{n_1}(Q_2\hat{T})^{n_2}S\hat{S}} \geq 0\, , \quad \tr\left[\rho_{(Q_1T)^{n_1}(Q_2\hat{T})^{n_2}S\hat{S}}\right] = 1 \\
& \rho_{(Q_1T)^{n_1}(Q_2\hat{T})^{n_2}S\hat{S}} \text{ perm. inv. on } (Q_1T)^{n_1} \text{ w.r.t. } (Q_2\hat{T})^{n_2}S\hat{S} \\
& \rho_{(Q_1T)^{n_1}(Q_2\hat{T})^{n_2}S\hat{S}} \text{ perm. inv. on } (Q_2\hat{T})^{n_2} \text{ w.r.t. } (Q_1T)^{n_1}S\hat{S} \\
&\tr_T\left[\rho_{(Q_1T)^{n_1}(Q_2\hat{T})^{n_2}S\hat{S}}\right] = \sum_{q_1}\pi(q_1)\ket{q_1}\!\!\bra{q_1}_{Q_1}\otimes\tr_{Q_1T}\left[\rho_{(Q_1T)^{n_1}(Q_2\hat{T})^{n_2}S\hat{S}}\right]\\
&\tr_{\hat{T}}\left[\rho_{(Q_1T)^{n_1}(Q_2\hat{T})^{n_2}S\hat{S}}\right] = \sum_{q_2}\pi(q_2)\ket{q_2}\!\!\bra{q_2}_{Q_2}\otimes\tr_{Q_2\hat{T}}\left[\rho_{(Q_1T)^{n_1}(Q_2\hat{T})^{n_2}S\hat{S}}\right] \, \\
& \rho^{T_{(Q_1T)^{n_1}}}_{(Q_1T)^{n_1}(Q_2\hat{T})^{n_2}S\hat{S}} \geq 0\, , \quad \rho^{T_{(Q_2\hat{T})^{n_2}}}_{(Q_1T)^{n_1}(Q_2\hat{T})^{n_2}S\hat{S}} \geq 0\, , \quad \rho^{T_{S\hat{S}}}_{(Q_1T)^{n_1}(Q_2\hat{T})^{n_2}S\hat{S}}\geq 0,\,\ldots\;.
\end{align}
We note that we do no longer have the $|T|^2$ pre-factor and that the size of the optimisation variable is smaller compared to $\overline{\sdp}_{n_1,n_2}(V,\pi,T)$ in \eqref{eq:obj_funct_original}. This allows for more efficient numerical implementations.

In fact, again writing out the block-diagonal structure 
\begin{align}
\rho_{(Q_1T)^{n_1}(Q_2\hat{T})^{n_2}(S\hat{S})}=:\sum_{q_1^{n_1},q_2^{n_2}} \ket{q_1^{n_1},q_2^{n_2}}\!\!\bra{q_1^{n_1},q_2^{n_2}}\otimes\rho_{T^{n_1}\hat{T}^{n_2}S\hat{S}}\big(q_1^{n_1},q_2^{n_2}\big),
\end{align}
and taking $n_1 = n_2 = |Q|$, one can check that the objective function can be expressed in terms of the single renormalised block
\begin{align}
W_{T^{|Q|}\hat{T}^{|Q|}S\hat{S}}
:=
\frac{\rho_{T^{|Q|}\hat{T}^{|Q|}S\hat{S}}(1,2,\cdots,|Q|,1,2,\cdots,|Q|)}
{\pi_1(1)\pi_1(2)\ldots\pi_1(|Q|)\pi_2(1)\pi_2(2)\ldots\pi_2(|Q|)},
\end{align}
where $\pi_1(q_1) \pi_2(q_2) = \pi(q_1,q_2)$. Then, ignoring all other blocks and only enforcing the positivity, normalisation, and PPT constraints gives for $|Q|$ questions on each side and, e.g., the uniform distribution $\pi_{\text{unif}} = \frac{1}{|Q|^2}$ for the questions, the outer relaxation
\begin{align}
\overline{\sdp}_{|Q|}^{\operatorname{proj}}\left(V,\pi_{\text{unif}},2\right)\leq\frac{1}{|Q|^2}\cdot\sdp_{PPT}(V)
\end{align}
with the SDP from \cite[Equation (6)]{navascues2014characterization}
\begin{align}
\sdp_{PPT}(V)&:=\max_W \sum_{a_1,a_2,q_1,q_2} V(a_1,a_2,q_1,q_2)
\tr\left[\Lambda_{T\hat{T}S\hat{S}}(a_1,a_2) W_{T_{q_1}\hat{T}_{q_2}S\hat{S}}\right] \\
\text{s.t.}\quad &W_{T^{|Q|}\hat{T}^{|Q|}S\hat{S}} \geq 0\, , \quad \tr\left[W_{T^{|Q|}\hat{T}^{|Q|}S\hat{S}}\right] = 1 \\
& W_{T^{|Q|}\hat{T}^{|Q|}S\hat{S}}^{T_{T^{|Q|}}}\geq 0,\,W_{T^{|Q|}\hat{T}^{|Q|}S\hat{S}}^{T_{\hat{T}^{|Q|}}}\geq 0,\,\ldots\;\text{(PPT constraints),}
\end{align}
where the $|Q|^{-2}$ coefficient arises as we are considering a joint probability distribution $p(a_1,a_2,q_1,q_2)$ rather than a conditional one $p(a_1,a_2|q_1,q_2)$. Hence, in this scenario for $|Q|$ questions on each side, the $|Q|$-th level of our hierarchy is never a worse upper bound than the previously studied $\sdp_{PPT}(V)$.


\subsection{Asymptotic convergence analysis}\label{app:asymptotic}

Even in this specific setting with $|A| = |T| = 2$, an advantage of our techniques compared to the previous $\sdp_{PPT}(V)$ and its extendibility extensions discussed in \cite{navascues2014characterization}, is that for a slight variation we can give an approximation scheme scaling polynomially in $|Q|$. Namely, when extending as $(Q_1T)(Q_2\hat{T})^{n_1}(S\hat{S})^{n_2}$ instead of $(Q_1T)^{n_1}(Q_2\hat{T})^{n_2}(S\hat{S})$, we get for the resulting SDP hierarchy denoted by $\sdp_{n_1,n_2}^{\operatorname{proj}}(V,\pi,2)$ that
\begin{align}
\Big|&\sdp_{n_1,n_2}^{\operatorname{proj}}(V,\pi,2) - \omega_{Q(2)}(V,\pi)\Big|\leq128\cdot\sqrt{\frac{ 17}{n_1} + \frac{1}{n_2}}.
\end{align}
The asymptotic convergence analysis is done similarly as for the general $\sdp_{n_1,n_2}(V,\pi,T)$ in the main text.


\subsection{Adding constraints to the general SDP hierarchy}

So far, we have discussed how the assumption of rank-one projective measurements can be used to rewrite
the SDP relaxations of the main text. The variable of the obtained program is smaller than the one in the original
relaxation, thus allowing for a more efficient numerical implementation. Alternatively, we can also enforce the rank-one projective assumption directly into our SDP relaxations, by suitably strengthening the non-signalling linear constraints.

In the following, we introduce additional constraints for a variation of the SDP hierarchy $\sdp_{n_1,n_2}(V,\pi,T)$, which follow from the rank-one projective assumption. As we did in Appendix \ref{app:adapted}, instead of extending the subsystem associated with the quantum assistance $S\hat{S}$, we extend $A_1Q_1T$ and $A_2Q_2\hat{T}$ and explicitly make use of the fact that $A_1Q_1$ and $A_2Q_2$ are classical registers. That is, the program's variable can be reduced into diagonal blocks
\begin{align}
&\rho_{(A_1Q_1T)^{n_1}(A_2Q_2\hat{T})^{n_2}S\hat{S}}\\
&=\sum_{a_1^{n_1},a_2^{n_2},q_1^{n_1},q_2^{n_2}}\, \ket{a_1^{n_1},a_2^{n_2},q_1^{n_1},q_2^{n_2}}\bra{a_1^{n_1},a_2^{n_2},q_1^{n_1},q_2^{n_2}}_{(A_1Q_1)^{n_1}(A_2Q_2)^{n_2}}\otimes\rho_{T^{n_1}\hat{T}^{n_2}S\hat{S}}
\left(a_1^{n_1},a_2^{n_2},q_1^{n_1},q_2^{n_2}\right).
\end{align}
Now, when Alice and Bob can only perform measurements composed by rank-$1$ projectors, we can replace the non-signalling linear constraints with the following strengthened conditions. For a fixed value of the indices $(a_1^{n_1},a_2^{n_2},q_1^{n_1},
q_2^{n_2})$, the new constraints are
\begin{align}\label{eqs:stronger_NS}
&\forall\,\tilde{a}^{(1)}_1, \tilde{q}^{(1)}_1:\quad\sum_{a^{(1)}_1} \,\rho_{T^{n_1}\hat{T}^{n_2}S\hat{S}}
\left(
a^{(1)}_1,\ldots,a^{(n_1)}_1,
a^{(1)}_2,\ldots,a^{(n_2)}_2,
q^{(1)}_1,\ldots,q^{(n_1)}_1,
q^{(1)}_2,\ldots,q^{(n_2)}_2
\right)= \\
&\frac{\pi(q^{(1)}_1)}{\pi(\tilde{q}^{(1)}_1)} \, \mathds{1}_{T_1}\otimes\Trp{T_1}{\rho_{T^{n_1}\hat{T}^{n_2}S\hat{S}}\left(\tilde{a}^{(1)}_1,a^{(2)}_1,\ldots,a^{(n_1)}_1,a^{(1)}_2,\ldots,a^{(n_2)}_2,\tilde{q}^{(1)}_1,q^{(2)}_1,\ldots,q^{(n_1)}_1,q^{(1)}_2,\ldots,q^{(n_2)}_2\right)}
\end{align}
as well as
\begin{align}
&\forall \, \tilde{a}^{(1)}_2, \tilde{q}^{(1)}_2:\quad\sum_{a^{(1)}_2} \, \rho_{T^{n_1}\hat{T}^{n_2}S\hat{S}}
\left(
a^{(1)}_1,\ldots,a^{(n_1)}_1,
a^{(1)}_2,\ldots,a^{(n_2)}_2,
q^{(1)}_1,\ldots,q^{(n_1)}_1,
q^{(1)}_2,\ldots,q^{(n_2)}_2
\right)=\\
&\frac{\pi(q^{(1)}_2)}{\pi(\tilde{q}^{(1)}_2)} \, \mathds{1}_{\hat{T}_1}\otimes\Trp{\hat{T}_1}{\rho_{T^{n_1}\hat{T}^{n_2}S\hat{S}}\left(a^{(1)}_1,\ldots,a^{(n_1)}_1,\tilde{a}^{(1)}_2,a^{(2)}_2,\ldots,a^{(n_2)}_2,q^{(1)}_1,\ldots,q^{(n_1)}_1,\tilde{q}^{(1)}_2,q^{(2)}_2,\ldots,q^{(n_2)}_2\right)}.
\end{align}
It is easy to see, using the form of the optimal variable for our SDP relaxations, that the above constraints
are satisfied only when the measurement operators are composed by rank-one projectors. When this is not the
case, we would obtain an unknown normalisation coefficient coming from the trace in the RHS of the above equations. Similarly as in Appendix \ref{app:adapted}, the resulting SDPs lead to bounds at least as good as $\frac{1}{|Q|^2}\cdot\sdp_{PPT}(V)$ from \cite{navascues2014characterization}, when going to the $|Q|$-level for games with $|Q|$ questions on each side. For example, when $|A|=|Q|=|T|=2$ this is achieved by expressing the objective function in terms of the renormalised single block variable
\begin{align}
Z_{T_1T_2\hat{T}_1\hat{T}_2S\hat{S}}:=\frac{|T|^4}{\pi_1(0)\pi_1(1)\pi_2(0)\pi_2(1)}\cdot\rho_{T_1T_2\hat{T}_1\hat{T}_2S\hat{S}}(0,0,0,0,1,2,1,2)
\end{align}
and then forgetting about all other blocks in the constraints.

\end{appendices}


\bibliographystyle{abbrv}
\bibliography{bibliography}

\end{document}

%% file: format.tex
\usepackage[utf8]{inputenc}
\usepackage{amsmath,amssymb,hyperref,array,multicol,verbatim,mathpazo,dsfont,physics,authblk,amsthm,multirow}
\usepackage[a4paper, top=2cm, bottom=2cm, left=2.2cm, right=2.2cm]{geometry}
\usepackage[normalem]{ulem}
\usepackage[numbers,sort&compress]{natbib} 
\usepackage[pdftex]{graphicx}
\usepackage[dvipsnames]{xcolor} 
\usepackage[title]{appendix}
\usepackage{mathtools}
\mathtoolsset{showonlyrefs}

\newtheorem{theorem}{Theorem}
\newtheorem{lemma}[theorem]{Lemma}
\newtheorem{corollary}[theorem]{Corollary}
\newtheorem{definition}{Definition}
\newtheorem{proposition}[theorem]{Proposition}

\newcommand{\id}{\mathds{1}}
\newcommand{\sdp}{\operatorname{sdp}}
\newcommand{\hil}{\mathcal{H}}
\newcommand{\uni}{\mathcal{U}}

\DeclareMathOperator*{\E}{\mathbb{E}}
\newcommand{\SH}[1]{\mathcal{S}\left({#1}\right)}

\newcommand{\Trt}[1]{\tr\left[{#1}\right]}
\newcommand{\Trp}[2]{\tr_{#1}\left[{#2}\right]}

\newcommand{\supp}{\mathrm{supp}}
\newcommand{\npa}{\operatorname{NPA}}

\newcommand{\ns}{\operatorname{NS}}